\documentclass[journal, onecolumn, 12pt, letterpaper, draftclsnofoot]{IEEEtran}
\IEEEoverridecommandlockouts

\makeatletter
\def\endthebibliography{%
	\def\@noitemerr{\@latex@warning{Empty `thebibliography' environment}}%
	\endlist
}
\makeatother

\usepackage{comment}
\usepackage{cite}
\usepackage{amsmath,amssymb,amsfonts}
\usepackage{algorithmic}
\usepackage{graphicx}
\usepackage{textcomp}
\usepackage{xcolor}
\usepackage[ruled,vlined]{algorithm2e}
\usepackage{algorithmic}
\usepackage{makecell}
\usepackage{diagbox}
\usepackage{multirow}

\usepackage{tocloft}%

\usepackage{etoc}

\def\BibTeX{{\rm B\kern-.05em{\sc i\kern-.025em b}\kern-.08em
    T\kern-.1667em\lower.7ex\hbox{E}\kern-.125emX}}

\usepackage{nicefrac}
\usepackage{hyperref}
\hypersetup{
    colorlinks=true,
    linkcolor=blue,
    filecolor=magenta,      
    urlcolor=cyan
}
\usepackage{mathtools}
\usepackage{amsmath,amsfonts}
\usepackage{bbm}

\usepackage{amssymb}
\usepackage[ruled,vlined]{algorithm2e}
\usepackage{algorithmic}
\usepackage{makecell}
\usepackage{diagbox}
\usepackage{multirow}
\usepackage{textgreek}

\usepackage{caption}
\usepackage{subcaption}

\usepackage{amsthm}
\newtheorem{theorem}{Theorem}
\newtheorem{corollary}{Corollary}
\newtheorem{proposition}{Proposition}
\newtheorem{lemma}{Lemma}
\newtheorem{assumption}{\textit{Assumption}}
\newtheorem{mechanism}{\textit{Mechanism}}

\usepackage[nameinlink,capitalize]{cleveref}

\crefname{section}{\S}{\S}
\Crefname{section}{\S}{\S}
\crefname{appendix}{Appendix}{Appendices}
\Crefname{appendix}{Appendix}{Appendices}
\crefname{theorem}{Theorem}{Theorems}
\Crefname{theorem}{Theorem}{Theorems}
\crefname{proposition}{Proposition}{Propositions}
\Crefname{proposition}{Proposition}{Propositions}
\crefname{algorithm}{Algorithm}{Algorithms}
\Crefname{algorithm}{Algorithm}{Algorithms}
\crefname{assumption}{Assumption}{Assumptions}
\Crefname{assumption}{Assumption}{Assumptions}
\crefname{mechanism}{Mechanism}{Mechanisms}
\Crefname{mechanism}{Mechanism}{Mechanisms}

\usepackage{mfirstuc}
\usepackage[T1]{fontenc}

\newcommand\numberthis{\addtocounter{equation}{1}\tag{\theequation}}

\newcommand{\myparatightestn}[1]{\smallskip \noindent\textbf{{#1}}~}

\newcommand{\myparaemphtightestn}[1]{\noindent\emph{{#1}}~}

\newcounter{packednmbr}

\newenvironment{packeditemize}{\begin{list}{$\bullet$}{\setlength{\itemsep}{0.5pt}\addtolength{\labelwidth}{-4pt}\setlength{\leftmargin}{\labelwidth}\addtolength{\leftmargin}{5pt}\setlength{\listparindent}{\parindent}\setlength{\parsep}{1pt}\setlength{\topsep}{0pt}}}{\end{list}}

\newcommand{\revision}[1]{{#1}}

\newcommand{\calD}{\mathcal{D}}

\newcommand{\calI}{\mathcal{I}}

\newcommand{\calL}{\mathcal{L}}
\newcommand{\calM}{\mathcal{M}}
\newcommand{\calN}{\mathcal{N}}
\newcommand{\calO}{\mathcal{O}}

\newcommand{\calS}{\mathcal{S}}

\newcommand{\bra}[1]{\left( #1 \right)}
\newcommand{\brb}[1]{\left[ #1 \right]}
\newcommand{\brab}[1]{\left( #1 \right]}
\newcommand{\brba}[1]{\left[ #1 \right)}
\newcommand{\brc}[1]{\left\{ #1 \right\}}
\newcommand{\brd}[1]{\left| #1 \right|}

\DeclarePairedDelimiter\abs{\lvert}{\rvert}
\newcommand{\babs}[1]{\left| #1 \right|}

\newcommand{\floor}[1]{\lfloor {#1} \rfloor}
\newcommand{\ceil}[1]{\lceil {#1} \rceil}

\newcommand{\wiki}{Wikipedia Web Traffic Dataset}
\newcommand{\wikishort}{WWT}

\newcommand{\fccmba}{Measuring Broadband America Dataset}
\newcommand{\fccmbashort}{MBA}

\newcommand{\attacker}{attacker}

\newcommand{\name}{summary statistic privacy}

\newcommand{\mechanismname}{quantization mechanism}

\newcommand{\MechanismName}{Quantization Mechanism}
\newcommand{\mechanismsname}{quantization mechanisms}

\newcommand{\attackerstrategy}{\attacker{} strategy}

\newcommand{\datamechanism}{data release mechanism}
\newcommand{\DataMechanism}{Data Release Mechanism}
\newcommand{\DataMechanisms}{Data Release Mechanisms}
\newcommand{\Datamechanism}{Data release mechanism}
\newcommand{\datamechanisms}{data release mechanisms}

\newcommand{\distortion}{distortion}
\newcommand{\Distortion}{Distortion}

\newcommand{\privacy}{privacy}
\newcommand{\Privacy}{Privacy}

\newcommand{\secret}{secret}
\newcommand{\Secret}{Secret}

\newcommand{\indicatorof}[1]{\mathbb{I}\bra{#1}}

\newcommand{\privacythreshold}{\epsilon}
\newcommand{\distortionthreshold}{T}
\newcommand{\privacymetricthreshold}{T}
\newcommand{\privacynotation}{\Pi_{\privacythreshold,\paramdistribution}}
\newcommand{\distortionnotation}{\Delta}

\newcommand{\rvprivatenotation}{X}
\newcommand{\rvreleasenotation}{X}
\newcommand{\paramdim}{q}

\newcommand{\rvprivatewithparam}[1]{\rvprivatenotation_{#1}}
\newcommand{\rvreleasewithparam}[1]{\rvreleasenotation_{#1}}
\newcommand{\rvparamnotation}{\theta}
\newcommand{\RVparamnotation}{\Theta}
\newcommand{\releaservparamnotation}{\theta'}
\newcommand{\rvparamupperbound}{\overline{\rvparamnotation}}
\newcommand{\rvparamlowerbound}{\underline{\rvparamnotation}}
\newcommand{\rvparamupperboundof}[1]{\overline{\rvparamnotation^{#1}}}
\newcommand{\rvparamlowerboundof}[1]{\underline{\rvparamnotation^{#1}}}
\newcommand{\rvparamothernotation}{v}
\newcommand{\RVparamothernotation}{V}

\newcommand{\rvprivate}{\rvprivatewithparam{\rvparamnotation}}

\newcommand{\rvmean}{u}
\newcommand{\RVmean}{U}
\newcommand{\rvmeanupperbound}{\overline{\rvmean}}
\newcommand{\rvmeanlowerbound}{\underline{\rvmean}}

\newcommand{\muupperbound}{\overline{\mu}}
\newcommand{\mulowerbound}{\underline{\mu}}

\newcommand{\sigmaupperbound}{\overline{\sigma}}
\newcommand{\sigmalowerbound}{\underline{\sigma}}

\newcommand{\mupperbound}{\overline{m}}
\newcommand{\mlowerbound}{\underline{m}}

\newcommand{\lambdalowerbound}{\underline{\lambda}}
\newcommand{\lambdaupperbound}{\overline{\lambda}}

\newcommand{\hlowerbound}{\underline{h}}
\newcommand{\hupperbound}{\overline{h}}

\newcommand{\numcat}{C}
\newcommand{\rvparamindpnotation}{t^*}

\newcommand{\distortionbudgetnotation}{T}
\newcommand{\distortionbudgetlowerbound}{\underline{B}}
\newcommand{\distortionbudgetupperbound}{\overline{B}}

\newcommand{\binarysearchprecision}{\eta}

\newcommand{\secnum}{N}
\newcommand{\seclen}{s}

\newcommand{\precision}{\kappa}

\newcommand{\distributionnotation}{\omega}
\newcommand{\distributionof}[1]{\distributionnotation_{#1}}
\newcommand{\privatedistribution}{\distributionof{\rvprivatewithparam{\rvparamnotation}}}
\newcommand{\releasedistribution}{\distributionof{\rvreleasewithparam{\releaservparamnotation}}}

\newcommand{\paramdistribution}{\distributionof{\RVparamnotation}}

\newcommand{\pdfnotation}{f}
\newcommand{\pdfof}[1]{\pdfnotation_{#1}}

\newcommand{\privatepdf}{\pdfof{\rvprivatewithparam{\rvparamnotation}}}

\newcommand{\cdfof}[1]{F_{#1}}

\newcommand{\privatedataset}{\mathcal{X}}
\newcommand{\releasedataset}{\mathcal{Y}}

\newcommand{\secretnotation}{g}
\newcommand{\secretof}[1]{\secretnotation\bra{#1}}

\newcommand{\secretofparam}{\secretof{\rvparamnotation}}

\newcommand{\secretestimatenotation}{\hat{\secretnotation}}
\newcommand{\secretestimatestarnotation}{\hat{\secretnotation}^*}
\newcommand{\secretestimateof}[1]{\secretestimatenotation\bra{#1}}
\newcommand{\secretestimatestarof}[1]{\secretestimatestarnotation\bra{#1}}
\newcommand{\secretestimatesubscriptof}[2]{\secretestimatenotation_{#1}\bra{#2}}
\newcommand{\secretestimate}{\secretestimateof{\releaservparamnotation}}

\newcommand{\mechanismnoisenotation}{z}
\newcommand{\RVmechanismnoisenotation}{Z}
\newcommand{\distributionofmechanismnoise}{\distributionof{\RVmechanismnoisenotation}}

\newcommand{\mechanismnotation}{\calM_{\secretnotation}}
\newcommand{\mechanismofwithnoise}[2]{\mechanismnotation\bra{#1, #2}}
\newcommand{\mechanismof}[1]{\mechanismofwithnoise{#1}{ \mechanismnoisenotation}}

\newcommand{\setofprivateparam}{\support{\RVparamnotation}}

\newcommand{\subsetofprivateparamof}[1]{\calS_{#1}}

\newcommand{\privateparamindexsetnotation}{\calI}

\newcommand{\privateparamindexnotation}{i}

\newcommand{\paramsetindexof}[1]{I\bra{#1}}

\newcommand{\setofintegers}{\mathbb{Z}}
\newcommand{\setofpositiveintegers}{\mathbb{Z}_{>0}}
\newcommand{\setofnaturalnumbers}{\mathbb{N}}
\newcommand{\setofpositivereal}{\mathbb{R}_{>0}}
\newcommand{\setofreal}{\mathbb{R}}

\newcommand{\releaseparamofindex}[1]{\rvparamnotation^*_{#1}}
\newcommand{\releasepdfofindex}[1]{\rvparamnotation'_{#1}}
\newcommand{\releasedistributionofindex}[1]{\distributionof{\rvreleasewithparam{\releaseparamofindex{#1}}}}

\newcommand{\dpf}[1]{pri\bra{#1}}

\newcommand{\normaldistribution}[2]{\calN\bra{#1,#2}}

\newcommand{\laplace}[2]{\text{Laplace}\bra{#1,#2}}

\newcommand{\uniformdistributionnotation}{\text{U}}

\newcommand{\uniformdistributionvar}[2]{\uniformdistributionnotation\bra{\brb{#1, #2}}}

\newcommand{\wassersteinof}[2]{d_{\text{Wasserstein-1}}\bra{#1\|#2}}

\newcommand{\TVof}[2]{d_{\text{TV}}\bra{#1\|#2}}

\newcommand{\distanceof}[2]{d\bra{#1\|#2}}

\newcommand{\ratio}{\gamma}

\newcommand{\probnotation}{\mathbb{P}}
\newcommand{\probof}[1]{\probnotation\bra{#1}}

\newcommand{\expectationnotation}{\mathbb{E}}
\newcommand{\expectationof}[1]{\expectationnotation\bra{#1}}

\newcommand{\auxdistance}[2]{D\bra{#1, #2}}
\newcommand{\auxrange}[2]{R\bra{#1, #2}}

\newcommand{\support}[1]{\text{Supp}\bra{#1}}

\newcommand{\distanceformula}[2]{\frac{1}{2} d\bra{\distributionof{#1}\|\distributionof{#2}}}
\newcommand{\distanceformulawass}[2]{\frac{1}{2} \wassersteinof{\distributionof{#1}}{\distributionof{#2}}}
\newcommand{\distanceformulaTV}[2]{\frac{1}{2} \TVof{\distributionof{#1}}{\distributionof{#2}}}

\newcommand{\rangeformula}[2]{\abs{\secretnotation{({#1})-\secretnotation{({#2})}}}}

\newcommand{\functionrange}[1]{range\bra{#1}}

\newcommand{\fraction}{j}

\newcommand{\dpmechanismnotation}{\calM}
\newcommand{\dpmechanismof}[1]{\dpmechanismnotation\bra{#1}}

\newcommand{\mutualinformationof}[2]{I(#1;#2)}

\newcommand{\timecomplexitynotation}{\mathcal{C}}

\crefformat{section}{§#2#1#3}

\AtBeginEnvironment{appendices}{\crefalias{section}{appendix}}
\AtBeginEnvironment{appendices}{\crefalias{subsection}{appendix}}
\AtBeginEnvironment{appendices}{\crefalias{subsubsection}{appendix}}

\makeatletter
\newcommand{\linebreakand}{%
\end{@IEEEauthorhalign}
\hfill\mbox{}\par
\mbox{}\hfill\begin{@IEEEauthorhalign}}
\makeatother

\title{Summary Statistic Privacy  in  Data Sharing
}

\author{
Zinan Lin$^*$, Shuaiqi Wang$^*$, Vyas Sekar, Giulia Fanti
\thanks{$^*$ These authors contributed equally to this work.}
\thanks{Zinan Lin is with Microsoft Research, Redmond, WA, 98052 USA, e-mail:
zinanlin@microsoft.com.}
\thanks{Shuaiqi Wang is with Carnegie Mellon University, Pittsburgh, PA, 15213 USA, e-mail:
shuaiqiw@andrew.cmu.edu.}
\thanks{Vyas Sekar is with Carnegie Mellon University, Pittsburgh, PA, 15213 USA, e-mail:
vsekar@andrew.cmu.edu.}
\thanks{Giulia Fanti is with Carnegie Mellon University, Pittsburgh, PA, 15213 USA, e-mail:
gfanti@andrew.cmu.edu.}
}

\begin{document}

\pagenumbering{gobble}

\maketitle

\begin{abstract}
We study a setting where a data holder wishes to share data with a receiver, \emph{without} revealing certain summary statistics of the data distribution (e.g., mean, standard deviation). 
It achieves this by passing the data through a randomization mechanism. 
We propose \emph{\name{}}, a metric for quantifying the privacy risk of such a mechanism based on the worst-case probability of an adversary guessing the distributional secret within some threshold. 
Defining distortion as a worst-case %
Wasserstein-1 distance between the real and released data, we prove lower bounds on the tradeoff between privacy and distortion. 
We then propose a class of quantization mechanisms that can be adapted to different data distributions. 
We show that the quantization mechanism's privacy-distortion tradeoff matches our lower bounds under certain regimes, up to small constant factors. 
Finally, we  demonstrate on real-world datasets that the proposed quantization mechanisms achieve better \privacy{}-\distortion{} tradeoffs than alternative privacy mechanisms.
\end{abstract}

\newpage
\tableofcontents

\newpage
\pagenumbering{arabic} 
\section{Introduction}
\label{sec:intro}

Data sharing  is an important enabler for data-driven product development \cite{lee2000information}, coordination efforts (e.g., cybersecurity \cite{choucri2016institutions}, law enforcement \cite{jacobs2008sharing}), and the creation of benchmarks for evaluating scientific progress \cite{deng2009imagenet,reiss2011google,luo2021characterizing}.  
However, %
\emph{summary statistics} of shared data may leak sensitive information \cite{suri2021formalizing,suridissecting}.
For example, \emph{property inference} attacks allow an attacker to infer properties about the individuals in the training dataset of a released machine learning model \cite{ateniese2015hacking,ganju2018property,zhang2021leakage,mahloujifar2022property,chaudhari2022snap}. 
{An institution that shares DNS data may not want to disclose even aggregated queries, as these quantities can be used to infer details about the institution \cite{imana2021institutional}.}
A cloud provider that shares cluster performance traces may not want to reveal the proportions of different server types that the cloud provider owns, which are regarded as business secrets \cite{lin2020using}. 
Note that this information ({aggregate DNS queries}, proportions of server types) cannot be inferred from any record, but is a property of the data distribution (or the aggregate dataset).

Our setup is as follows (detailed formulation in \S\ref{sec:framework}). A data holder possesses a data distribution. The data holder chooses one or more secrets, which are  defined as deterministic functions of the distribution. For example, a video analytics company might choose the mean daily observed traffic as a secret quantity. 
Then, the data holder obfuscates their data distribution according to a randomization \emph{mechanism} and releases the output (\cref{fig:overview}). 
The goal is to prevent an adversary from estimating the value of the secrets, while preserving data utility. 

\begin{figure}[htbp]
    \centering
    \includegraphics[width=0.8\linewidth]{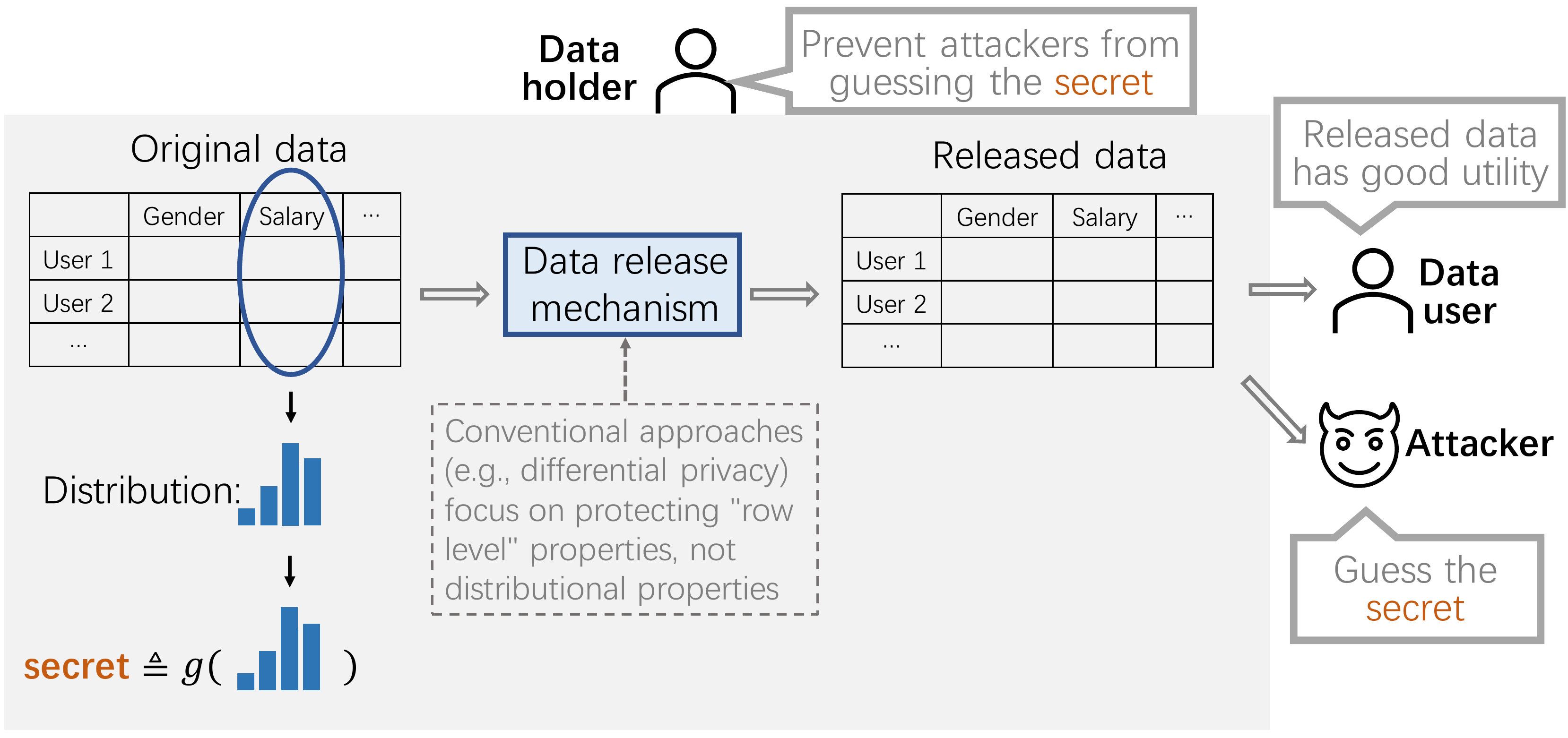}
    \vspace{-0.2cm}
    \caption{Problem overview. 
    The data holder produces released data and wants to hide \emph{statistical secrets} of the original data. 
    The attacker (could be the data user) observes the released data, and wants to guess the \emph{secrets} of the original data. 
    We focus on secrets about the  \emph{underlying distribution} (e.g., mean, quantile, standard deviation of a specific data \emph{column}). Many existing frameworks (e.g., differential privacy \cite{dwork2014algorithmic}) protect information from \emph{individual samples (rows)}.
    }
    \label{fig:overview}
    \label{fig:difference_to_prior}
    \label{fig:toolbox}
\end{figure}

Many widely-used privacy metrics and data sharing algorithms are not designed to protect \name{}, %
instead protecting the privacy of individual records in a database (e.g., differential privacy \cite{dwork2014algorithmic}, anonymization \cite{reiss2012obfuscatory}, sub-sampling \cite{reiss2012obfuscatory}). %
For example, differential privacy (DP) \cite{dwork2014algorithmic} evaluates how much individual samples influence the final output of an algorithm, and does not inherently protect summary statistics 
\cite{ateniese2015hacking}. 

Many other frameworks have been designed specifically to hide aggregate properties of a dataset (or a distribution) \cite{zhang2022attribute,makhdoumi2014information,issa2019operational}; we discuss these in detail in Section \ref{sec:related}.
{Many of these frameworks define privacy in terms of information-theoretic quantities such as mutual information \cite{makhdoumi2014information} or other divergences \cite{wang2019privacy}.}
{In this work, we directly define the \emph{privacy} of a mechanism as the posterior probability that a worst-case attacker can infer the data holder's true secret after observing the released data. 
This definition is related to prior work analyzing min-entropy as a privacy metric \cite{asoodeh2017privacy}.} 
To capture the utility of released data, 
we define the \emph{distortion} of a mechanism as the worst-case {Wasserstein-1} %
distance %
between the original and released data distributions.
Our goal is to design \datamechanisms{} that efficiently trade off  privacy and distortion (defined in \ref{sec:model}).

\subsection{Contributions}
Our contributions are as follows. 
\begin{packeditemize}
    \item \textbf{Lower bounds (\cref{sec:general_lower_bound}):} 
    We derive general  lower bounds on \distortion{} given a \privacy{} budget for any mechanism.
    These bounds depend on both the secret function and the data distribution. 
    We derive closed-form lower bounds for a number of case studies (i.e., combinations of prior beliefs on the data distribution and secret functions).  
    \item \textbf{Mechanism design and upper bounds (\cref{sec:general_upper_bound}):} 
    We propose a class of mechanisms that achieve \name{} called \emph{quantization mechanisms}, which intuitively quantize a data distribution's parameters\footnote{We assume data distributions are drawn from a parametric family; more details in \cref{sec:model}.}  into bins.
    We show that for the case studies analyzed theoretically in Table \ref{tbl:summary}, the quantization mechanism achieves a privacy-distortion tradeoff within a small constant factor of optimal (usually $\leq$3) in the regime where quantization bins are small relative to the overall support set of the distribution parameters.
    We present a \emph{sawtooth technique} for theoretically analyzing  the quantization mechanism's \privacy{} tradeoff under various types of secret functions and data distributions (\cref{sec:sawtooth}).
    Intuitively, the sawtooth technique exploits the geometry of the distribution parameter(s) to divide the parametric space into two regions: one in which privacy risk is small and analytically tractable, and another in which privacy risk can be high, but which occurs with low probability. 
    For the case studies that we do not analyze theoretically, we provide a dynamic programming algorithm that efficiently numerically instantiates the quantization mechanism.

    \item \textbf{Empirical evaluation (\cref{sec:experiments}):} We give empirical results showing how to use \name{} to release a real dataset, and how to evaluate the corresponding \name{} metric. 
    We show that the proposed quantization mechanism achieves better \privacy{}-\distortion{} tradeoffs than other related privacy mechanisms.

\end{packeditemize}

\section{Related Work }
\label{sec:motivation}
\label{sec:related}

\begin{figure}[t]
    \centering
    \includegraphics[width=0.6\linewidth]{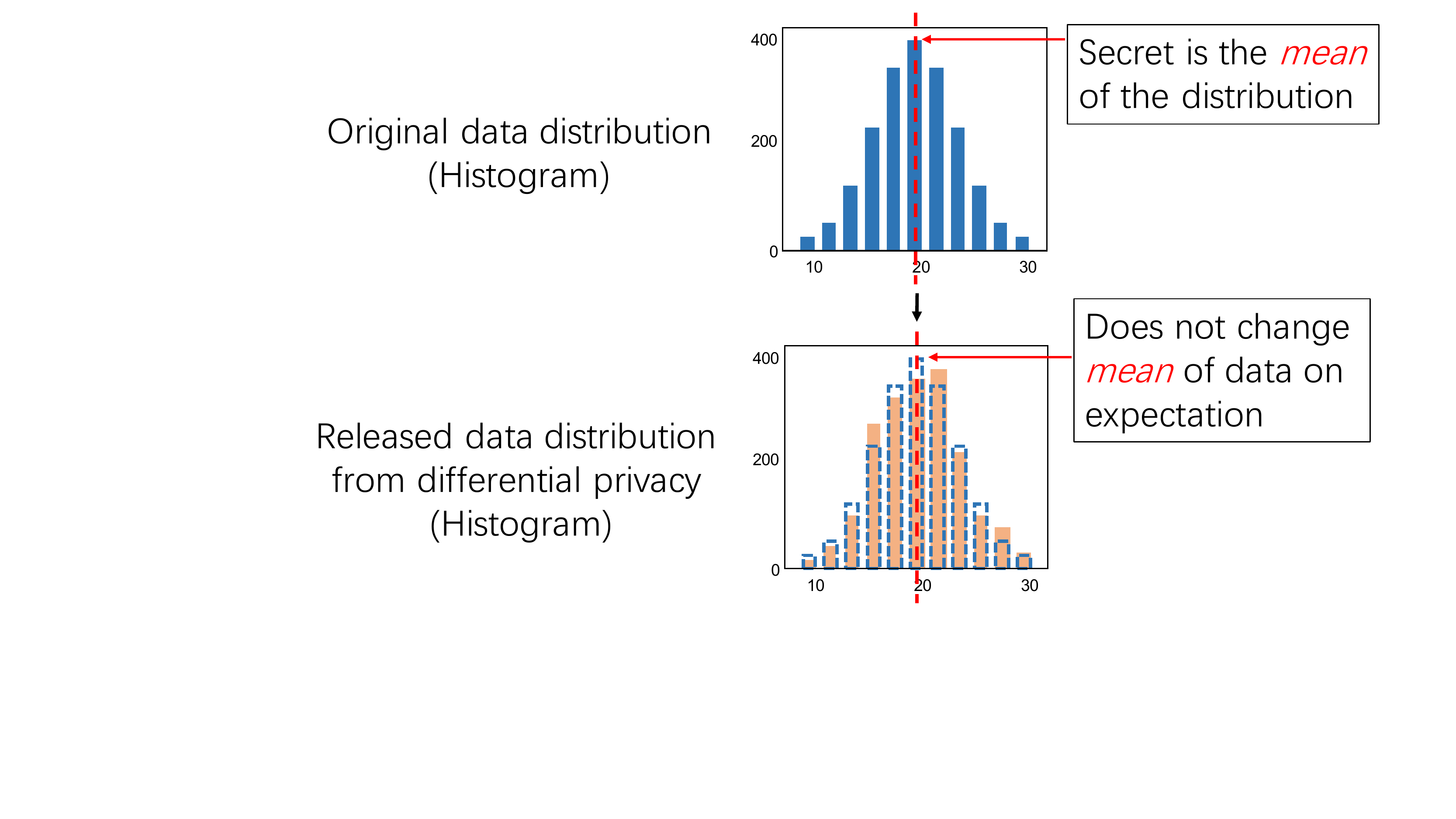}
    \caption{An illustrative example of why naive differential privacy mechanisms do not protect summary statistics. %
    Assume that we want to protect the \emph{mean} of the data. A typical differential privacy algorithm \cite{wasserman2010statistical} would add zero-mean noise (e.g., Laplace noise) to the bins. %
    This mechanism does not change the expected mean of the data. 
    }
    \label{fig:prior_work}
\end{figure}

We divide the related work into two categories: approaches based on indistinguishability over candidate inputs, and information-theoretic approaches.

\subsection{Indistinguishability-Based Approaches}
    \noindent \textbf{Differential privacy (DP)} \cite{dwork2014algorithmic} is one of the most commonly-adopted privacy frameworks. 
    A random mechanism $\dpmechanismnotation{}$ is $(\epsilon, \delta)$-differentially-private if for any neighboring  datasets   $\mathcal{X}_0$ and $\mathcal{X}_1$ (i.e., $\mathcal{X}_0$ and $\mathcal{X}_1$ differ one sample), and any set $S\subseteq \functionrange{\dpmechanismnotation{}}$, we have
    \begin{align*}
        \probof{\dpmechanismof{\mathcal{X}_0} \in S} \leq e^\epsilon\cdot \probof{\dpmechanismof{\mathcal{X}_1} \in S} + \delta~~.
    \end{align*}

    One could try to apply DP to our problem by treating $\dpmechanismnotation{}$ as the data release mechanism that reads the original dataset and outputs the released dataset. 
    However, the threat models of DP and our framework are different: we want to hide functions of \emph{a distribution}, while DP aims to hide whether \emph{any given sample} contributed to the shared data.
    For example, when releasing the data in \cref{fig:prior_work} as a histogram while preserving its mean,
    a typical DP algorithm \cite{wasserman2010statistical} adds zero-mean noise (e.g., Laplace noise) to the bins. 
    This process maintains the expected mean of the data, still allowing the attacker to derive an unbiased estimator of the mean from the released data.
    
    A natural alternative is to devise a DP-like definition that explicitly protects the secret quantity. %
    For instance, we could ask that for any pair of input distributions that differ in their secret quantity, the data release mechanism outputs similar released data distributions. 
    Such an approach provides strong privacy guarantees, but may have poor utility. 
    For instance, consider two Gaussian input distributions  $\normaldistribution{\mu_1}{\sigma_1^2}$ and $\normaldistribution{\mu_2}{\sigma_2^2}$ with the secret as the mean. 
    The values of  $\sigma_1^2$ and $\sigma_2^2$ could be arbitrarily different. 
    To make input distributions indistinguishable given the released data, 
    we must destroy information about the true $\sigma$, which requires adding potentially unbounded  noise. 
    While relaxations like metric differential privacy may help \cite{chatzikokolakis2013broadening}, they may introduce new challenges, e.g., how to choose the metric function to map dataset distance to a privacy parameter.

    \noindent \textbf{Attribute privacy} \cite{zhang2022attribute} tackles these challenges in part by constraining the space of distributions that should be indistinguishable  \cite{zhang2021leakage}. 
    Attribute privacy protects a function of a sensitive column in the dataset (named \emph{dataset attribute privacy}) or a sensitive parameter of the underlying distribution from which the data is sampled (named \emph{distribution attribute privacy}). It addresses the previously-mentioned shortcomings of vanilla DP under the \emph{pufferfish privacy framework} \cite{kifer2014pufferfish}. %
    Precisely, let $\mathcal{X}$ be the dataset, $\mathcal{G}$ be the possible range of a secret $g$, and $\mathcal{G}_a, \mathcal{G}_b\subseteq\mathcal{G}$ be two non-overlapping subsets of the secret range $\mathcal{G}$. A mechanism $\mathcal{M}$ is $(\epsilon,\delta)$-attribute private if for any dataset $\mathcal{X}$, secret range pairs $\mathcal{G}_a, \mathcal{G}_b$, and any set $S\subseteq range(\mathcal{M})$:
    \begin{align*}
        \probof{\mathcal{M}(\mathcal{X})\in S|g(\mathcal{X})\in \mathcal{G}_a} \leq e^{\epsilon} \probof{\mathcal{M}(\mathcal{X})\in S|g(\mathcal{X})\in \mathcal{G}_b} + \delta.
    \end{align*}
    Attribute privacy focuses on algorithms that output \emph{a statistical query of the dataset} instead of the entire dataset.
    Though we may apply attribute privacy to analyze full-dataset-sharing algorithms; it may need to add substantial noise due to the high dimensionality of the dataset (\cref{sec:experiments}). 

    \noindent \textbf{Distribution privacy} \cite{kawamoto2019local} is a closely related notion, which releases a full data distribution under DP-style indistinguishability guarantees. Roughly, for any two input distributions with parameters $\theta_0$ and $\theta_1$ from a pre-defined set of candidate distributions, a distribution private mechanism outputs a distribution $\mathcal M(\theta_i)$ for $i\in \{0,1\}$ such that for any set $S$ in the output space, we have $\mathbb P[\mathcal M(\theta_i)\in S] \leq e^\epsilon \mathbb P[\mathcal M(\theta_{1-i})\in S]+\delta$. 
    By obfuscating the whole distribution, distribution privacy inherently protects the private information.
    However the required noise may be more than what is needed  to protect only select secret(s). 
    For example, as mentioned above, two datasets can have \emph{exactly the same} secret statistic (e.g., mean), while differing significantly in other respects (e.g., variance)---this requires significant noise in general. 
    A recent work \cite{chen2022protecting} proposes mechanisms for distribution privacy, and  we observe this trend experimentally in \cref{sec:experiments}; the noise added by the mechanisms in \cite{chen2022protecting} is larger than what we require with \name{} (though the privacy guarantees are different, so it is difficult to do a fair comparison). %

    \noindent \textbf{Distribution inference} \cite{suri2021formalizing,suridissecting}
    considers a hypothesis test in which the adversary must choose whether released data comes from one of two fixed input data distributions $\omega_1, \omega_2$. Both distributions are assumed to be known to all parties. 
    By defining the attacker's guessed distribution as $\hat{\omega}$ and attacker's advantage as $
    \left|{\probof{\hat{\omega}|\omega_1}-\probof{\hat{\omega}|\omega_2}}\right|$, distribution inference aims to ensure that the attacker's advantage is negligible.
    However, it is unclear how to establish a reasonable pair of candidate distributions; moreover, as with distribution privacy and attribute privacy, distribution inference may require high noise since it requires the data distributions to be indistinguishable.

    \subsection{Information-Theoretic Approaches}
    The second category of frameworks use information-theoretic measures of privacy and utility
    \cite{yamamoto1983source,smith2009foundations,alvim2014additive,calmon2015fundamental,issa2019operational,asoodeh2016information,asoodeh2017privacy,liao2019tunable,calmon2015fundamental,saeidian2022pointwise,kurri2022operational,gilani2023alpha}.
    Such works often measure disclosure via divergences, such as mutual information \cite{makhdoumi2014information,calmon2015fundamental,rassouli2021perfect,zamani2022bounds}, $f$-divergences \cite{wang2019privacy,rassouli2019optimal}, or min-entropy 
    \cite{smith2009foundations,alvim2012measuring,alvim2014additive,asoodeh2017privacy,asoodeh2018estimation}.
    We discuss a few examples in detail here.
    
        \noindent \textbf{Privacy funnel} \cite{makhdoumi2014information} is a well-known information-theoretic privacy framework. 
        Let $X$ be the random variable of the original data, containing sensitive information $U$, and let $Y$ represent the (random) released data. The privacy funnel framework evaluates privacy leakage with the mutual information $\mutualinformationof{U}{Y}$, and the utility of $Y$ with mutual information $\mutualinformationof{X}{Y}$. To find a data release mechanism $P_{Y|X}$, privacy funnel solves the optimization
        $
        \min_{P_{Y|X}:\mutualinformationof{X}{Y}\geq R} \mutualinformationof{U}{Y},
        $
    where $R$ is a desired threshold on the utility of $Y$.
    Prior work has argued that mutual information is not a good metric for either privacy or utility \cite{issa2019operational}. On the privacy front, $\mutualinformationof{U}{Y}$ can be reduced while allowing the attacker to guess $S$ correctly from $Y$ with higher probability  (see Example 1 in \cite{issa2019operational}). On the utility front, high mutual information $\mutualinformationof{X}{Y}$ does not mean that the released data $Y$ is a useful representation of $X$; $Y$ could be an arbitrary one-to-one transformation of $X$.

    \noindent \textbf{Maximal leakage} \cite{issa2019operational} is an information-theoretic framework for quantifying the leakage of sensitive information. 
    Using the same notation as before, the adversary's guess of secret $U$ is denoted by $\hat{U}$. Based on this setup, the Markov chain $U-X-Y-\hat{U}$ holds. Maximal leakage $\calL$ from $X$ to $Y$ is defined as 
    \begin{align}
        \calL{}\bra{X\to Y} = \sup_{U-X-Y-\hat{U}} \log \frac{\probof{U=\hat{U}}}{ \max_{u}P_U(u) }, \label{eq:maximal_leakage}
    \end{align}
    where the $\sup$ is taken over $U$ (i.e., considering the worst-case secret) and $\hat{U}$ (i.e., considering the strongest attacker). Intuitively, \cref{eq:maximal_leakage} evaluates the ratio (in nats) of the probabilities of guessing the secret $U$ correctly with and without observing $Y$.
    Variants and generalizations of maximal leakage have been proposed, modifying \cref{eq:maximal_leakage} to penalize different values of $\probof{U=\hat{U}}$ differently, using so-called \emph{gain functions}  \cite{liao2019tunable,saeidian2022pointwise,gilani2023alpha,kurri2022operational}.
    Maximal leakage and its variants assume that the secret $U$ is unknown a priori and therefore considers the worst-case leakage over all possible secrets. However, in our problem, data holders know what secret they want to protect.

    \noindent \textbf{Min-entropy metrics}. Several papers have studied privacy metrics related to min-entropy, or the probability of guessing the secret correctly \cite{smith2009foundations,alvim2012measuring,alvim2014additive,asoodeh2017privacy,asoodeh2018estimation}.
    Among these, the most closely-related paper is by Asoodeh \emph{et al.} \cite{asoodeh2017privacy}, which directly analyzes the probability of guessing the secret, as we do (within a threshold). 
    {Adopting the same notation as before (i.e., the Markov chain $U-X-Y$), \cite{asoodeh2017privacy} aims to maximize the disclosure of $X$ (i.e., $\max_{f}\probof{X=f(Y)}$, where the $\max$ is taken over all functions $f$) to ensure high utility. This optimization is subject to a privacy constraint on the sensitive information $U$: $\max_{\hat{g}}\probof{U=\hat{g}(Y)}\leq T$, where the $\max$ is taken over all attack strategies $\hat{g}$. %
    However, the authors assume that for random variables $X$ and $Y$, %
    the value of each dimension can only be either $0$ or $1$ (i.e, each dimension of the data distribution parameter is binary). Since their analysis relies on the properties of Bernoulli distribution, the  results cannot be trivially extended to non-binary case, significantly constraining the range of distribution settings this framework can analyze. Furthermore, they assess utility based on the probability of precisely guessing the original data. However, in data-sharing contexts,} {this utility measure suffers from the same shortcomings as mutual information, namely that any random one-to-one mapping can achieve a high utility metric without having practical utility.}

\section{Problem Formulation}
\label{sec:framework}
\label{sec:model}

\myparatightestn{Notation.}
We denote random variables with uppercase English letters or upright Greek letters (e.g., $X, \text{\textmugreek}$), %
and their realizations with italicized lowercase letters (e.g., $x, \mu$). %
For a random variable $\rvprivatenotation$, we denote its probability density function (PDF) %
as $\pdfof{\rvprivatenotation}$, and its distribution measure as $\distributionof{\rvprivatenotation}$.
If a random variable $X$ is drawn from a parametric family (e.g., Gaussian with specified mean and covariance), the parameters will be denoted with a subscript of $X$, i.e., the above notations become $\rvprivate$, $\privatepdf$, $\privatedistribution$ respectively for parameters $\theta \in \mathbb R^\paramdim$, where $\paramdim \geq 1$ denotes the dimension of the parameters.
In addition, we denote $\pdfof{X|Y}$ as the conditional PDF or PMF of $X$ given another random variable $Y$.
We use $\setofintegers,\setofpositiveintegers,\setofnaturalnumbers,\setofreal,\setofpositivereal,$ to denote the set of integers, positive integers, natural numbers, real numbers, and positive real numbers respectively.

\myparatightestn{Original data.}
Consider a data holder who possesses a dataset of $n$ samples $\privatedataset=\brc{x_1,\ldots, x_n}$, %
where for each $i\in [n]$,  $x_i \in \mathbb R$ is drawn i.i.d. from an underlying distribution. 
We assume the distribution comes from a parametric family, and the parameter vector $\rvparamnotation \in \mathbb R^\paramdim$ of the distribution fully specifies the distribution. 
That is, $x_i \sim \privatedistribution$,
where we further assume that $\rvparamnotation$ is itself a realization of random parameter vector $\RVparamnotation$, and $\paramdistribution$ is the probability measure for $\RVparamnotation$. %
{We will discuss how to relax the assumption on this prior distribution of $\rvparamnotation$ in \cref{sec:discussions}.}
We assume that the data holder knows $\rvparamnotation$ (and hence knows its full data distribution $\privatedistribution$); our results and mechanisms generalize to the case when the data holder 
only possesses the dataset  $\privatedataset$ (see \cref{sec:case_study}). 

For example, suppose the original data samples come from a Gaussian distribution. We have $\rvparamnotation=\bra{\mu,\sigma}$, and $\rvprivate\sim \normaldistribution{\mu}{\sigma}$. %
$\paramdistribution$ (or $\pdfof{\RVparamnotation}$) describes the prior distribution over $\bra{\mu,\sigma}$. For example, if we know a priori that the mean of the Gaussian is drawn from a uniform distribution between 0 and 1, and $\sigma$ is always 1, we could have $\pdfof{\RVparamnotation}\bra{\mu,\sigma}=\indicatorof{\mu\in\brb{0,1}}\cdot \delta\bra{\sigma}$, where $\indicatorof{\cdot}$ is the indicator function, and $\delta$ is the Dirac delta function.

\myparatightestn{Statistical \secret{} to protect.}
We assume the data holder wants to hide 
a secret quantity, which is defined as a function of the original data distribution.  
Since the true data distribution is fully specified by parameter vector $\theta$, we define the secret as a function of $\rvparamnotation$ as follows:  $\secretofparam: \mathbb R^\paramdim \to \mathbb R$.
In the Gaussian example $\rvprivate\sim \normaldistribution{\mu}{\sigma}$,
suppose the data holder wishes to hide the mean; 
we thus have that $\secretof{{\mu,\sigma}}=\mu$.

\myparatightestn{\Datamechanism{}.}
The data holder releases data by passing the private parameter $\rvparamnotation$ through a \emph{\datamechanism{}} $\mechanismnotation$. 
That is, for a given $\rvparamnotation$, the data holder first draws internal randomness $\mechanismnoisenotation\sim \distributionofmechanismnoise$, and then releases another
distribution parameter 
$\releaservparamnotation=\mechanismof{\rvparamnotation}$, where $\mechanismnotation$ is a deterministic function, and $\distributionofmechanismnoise$ is a fixed distribution from which $\mechanismnoisenotation$ is sampled.
Note that  we assume both the input and output of $\mechanismnotation$ are distribution parameters. 
It is straightforward to generalize to the case when the input and/or output are datasets of samples (see \cref{sec:case_study}).

For example, in the Gaussian case discussed above, one \datamechanism{} could be $\mechanismof{\bra{\mu,\sigma}}=\bra{\mu+\mechanismnoisenotation, \sigma}$  where $\mechanismnoisenotation\sim \normaldistribution{0}{1}$. 
I.e., the mechanism shifts the mean by a random amount drawn from a standard Gaussian distribution and keeps the variance.%

\myparatightestn{Threat model.}
We assume that the attacker knows the parametric family from which our data is drawn, and has a prior over the parameter realization, but does not know the initial parameter $\rvparamnotation$.
The attacker is also assumed to know the \datamechanism{} $\mechanismnotation$ and output $\releaservparamnotation$ but not the realization of the data holder's internal randomness $\mechanismnoisenotation$.
The attacker guesses the initial secret $\secretofparam$ 
based on the released parameter $\releaservparamnotation$
according to estimate $\secretestimate$.
$\secretestimatenotation$ can be either random or deterministic, and we assume no computational bounds on the adversary.
For instance, in the running Gaussian example, an attacker may choose $\secretestimateof{{\mu',\sigma'}}=\mu'$.
When the data holder releases a dataset of samples instead of the parameter $\releaservparamnotation$, this formulation can be used to upper bound the attacker's performance on correctly guessing the secret, %
since the estimation error on released distribution parameter is induced due to the finite samples in the released dataset.

\myparatightestn{Privacy metric.}
The data holder wishes to prevent an attacker from guessing its secret $\secretofparam$.
\label{sec:privacy_metric}
 We define our privacy metric \privacy{} $\privacynotation{}$ as the attacker's probability of guessing the secret(s) to within a tolerance $\epsilon$, taken worst-case over all attackers $\secretestimatenotation$:
\begin{align}
    \privacynotation{} \triangleq ~\sup_{\secretestimatenotation} ~\probof{ \abs{\secretestimateof{\releaservparamnotation}- \secretofparam } \leq \privacythreshold }~.
    \label{eq:privacy}
\end{align}
The probability is taken over the randomness of the original data distribution ($\rvparamnotation\sim \paramdistribution$), the  \datamechanism{} ($\mechanismnoisenotation\sim \distributionofmechanismnoise$), and the \attackerstrategy{} ($\secretestimatenotation$).

\myparatightestn{Distortion metric.}
\label{sec:distortion_metric}
The main goal of data sharing is to provide useful data; hence, we (and data holders and users) want to understand how much the released data distorts the original data.
We define the \emph{distortion} $\distortionnotation$ of a mechanism as the worst-case distance between the original distribution and the released distribution:
\begin{align}
    \distortionnotation \triangleq \sup_{\substack{\rvparamnotation\in\support{\paramdistribution},
    \releaservparamnotation, \\
    \mechanismnoisenotation\in\support{\distributionofmechanismnoise}: \mechanismof{\rvparamnotation}=\releaservparamnotation}}\distanceof{\privatedistribution} {\releasedistribution},
    \label{eq:distortion}
\end{align}
where $d$ is 

Wasserstein-1 distance.
Wasserstein-1 distance is 
commonly used as the distance metric in neural network design (e.g., \cite{arjovsky2017wasserstein,lin2018pacgan}). 
Note that the definition in \cref{eq:distortion} 
can be  extended to \datamechanisms{} that take datasets as inputs and/or outputs.

\myparatightestn{Formulation.}
To summarize, the data holder's objective is to choose a data release mechanism that minimizes \distortion{} $\distortionnotation{}$  subject to a constraint on \privacy{} $\privacynotation{}$:
\begin{align}
    \begin{split}
    \min_{\mechanismnotation} & \qquad \distortionnotation{} \\
    \text{subject to} & \qquad \privacynotation \leq \privacymetricthreshold .
    \end{split}
    \label{eq:opt}
\end{align}
The reverse formulation, 
    $\min_{\mechanismnotation} \privacynotation{} ~~ %
    \text{subject to} ~~\distortionnotation \leq T 
    $
is analyzed in \cref{sec:analysis_alternative}.

The optimal \datamechanisms{} for \cref{eq:opt} depends on the secrets
and the characteristics of the original data. %
Data holders specify the secret function they want to protect and select the \datamechanism{} to process the raw data for sharing. 

Our goal is to study: (1) What are fundamental limits on the tradeoff between privacy and distortion? (2) Do there exist \datamechanisms{} that can match or approach these fundamental limits? 
In general, these questions can have different answers for different parametric families of data distributions and secret functions. %
In \cref{sec:general_lower_bound} and \cref{sec:general_upper_bound}, we first present general results that do not depend on data distribution or secret function. 
We then present case studies for specific secret functions and data distributions {in \cref{sec:case_study}}.

\section{General Lower Bound on \Privacy{}-\Distortion{} Tradeoffs}
\label{sec:general_lower_bound}

Given a \privacy{} budget $T$, we first present a lower bound on \distortion{}  that applies regardless of the prior distribution of data $\paramdistribution$ and regardless of the \secret{} $\secretnotation$. 

\begin{theorem}[Lower bound of \privacy{}-\distortion{} tradeoff]
\label{thm:trade_off_general}
Let $\auxdistance{\rvprivatewithparam{\rvparamnotation_1}}{\rvprivatewithparam{\rvparamnotation_2}} \triangleq \distanceformula{\rvprivatewithparam{\rvparamnotation_1}}{\rvprivatewithparam{\rvparamnotation_2}}$, 
where $\distanceof{\cdot}{\cdot}$ denotes Wasserstein-1 distance.
Further, let
$\auxrange{\rvprivatewithparam{\rvparamnotation_1}}{\rvprivatewithparam{\rvparamnotation_2}} \triangleq \rangeformula{{\rvparamnotation_1}}{{\rvparamnotation_2}}$ and  
\begin{align} 
\ratio \triangleq \inf_{\rvparamnotation_1, \rvparamnotation_2 \in\support{\paramdistribution}}
\frac{\auxdistance{\rvprivatewithparam{\rvparamnotation_1}}{\rvprivatewithparam{\rvparamnotation_2}}}{\auxrange{\rvprivatewithparam{\rvparamnotation_1}}{\rvprivatewithparam{\rvparamnotation_2}}}.
\label{eq:gamma}
\end{align}
For any $\privacymetricthreshold\in\bra{0,1}$, 
when $\privacynotation\leq \privacymetricthreshold$, 
\begin{align}
\distortionnotation> \bra{\ceil{\frac{1}{\privacymetricthreshold}}-1}\cdot 2\ratio\privacythreshold~.
\end{align}
\end{theorem}

The proof is shown as below. %
From \cref{thm:trade_off_general} we see that the lower bound of \distortion{} scales inversely  with the \privacy{} budget and positively  with the  tolerance threshold $\epsilon$.
The dependent quantity $\gamma$ in \cref{eq:gamma} can be thought of as a conversion factor that bounds the translation from probability of detection to distributional distance. 
Note that we have not made $\ratio$ exact as its form depends on the type of the secret and prior distribution of data. We will instantiate it in the cases studies in \cref{sec:case_study}.

\begin{proof}
Our proof proceeds by constructing an ensemble of attackers, such that at least one of them will be correct by construction. We do this by partitioning the space of possible secret values, and having each attacker output the midpoint of one of the subsets of the partition. We then use the fact that each attacker can be correct with probability at most $T$, combined with $\gamma$, which intuitively relates the distance between distributions to the distance between their secrets, to derive the claim.
\revision{Recall that $\theta$ is the true private parameter vector, $\theta'$ is the released parameter vector as a result of the data release mechanism.}

\begin{align*}
    \privacymetricthreshold
    &\geq \privacynotation\\
    &= \sup_{\secretestimatenotation} \probof{ \secretestimateof{\releaservparamnotation}\in\brb{ \secretofparam - \privacythreshold, \secretofparam + \privacythreshold } }\\
    &=\sup_{\secretestimatenotation} \expectationof{\probof{ \secretestimateof{\releaservparamnotation}\in\brb{ \secretofparam - \privacythreshold, \secretofparam + \privacythreshold }\bigg| \releaservparamnotation } }\\
    &= \expectationof{\sup_{\secretestimatenotation}\probof{ \secretestimateof{\releaservparamnotation}\in\brb{ \secretofparam - \privacythreshold, \secretofparam + \privacythreshold }\bigg| \releaservparamnotation } }~~,\numberthis\label{eq:esup_supe}
\end{align*}
where \cref{eq:esup_supe} is due to the following facts: (1) LHS $\leq$ RHS because \\
$ \sup_{\secretestimatenotation}\probof{ \secretestimateof{\releaservparamnotation}\in\brb{ \secretofparam - \privacythreshold, \secretofparam + \privacythreshold }\bigg| \releaservparamnotation } \geq \probof{ \secretestimateof{\releaservparamnotation}\in\brb{ \secretofparam - \privacythreshold, \secretofparam + \privacythreshold }\bigg| \releaservparamnotation }  $ for any $\releaservparamnotation$; (2) RHS $\leq$ LHS because $\secretestimatenotation$ \revision{can only depend on $\releaservparamnotation$. Therefore, we can map any $\arg\sup_{\secretestimatenotation}$ in the RHS to the LHS and obtain the same value, since the expectation is taken over $\releaservparamnotation$.}
Thus, there exists $\releaservparamnotation$ s.t.  $\sup_{\secretestimatenotation}\probof{ \secretestimateof{\releaservparamnotation}\in\brb{ \secretofparam - \privacythreshold, \secretofparam + \privacythreshold }\bigg| \releaservparamnotation } \leq \privacymetricthreshold$. 
Let $$L_{\releaservparamnotation} \triangleq \inf_{\rvparamnotation \in
\support{\paramdistribution},\mechanismnoisenotation:
\mechanismofwithnoise{\rvparamnotation}{\mechanismnoisenotation}= \releaservparamnotation} \secretofparam~,$$
$$R_{\releaservparamnotation} \triangleq \sup_{\rvparamnotation \in
\support{\paramdistribution},\mechanismnoisenotation:
\mechanismofwithnoise{\rvparamnotation}{\mechanismnoisenotation}= \releaservparamnotation} \secretofparam~.$$
We can define a sequence of attackers and a constant $N$ such that $\secretestimatesubscriptof{i}{\releaservparamnotation} = L_{\releaservparamnotation} + 
\bra{i+0.5}\cdot 2\privacythreshold$ for $i\in\brc{0,1,\ldots,N-1}$ and $L_{\releaservparamnotation}+2N\privacythreshold\geq R_{\releaservparamnotation}>L_{\releaservparamnotation}+2(N-1)\privacythreshold$ (\cref{fig:proof_lower_bound}).
\begin{figure}[t]
    \centering
    \includegraphics[width=0.8\linewidth]{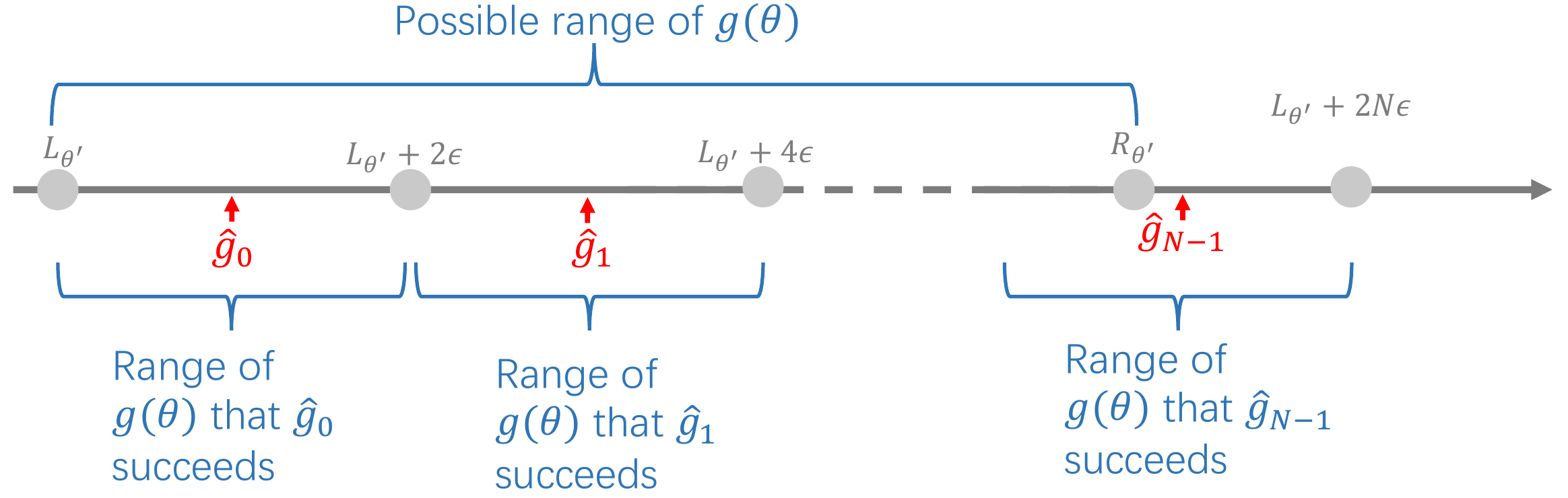}
    \caption{The construction of attackers for proof of \cref{thm:trade_off_general}. The $2\privacythreshold$ ranges of $\secretestimatenotation_0, ..., \secretestimatenotation_{N-1}$ jointly cover the entire range of possible secret $\brb{L_{\releaservparamnotation}, R_{\releaservparamnotation}}$. 
    The probability of guessing the secret correctly for any attacker is $\leq \privacymetricthreshold$.
    Therefore, 
    $R_{\releaservparamnotation}-L_{\releaservparamnotation} >\bra{\ceil{\frac{1}{\privacymetricthreshold}}-1}\cdot 2\privacythreshold$ (\cref{proof_lower_bound_given_thetap2}).
    }
    \label{fig:proof_lower_bound}
\end{figure}
From the above, we have
\begin{align*}
    \privacymetricthreshold\cdot N \geq \sum_i \probof{ \secretestimatesubscriptof{i}{\releaservparamnotation}\in\brb{ \secretofparam - \privacythreshold, \secretofparam + \privacythreshold } \bigg|\releaservparamnotation} \geq 1,
\end{align*}
Therefore, we have $N\geq \ceil{\frac{1}{\privacymetricthreshold}}$, and 
\begin{align}
    R_{\releaservparamnotation}-L_{\releaservparamnotation} >\bra{\ceil{\frac{1}{\privacymetricthreshold}}-1}\cdot 2\privacythreshold~~.
    \label{proof_lower_bound_given_thetap2}
\end{align}
Then we have
\begin{align*}
\distortionnotation &\geq\sup_{\rvparamnotation\in\support{\paramdistribution},\mechanismnoisenotation\in\support{\distributionofmechanismnoise}: \mechanismof{\rvparamnotation}=\releaservparamnotation}\distanceof{\privatedistribution} {\releasedistribution}\\
&\geq \sup_{\rvparamnotation_i \in
\support{\paramdistribution},\mechanismnoisenotation_i:
\mechanismofwithnoise{\rvparamnotation_i}{\mechanismnoisenotation_i}= \releaservparamnotation}
\auxdistance{\rvprivatewithparam{\rvparamnotation_1}}{\rvprivatewithparam{\rvparamnotation_2}}\numberthis \label{eq:lowerbound_tri}\\
&> \bra{\ceil{\frac{1}{\privacymetricthreshold}}-1}\cdot 2\ratio\privacythreshold \numberthis\label{eq:lowerbound_ratio_def}.
\end{align*}
where in
\cref{eq:lowerbound_tri}, 
\revision{$\rvparamnotation_i$ for $i\in\{1,2\}$  denotes two arbitrary parameter vectors in the support space, and \cref{eq:lowerbound_tri} }
comes from the triangle inequality, and \cref{eq:lowerbound_ratio_def} utilizes $R_{\releaservparamnotation}-L_{\releaservparamnotation} >\bra{\ceil{\frac{1}{\privacymetricthreshold}}-1}\cdot 2\privacythreshold$ and the definition of $\ratio$. 
\end{proof}
\section{
\DataMechanisms{}}
\label{sec:general_upper_bound}

We first present in \cref{sec:mechanism_description} the \emph{\mechanismname{}}, a template for \datamechanisms{} used in the case studies of \cref{sec:case_study}.
The \mechanismname{} can be instantiated differently for different secret functions and data distributions. 
We show in \cref{sec:mechanism_strategy_algorithm} techniques for instantiating the \mechanismname{}, either based on theoretical insights or numerically.
Finally, we give some intuition in \cref{sec:sawtooth} about how to analyze the \mechanismname{}.
These insights will be used in our case studies (\cref{sec:case_study})
to show that we can sometimes match the lower bounds from \cref{sec:general_lower_bound} up to small constant factors. 

\subsection{The \MechanismName{}}
\label{sec:mechanism_description}

At a high level, the \mechanismname{}s follow two steps:
\begin{enumerate}
    \item \textbf{Offline Phase:} Partition  the space of parameters $\setofprivateparam$ into carefully-chosen bins. 
    \item \textbf{Online Phase:} For an observed data distribution parameter $\rvparamnotation$, deterministically release the quantized parameters, according to the partition from the Offline Phase.
\end{enumerate}

More precisely, we first divide the set of possible distribution parameters $\setofprivateparam$ into subsets $\subsetofprivateparamof{\privateparamindexnotation}$ such that $\cup_{\privateparamindexnotation\in\privateparamindexsetnotation} \subsetofprivateparamof{\privateparamindexnotation} \supseteq \setofprivateparam$ and $ \subsetofprivateparamof{\privateparamindexnotation_1} \cap \subsetofprivateparamof{\privateparamindexnotation_2} =\emptyset$ for $\privateparamindexnotation_1\not=\privateparamindexnotation_2$, where $\privateparamindexsetnotation$ is the (possibly uncountable) set of indices of the subsets. 
For $\rvparamnotation\in \setofprivateparam$, $\paramsetindexof{\rvparamnotation}$ is the index of the set that $\rvparamnotation$ belongs to; in other words, we have
$
\paramsetindexof{\rvparamnotation} = \privateparamindexnotation
$, where $\rvparamnotation \in \subsetofprivateparamof{\privateparamindexnotation}$.
The mechanism first looks up which set $\rvparamnotation$ belongs to (i.e., $\paramsetindexof{\rvparamnotation}$), then \emph{deterministically} releases a parameter $\releaseparamofindex{\paramsetindexof{\rvparamnotation}}$ that corresponds to the set. Here, $\releaseparamofindex{\privateparamindexnotation}$ for $\privateparamindexnotation\in \privateparamindexsetnotation$ denotes another parameter. 
In short, our \datamechanism{} has the form
\begin{align*}
    \mechanismof{\rvparamnotation} = \releaseparamofindex{\paramsetindexof{\rvparamnotation}}~~.
\end{align*}

Note that the policy is fully determined by  $\subsetofprivateparamof{\privateparamindexnotation}$ and $\releaseparamofindex{\privateparamindexnotation}$.
In the remainder of the paper, we will 
show
different ways of instantiating \mechanismname{} to approach the lower bound in \cref{sec:general_lower_bound}.

Intuitively,
\mechanismsname{} will have a bounded \distortion{} as long as $\distanceof{\privatedistribution}{\releasedistributionofindex{\paramsetindexof{\rvparamnotation}}}$ is bounded for all $\rvparamnotation\in\setofprivateparam$. At the same time, they obfuscate the secret as different data distributions within the same set are mapped to the same released parameter.
It turns out this simple \emph{deterministic} mechanism is sufficient to achieve the (order) optimal \privacy{}-\distortion{} trade-offs in many cases, as opposed to differential privacy, which requires randomness to provide theoretical  guarantees \cite{dwork2014algorithmic} (examples in the case studies \cref{sec:case_study}).

\subsection{Algorithms for Instantiating the \MechanismName{}}
\label{sec:mechanism_strategy_algorithm}
To implement the \mechanismname{}, we need to define the quantization bins $\subsetofprivateparamof{\privateparamindexnotation}$ and the released parameter per bin $\releaseparamofindex{\privateparamindexnotation}$.
Depending on the data distribution, the secret function, and \mechanismname{} parameters, the mechanism can have very different  \privacy{}-\distortion{} tradeoffs. 
We present two methods for selecting quantization parameters: (1) an analytical approach, and (2) a numeric approach. 
\paragraph{Analytical approach (sketch)}
In some cases, outlined in the case studies of \cref{sec:case_study} and the appendices, we can find analytical expressions for $\subsetofprivateparamof{\privateparamindexnotation}$ and $\releaseparamofindex{\privateparamindexnotation}$ while (near-)optimally trading off privacy for distortion. 
This is usually possible when the lower bound depends on the problem parameters in a specific way (see below).
We will next illustrate the procedure through an example; precise analysis is given in \S\ref{sec:case_study}.

For example, for the Gaussian distribution where $\rvparamnotation = \bra{\mu, \sigma}$, when \secret{}=standard deviation, we can work out the lower bound from \cref{thm:trade_off_general} (details in \cref{sec:case_study_std_more}). 
Note that the lower bound is tight if our mechanism minimizes %
\begin{align}
\frac{\auxdistance{\rvprivatewithparam{\mu_1,\sigma_1}}{\rvprivatewithparam{\mu_2,\sigma_2}}}{\auxrange{\rvprivatewithparam{\mu_1,\sigma_1}}{\rvprivatewithparam{\mu_2,\sigma_2}}}= \sqrt{\frac{1}{2\pi}}e^{-\frac{1}{2}\bra{\frac{\mu_1-\mu_2}{\sigma_1-\sigma_2}}^2}-\bra{\frac{\mu_1-\mu_2}{\sigma_1-\sigma_2}}\bra{\frac{1}{2}-\Phi\bra{\bra{\frac{\mu_1-\mu_2}{\sigma_1-\sigma_2}}}}%
\label{eq:minimize}
\end{align}
where where $\auxdistance{\rvprivatewithparam{\rvparamnotation_1}}{\rvprivatewithparam{\rvparamnotation_2}}$ and $\auxrange{\rvprivatewithparam{\rvparamnotation_1}}{\rvprivatewithparam{\rvparamnotation_2}}$ are defined in \cref{thm:trade_off_general}, and $\Phi$ denotes the CDF of the standard Gaussian distribution. That is, for any true parameters $\mu_1$ and $\sigma_1$, the mechanism should always choose to release $\mu_2$ and $\sigma_2$ such that \cref{eq:minimize} is as small as possible.
The exact form of \cref{eq:minimize} is not important for now;
notice instead that the problem parameters $(\sigma_i,\mu_i)$ take the same form every time they appear in this equation. 
We define
$t(\rvparamnotation_1, \rvparamnotation_2)=\frac{\mu_1-\mu_2}{\sigma_1-\sigma_2}$ to be that form.\footnote{Indeed, for many of the case studies in \cref{sec:case_study}, $t(\rvparamnotation)$ takes an analogous form; we will see the implications of this in the analysis of the upper bound in \cref{sec:sawtooth}.}
Next, we find the $t(\rvparamnotation_1, \rvparamnotation_2)$ that minimizes \cref{eq:minimize}:
$$
t_0 \triangleq \underset{t(\rvparamnotation_1, \rvparamnotation_2)}{\arg\inf} \ \frac{\auxdistance{\rvprivatewithparam{\rvparamnotation_1}}{\rvprivatewithparam{\rvparamnotation_2}}}{\auxrange{\rvprivatewithparam{\rvparamnotation_1}}{\rvprivatewithparam{\rvparamnotation_2}}}
$$
For instance, in our Gaussian example, we can write $t_0$ as
\begin{align*}
t_0 = \underset{t(\rvparamnotation_1, \rvparamnotation_2)}{\arg\inf} \sqrt{\frac{1}{2\pi}}e^{-\frac{1}{2}\bra{t(\rvparamnotation_1, \rvparamnotation_2)}^2}-\bra{t(\rvparamnotation_1, \rvparamnotation_2)}\bra{\frac{1}{2}-\Phi\bra{t(\rvparamnotation_1, \rvparamnotation_2)}},
\end{align*}
which can be solved numerically.
Finally, we can choose $\subsetofprivateparamof{\privateparamindexnotation}$ and $\releaseparamofindex{\privateparamindexnotation}$ to be sets for which  
$
t\bra{\rvparamnotation, \releaseparamofindex{\privateparamindexnotation}} = t_0, \ 
\forall \rvparamnotation \in \subsetofprivateparamof{\privateparamindexnotation}.
$
Using this rule, we derive the mechanism:
\begin{align*}
\subsetofprivateparamof{\mu,\privateparamindexnotation} &= \brc{\bra{\mu+t_0\cdot t, \sigmalowerbound + \bra{\privateparamindexnotation+0.5}\cdot\seclen+t}| t\in \brba{-\frac{\seclen}{2}, \frac{\seclen}{2}}}
~~,\\
\releaseparamofindex{\mu,\privateparamindexnotation} &= \bra{\mu,\sigmalowerbound+\bra{\privateparamindexnotation+0.5}\cdot\seclen} ~~,\\
\privateparamindexsetnotation &=  \brc{\bra{\mu,\privateparamindexnotation}| \privateparamindexnotation\in\setofnaturalnumbers, \mu\in\setofreal},
\end{align*}
where $\seclen$ is a hyper-parameter of the mechanism that divides $\bra{\sigmaupperbound - \sigmalowerbound}$, and $\sigmaupperbound,\sigmalowerbound$ are upper and lower bounds on $\sigma$, determined by the adversary's prior.

For our Gaussian example, the resulting sets $\subsetofprivateparamof{\mu,\privateparamindexnotation}$ for the  quantization mechanism are shown in \cref{fig:upperbound_general}; the space of possible parameters is divided into infinitely many subsets $\subsetofprivateparamof{\mu,\privateparamindexnotation}$, each consisting of a diagonal line segment (parallel blue lines in \cref{fig:upperbound_general}). The space of possible $\sigma$ values is divided into segments of length $s$, which correspond to the horizontal bands in \cref{fig:upperbound_general}.
Given this choice of intervals, the mechanism proceeds as follows: when the true distribution parameters fall in one of these intervals, the mechanism releases the midpoint of the interval.
The fact that the  intervals $\subsetofprivateparamof{\mu,\privateparamindexnotation}$ are diagonal lines arises from choosing  $t(\rvparamnotation_1, \rvparamnotation_2)=\frac{\mu_1-\mu_2}{\sigma_1-\sigma_2}$;
each interval corresponds to a set of points $(\tilde \mu, \tilde \sigma)$ that satisfy $t(\rvparamnotation_1, \rvparamnotation_2)=t_0$, i.e., with slope $1/t_0$.

We will see how to use this construction to obtain upper bounds on privacy-distortion tradeoffs in \cref{sec:sawtooth}.

\paragraph{Numeric approach}
In some cases, the above procedure may not be possible.
To this end, we present a dynamic programming algorithm to  %
numerically compute the \mechanismname{} parameters. 
This algorithm achieves an optimal \privacy{}-\distortion{} tradeoff \cite{bellman1966dynamic} among the class of quantization algorithms with finite precision and continuous intervals $\subsetofprivateparamof{\privateparamindexnotation}$.
We use this algorithm in some of the case studies in \cref{sec:case_study}.
We present our dynamic programming algorithm for univariate data 
distributions.

We assume $\setofprivateparam=\brba{\rvparamlowerbound, \rvparamupperbound}$, where $\rvparamlowerbound,\rvparamupperbound$ are lower and upper bounds of $\rvparamnotation$, respectively.
We consider the class of \mechanismsname{} such that $\subsetofprivateparamof{\privateparamindexnotation}=\brba{\rvparamlowerboundof{\privateparamindexnotation}, \rvparamupperboundof{\privateparamindexnotation}}$, i.e., each subset of parameters are in a continuous range.
Furthermore, 
we explore  mechanisms such that $\rvparamlowerboundof{\privateparamindexnotation}, \rvparamupperboundof{\privateparamindexnotation},\releaseparamofindex{\privateparamindexnotation} \in \brc{\rvparamlowerbound, \rvparamlowerbound + \precision, \rvparamlowerbound + 2\precision, \ldots, \rvparamupperbound}$, where $\precision$ is a hyper-parameter that encodes numeric precision (and therefore divides $(\rvparamupperbound - \rvparamlowerbound )$). 
For example, if we want to hide the mean of a Geometric random variable with $\rvparamlowerbound=0.1$ and $\rvparamupperbound=0.9$, we could consider three-decimal-place precision, i.e.,  $\precision=0.001$ and  $\rvparamlowerboundof{\privateparamindexnotation}, \rvparamupperboundof{\privateparamindexnotation},\releaseparamofindex{\privateparamindexnotation} \in \brc{0.100, 0.101,0.102,\ldots,0.900}$.

Since $\distortionnotation$ (\cref{eq:distortion}) is defined as the \emph{worst-case} distortion whereas $\privacynotation$ (\cref{eq:privacy}) is defined as a \emph{probability}, which is related to the original data distribution, optimizing $\privacynotation$ given bounded $\distortionnotation$ (\cref{eq:opt_alternative}) %
is easier to solve than the final goal of optimizing $\distortionnotation$ given bounded $\privacynotation$ (\cref{eq:opt}).
\begin{align}
    \begin{split}
    \min_{\mechanismnotation} ~~\privacynotation{} \quad\quad\quad
    \text{subject to} ~~\distortionnotation \leq T .
    \end{split}
    \label{eq:opt_alternative}
\end{align}
Observing that in \cref{eq:opt} the optimal value of $\min_{\mechanismnotation} \distortionnotation{}$ is a monotonic decreasing %
function w.r.t. the threshold $\privacymetricthreshold$, we can use a binary search algorithm (shown in \cref{sec:algorithm_appendix}) to reduce problem \cref{eq:opt} to problem \cref{eq:opt_alternative}. 
It calls an algorithm that finds the optimal quantization mechanism with numerical precision over continuous intervals under a \distortion{} budget $\distortionbudgetnotation$ (i.e., solving \cref{eq:opt_alternative}).
This problem can be solved by a dynamic programming algorithm.
Let $\dpf{\rvparamindpnotation}$ ($\rvparamindpnotation\in \brc{\rvparamlowerbound, \rvparamlowerbound + \precision, \rvparamlowerbound + 2\precision, \ldots, \rvparamupperbound}$) be the minimal \privacy{} $\privacynotation{}$ we can get for $\setofprivateparam=\brc{\rvprivatewithparam{\rvparamnotation}: \rvparamnotation\in \brba{\rvparamlowerbound, \rvparamindpnotation}}$ such that %
$\distortionnotation\leq \distortionbudgetnotation$. 
Denote $\mathcal{D}\bra{\rvparamnotation_1, \rvparamnotation_2}$ as the minimal \distortion{} a quantization mechanism can achieve under the quantization bin $\brba{\rvparamnotation_1, \rvparamnotation_2} $, we have %
\begin{align*}
\mathcal{D}\bra{\rvparamnotation_1, \rvparamnotation_2}=
\inf_{\rvparamnotation\in \mathbb R^\paramdim }
\sup_{\rvparamnotation''\in \brba{\rvparamnotation_1, \rvparamnotation_2} } \distanceof{\distributionof{\rvprivatewithparam{\rvparamnotation''}}}{\distributionof{\rvprivatewithparam{\rvparamnotation}}},
\end{align*}
where $\distanceof{\cdot}{\cdot}$ is defined in \cref{eq:distortion}.
We also denote 
$\mathcal{D}^*\bra{\rvparamnotation_1, \rvparamnotation_2}=
\arg\inf_{\rvparamnotation\in \brba{\rvparamnotation_1, \rvparamnotation_2} }
\sup_{\rvparamnotation''\in \brba{\rvparamnotation_1, \rvparamnotation_2} } \distanceof{\distributionof{\rvprivatewithparam{\rvparamnotation''}}}{\distributionof{\rvprivatewithparam{\rvparamnotation}}}$. 
If the prior over parameters is $\pdfof{\Theta}$, we have the Bellman equation
$$    \dpf{\rvparamindpnotation} =\min_{%
    \rvparamnotation \in \brb{\rvparamlowerbound, \rvparamindpnotation-\precision}, \mathcal{D}\bra{\rvparamnotation, \rvparamindpnotation}\leq\distortionbudgetnotation} 
    \frac{\int_{\rvparamlowerbound}^{\rvparamnotation} \pdfof{\Theta}\bra{t} \mathrm{d} t}{\int_{\rvparamlowerbound}^{\rvparamindpnotation} \pdfof{\Theta}\bra{t} \mathrm{d} t}\cdot\dpf{\rvparamnotation} + \frac{\int_{\rvparamnotation}^{\rvparamindpnotation} \pdfof{\Theta}\bra{t} \mathrm{d} t}{\int_{\rvparamlowerbound}^{\rvparamindpnotation} \pdfof{\Theta}\bra{t} \mathrm{d} t}
    \cdot
    \mathcal{P}\bra{\rvparamnotation, \rvparamindpnotation}
$$
with the initial state $\dpf{\rvparamlowerbound}=0$, where
\begin{align*}
\mathcal{P}\bra{\rvparamnotation, \rvparamindpnotation} &= \probof{ \secretestimatestarof{\releaservparamnotation}\in\brb{ \secretof{\rvparamnotation_0} - \privacythreshold, \secretof{\rvparamnotation_0} + \privacythreshold } | \rvparamnotation_0 \in \brb{\rvparamnotation, \rvparamindpnotation}, \releaservparamnotation} \\
&=%
\sup_{t_1, t_2: \ \sup_{t',t''\in\brb{t_1,t_2}}\abs{\secretof{t''}-\secretof{t'}} = 2\privacythreshold}
\frac{\int_{\max{\brc{t_1, \rvparamnotation}}}^{\min{\brc{t_2, \rvparamindpnotation}}} \pdfof{\Theta}\bra{t} \mathrm{d} t}{\int_{\rvparamnotation}^{\rvparamindpnotation} \pdfof{\Theta}\bra{t} \mathrm{d} t}.
\end{align*}
$\rvparamnotation'$ is the released parameter when the private parameter $\rvparamnotation_0 \in \brb{\rvparamnotation, \rvparamindpnotation}$ and $\secretestimatestarnotation$ is the optimal attack strategy. %
The full algorithm is listed in \cref{alg:dp}. The time complexity of this algorithm is %
\small
$\calO\bra{\bra{\nicefrac{\rvparamupperbound-\rvparamlowerbound}{\precision}}^2 \cdot \timecomplexitynotation_D \cdot \timecomplexitynotation_P\cdot \timecomplexitynotation_I}$, where
$\timecomplexitynotation_D$ 
\normalsize
is the time complexity for computing $\mathcal{D}$ and $\mathcal{D}^*$, %
$\timecomplexitynotation_P$ is the time complexity for computing
$\mathcal{P}$, and $\timecomplexitynotation_I$ is the time complexity for computing the integrals in the Bellman equation. 
In our cases studies, $\mathcal{D}$ and $\mathcal{D}^*$ can be computed in $\timecomplexitynotation_D =\calO\bra{\nicefrac{\rvparamupperbound-\rvparamlowerbound}{\precision}}$, and $\mathcal{P}$ and the integrals can be computed in closed forms within constant time, i.e., $\timecomplexitynotation_P=\timecomplexitynotation_I=\calO\bra{1}$.

\begin{algorithm}[htpb]
    \LinesNumbered
	\BlankLine
	\SetKwInOut{Input}{Input}
	\caption{Dynamic-programming-based \datamechanism{} for single-parameter distributions.}
    \label{alg:dp}
	\Input{Parameter range: $\brba{\rvparamlowerbound, \rvparamupperbound}$\\
 Prior over parameter: $\pdfof{\Theta}$\\
	\Distortion{} budget: $\distortionbudgetnotation$\\
	Step size: $\precision$ (which divides $\rvparamupperbound - \rvparamlowerbound$)
	}
	\BlankLine
	$pri(\rvparamlowerbound) \leftarrow 0$\\
	$\privateparamindexsetnotation\bra{\rvparamlowerbound} \leftarrow \emptyset$\\

	\For{$\rvparamindpnotation{} \leftarrow \rvparamlowerbound + \precision, \rvparamlowerbound + 2\precision, \ldots, \rvparamupperbound$}
	{
        $pri(\rvparamindpnotation)\leftarrow \infty$\\
        $min\_t \leftarrow$ NULL\\
	    \For{$\rvparamnotation \leftarrow \rvparamindpnotation - \precision, \ldots, \rvparamlowerbound$}
	    {
	        \If{$\mathcal{D}\bra{\rvparamnotation, \rvparamindpnotation} > \distortionbudgetnotation$}{break}
	        $p \leftarrow \frac{\int_{\rvparamlowerbound}^{\rvparamnotation} \pdfof{\Theta}\bra{t} \mathrm{d} t}{\int_{\rvparamlowerbound}^{\rvparamindpnotation} \pdfof{\Theta}\bra{t} \mathrm{d} t}\cdot\dpf{\rvparamnotation} + \frac{\int_{\rvparamnotation}^{\rvparamindpnotation} \pdfof{\Theta}\bra{t} \mathrm{d} t}{\int_{\rvparamlowerbound}^{\rvparamindpnotation} \pdfof{\Theta}\bra{t} \mathrm{d} t}
    \cdot
    \mathcal{P}\bra{\rvparamnotation, \rvparamindpnotation}$\\
	        \If{$%
         p<pri(\rvparamindpnotation)$}
	        {
	            $pri(\rvparamindpnotation) \leftarrow p
             $\\
	            $min\_t \leftarrow \rvparamnotation$
	        }
	    }
	    \If{$min\_t$ is not NULL}
	    {
	        $\subsetofprivateparamof{\rvparamindpnotation} \leftarrow \brba{min\_t,~\rvparamindpnotation}$\\
	        $\releasepdfofindex{\rvparamindpnotation}\leftarrow
         \mathcal{D}^*\bra{min\_t, \rvparamindpnotation}$\\ %
	        $\privateparamindexsetnotation\bra{\rvparamindpnotation} \leftarrow \privateparamindexsetnotation\bra{min\_t} \cup \brc{\rvparamindpnotation}$\\
	    }
	}
    \If{$pri(\rvparamupperbound)=\infty$}{ERROR: No answer}
	\Return{ $pri(\rvparamupperbound)$,
	$\brc{\subsetofprivateparamof{\privateparamindexnotation}: \privateparamindexnotation\in\privateparamindexsetnotation\bra{\rvparamupperbound}}, \brc{\releasepdfofindex{\privateparamindexnotation}: \privateparamindexnotation\in\privateparamindexsetnotation\bra{\rvparamupperbound}}$}
\end{algorithm}

When dynamic programming is not practical (e.g., in high-dimensional problems), we also provide a greedy algorithm in \cref{sec:algorithm_appendix} as a baseline and show the empirical comparison between these two algorithms in the case studies (\cref{sec:case_study_mean_discrete,sec:case_study_std_more,sec:case_study_fraction}).

\subsection{Technique for Analyzing the \MechanismName{}}
\label{sec:sawtooth}
We next provide an overview of techniques for analyzing the \mechanismname{}, both for privacy and for distortion. 
We use these techniques for the analysis in our case studies, \revision{where we will make the expressions and claims more precise}. 
{For concreteness, we will recall the Gaussian example from \cref{sec:mechanism_strategy_algorithm}, for which we have already derived a mechanism. 
}

The mechanism presented in \cref{sec:mechanism_strategy_algorithm} can geometrically be interpreted as follows. Over the square of possible parameter values $\mu$ and $\sigma$ (\cref{fig:upperbound_general}), the mechanism selects intervals $\subsetofprivateparamof{\mu,\privateparamindexnotation}$ that consist of short diagonal line segments (e.g., blue line segments in \cref{fig:upperbound_general}).
When the true distribution parameters fall in one of these intervals, the mechanism releases the midpoint of the interval.

We find that many of our case studies naturally give rise to the same form of $t(\theta)$.
As a result, all of the case studies we analyze theoretically (with multiple parameters) have mechanisms that instantiate intervals $\subsetofprivateparamof{\mu,\privateparamindexnotation}$ as diagonal lines, as shown in \cref{fig:upperbound_general}.
The sawtooth technique, which we present next, can be used to analyze the privacy of all such mechanism instantiations. 
More precisely, the following pattern of \mechanismname{} admits diagonal line intervals, and can be analyzed with the sawtooth technique
(\cref{sec:case_study,sec:case_study_mean_discrete,sec:case_study_std_more}):
\begin{align*}
\subsetofprivateparamof{\mu,\privateparamindexnotation} &= \brc{\bra{\mu+t_0\cdot t, \sigmalowerbound + \bra{\privateparamindexnotation+0.5}\cdot\seclen+t}| t\in \brba{-\frac{\seclen}{2}, \frac{\seclen}{2}}}
~~,\\
\releaseparamofindex{\mu,\privateparamindexnotation} &= \bra{\mu,\sigmalowerbound+\bra{\privateparamindexnotation+0.5}\cdot\seclen} ~~,\\
\privateparamindexsetnotation &=  \brc{\bra{\mu,\privateparamindexnotation}| \privateparamindexnotation\in\setofnaturalnumbers, \mu\in\setofreal},
\end{align*}
where $\seclen$ is a hyper-parameter of the mechanism that denotes quantization bin size and divides $\bra{\sigmaupperbound - \sigmalowerbound}$ and $t_0$ is a constant that can be determined by the mechanism design strategy described in \cref{sec:mechanism_strategy_algorithm}.

\begin{figure}[t]
    \centering
    \includegraphics[width=0.65\linewidth]{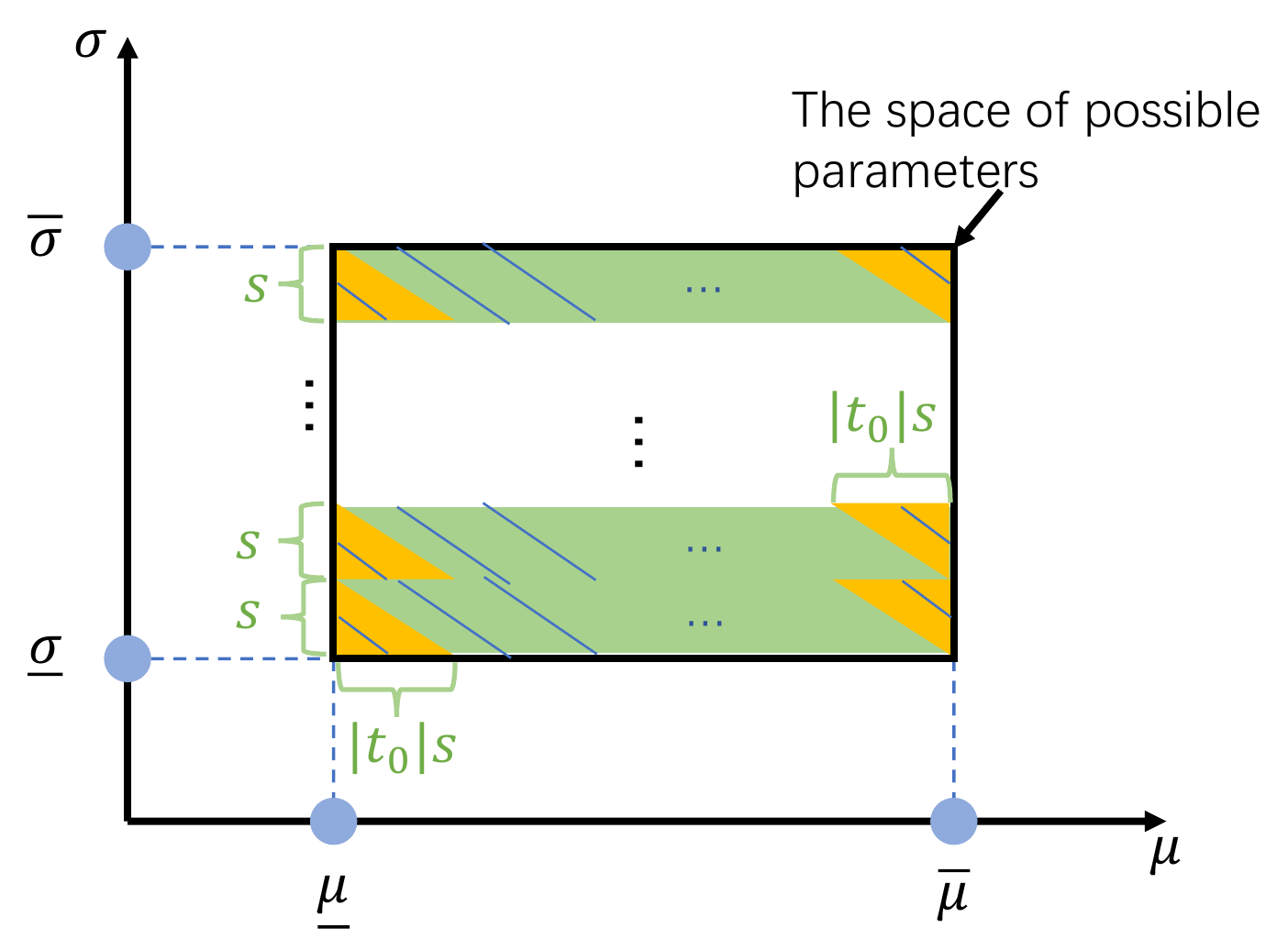}
    \vspace{-0.5cm}
    \caption{%
    We separate the space of possible parameters into two regions (yellow and green) and bound the attacker's success rate on each region separately. The blue lines represent examples of $\subsetofprivateparamof{\mu,\privateparamindexnotation}$.}
    \vspace{-0.5cm}
    \label{fig:upperbound_general}
\end{figure}

\myparatightestn{(1) Privacy analysis.}
For ease of illustration, we assume that the support of parameters is 
\small
$%
\setofprivateparam = \brc{(a,b)| a\in\brba{\underline{\mu}, \overline{\mu}},b\in\brba{\underline{\sigma}, \overline{\sigma}}}$,
\normalsize
but the analysis can be generalized to any case.

In \cref{fig:upperbound_general}, we separate the space of possible data parameters into two regions represented by yellow and green colors. The yellow regions $S_{yellow}$ constitute right triangles with height $\seclen$ and width $\abs{t_0}\seclen$.  
The green region $S_{green}$ is the rest of the parameter space. 
The high-level idea of our proof is as follows.
Note that for any parameter $\rvparamnotation\in S_{green}$, there exists a quantization bin $\subsetofprivateparamof{\mu,\privateparamindexnotation}$ s.t. $\rvparamnotation \in \subsetofprivateparamof{\mu,\privateparamindexnotation}$ and $\subsetofprivateparamof{\mu,\privateparamindexnotation}\subset S_{green}$. 
{This occurs because the mechanism intervals (blue lines in \cref{fig:upperbound_general}) all have the same slope and a length of at most $\seclen$ for $\sigma$. As such, each interval is either fully in the green region, or fully in the yellow region.}
{Since we know the length of each bin, } we can upper bound the attack success rate if  $\rvparamnotation\in S_{green}$. While the attacker can be more successful in the yellow region, the probability of $\rvparamnotation\in S_{yellow}$ is small.
Hence, we upper bound the overall attacker's success rate (i.e., $\privacynotation$). More specifically, let the optimal attacker be $\secretestimatestarnotation$. We have
\small
\begin{equation}
\begin{aligned}
\nonumber
    \privacynotation{} 
    &= \probof{ \secretestimatestarof{\releaservparamnotation}\in\brb{ \secretofparam - \privacythreshold, \secretofparam + \privacythreshold } }\\
    &= \int_{ \rvparamnotation\in S_{green}}p(\rvparamnotation)\probof{ \secretestimatestarof{\releaservparamnotation}\in\brb{ \secretofparam - \privacythreshold, \secretofparam + \privacythreshold } }d\rvparamnotation \\
    &\quad+  \int_{ \rvparamnotation\in S_{yellow}}p(\rvparamnotation)\probof{ \secretestimatestarof{\releaservparamnotation}\in\brb{ \secretofparam - \privacythreshold, \secretofparam + \privacythreshold }} d\rvparamnotation\\
    &< \sup_{\rvparamnotation\in S_{green}}\probof{ \secretestimatestarof{\releaservparamnotation}\in\brb{ \secretofparam - \privacythreshold, \secretofparam + \privacythreshold } } + \int_{ \rvparamnotation\in S_{yellow}}p(\rvparamnotation)d\rvparamnotation\\
    \label{eqn:green-yellow}
\end{aligned}
\end{equation}
\normalsize
The first term can be bounded away from 1 due to the carefully chosen $t_0$. The second term is bounded away from 1 because the size of $S_{yellow}$ is relatively small.
\revision{The formal justification is given in \cref{thm:upperbound_continuous_quantile,sec:proof_upperbound_continuous_quantile_SFHITED_EXP,sec:proof_upperbound_continuous_quantile_Gaussian_uniform,sec:proof_upperbound_continuous_std}. }

\myparatightestn{(2) Distortion analysis.}
For the \distortion{} performance, it is straightforward to show that \\
$\distortionnotation = \sup_{\rvparamnotation\in\setofprivateparam}\distanceof{\privatedistribution}{\releasedistributionofindex{\paramsetindexof{\rvparamnotation}}}$,
where $\releaseparamofindex{\paramsetindexof{\rvparamnotation}}$ is the released parameter when the original parameter is $\rvparamnotation$.
This quantity can often be derived directly from the mechanism and parameter support.

\section{Case Studies}
\label{sec:case_study}

In this section, we instantiate the general results on concrete distributions and secrets (mean \cref{sec:case_study_mean}, quantile \cref{sec:case_study_quantiles}, and we defer standard deviation and discrete distribution fractions to  \cref{sec:case_study_std_more,sec:case_study_fraction}).
See
\cref{tbl:summary} for a summary of each setting we consider, and a pointer to any theoretical results.
Our results in each setting generally include a \privacy{} lower bound, a concrete instantiation of the  \mechanismname{}, and \privacy{}-\distortion{} analysis of the \datamechanisms{}. In \cref{sec:case_study_extend_dataset}, we will discuss how to extend the \datamechanisms{} to the cases when data holders only have data samples and do not know the parameters of the underlying distributions.

\begin{table*}[th]
  \centering
  \caption{Summary of the case studies we cover, and links to the corresponding results.}
  \label{tbl:summary}
  \scriptsize
  \begin{tabular}{|c|c|c|c|c|c|c|}
    \hline
    \multirow{2}*{\diagbox{Secret}{Distribution}} & \multicolumn{3}{c|}{\makecell[c]{Continuous Distribution \\ (order-optimal mechanism)}} & \multicolumn{3}{c|}{\makecell[c]{Ordinal Distribution \\ (\cref{alg:dp} and \cref{alg:greedy})}}
    \\
    \cline{2-7}
    ~ & Gaussian & Uniform & Exponential & Geometric & Binomial & Poisson
    \\
    \hline
    Mean & \multicolumn{3}{c|}{\makecell[c]{\cref{sec:case_study_mean}}} & \multicolumn{3}{c|}{\makecell[c]{\cref{sec:case_study_mean_discrete}}}
    \\
    \hline
    Quantile & \multicolumn{3}{c|}{\makecell[c]{\cref{sec:case_study_quantiles,sec:case_study_Gaussian_uniform_quantiles}}} & \multicolumn{3}{c|}{\makecell[c]{Not applicable}}
    \\
    \hline
    Standard Deviation & \multicolumn{3}{c|}{\makecell[c]{\cref{sec:case_study_std_continuous}}} & \multicolumn{3}{c|}{\makecell[c]{\cref{sec:case_study_std_discrete}}} 
    \\
    \hline
    Fraction & \multicolumn{3}{c|}{\makecell[c]{Not applicable}} & \multicolumn{3}{c|}{\makecell[c]{\cref{sec:case_study_frac_discrete_ordinal}}} %
    \\
 \hline
     \end{tabular}
\end{table*}

\subsection{\Secret{} = Mean}
\label{sec:case_study_mean}

In this section, we discuss how to protect the mean of a distribution
for general continuous distributions.
We start with a lower bound.

\begin{corollary}[Privacy lower bound, secret = mean of a continuous distribution]
\label{thm:continuous_mean}
Consider the secret function $\secretof{\rvparamnotation}=\int_x x\privatepdf\bra{x}dx$. For any  $\privacymetricthreshold\in\bra{0,1}$, when $\privacynotation\leq \privacymetricthreshold$, we have $\distortionnotation> \bra{\ceil{\frac{1}{\privacymetricthreshold}}-1}\cdot \privacythreshold$.
\end{corollary}
The proof is in \cref{sec:proof_lowerbound_mean}. 
We next design a \datamechanism{} that achieves a tradeoff close to this bound.

\myparatightestn{\Datamechanism{}.} 
To begin, we restrict ourselves to continuous distributions that can be parameterized with a location parameter, where the prior distribution of the location parameter is  uniform and independent of other factors: 
\begin{assumption}
\label{assu:mean_continuous}
The distribution parameter vector $\rvparamnotation$ can be written as $(\rvmean, \rvparamothernotation)$, where $\rvmean\in \setofreal$, $\rvparamothernotation\in \setofreal^{\paramdim-1}$, and for any $\rvmean\neq \rvmean'$, $\pdfof{\rvprivatewithparam{\rvmean,\rvparamothernotation}}\bra{x}=\pdfof{\rvprivatewithparam{\rvmean',\rvparamothernotation}}\bra{x-\rvmean'+\rvmean}$.
The prior over distribution parameters is $\pdfof{\RVmean,\RVparamothernotation}\bra{a,b}=\pdfof{\RVmean}\bra{a}\cdot \pdfof{\RVparamothernotation}\bra{b}$, where 
$\pdfof{\RVmean}\bra{a} =
\frac{1}{\rvmeanupperbound-\rvmeanlowerbound} \indicatorof{a\in\brba{\rvmeanlowerbound, \rvmeanupperbound}}$.
\end{assumption}
Examples include the Gaussian, Laplace, and uniform distributions, as well as shifted distributions (e.g., shifted exponential, shifted log-logistic). 
We relax this assumption to Lipschitz-continuous priors 
in \cref{sec:proof_trade-off_mechanism_mean_continuous_relaxed}.
Using the strategy from \cref{sec:mechanism_strategy_algorithm}, we derive the following \mechanismname{}.
\begin{mechanism}[For secret = mean of a continuous distribution]
    \label{mech:mean_continuous}
    The parameters of the \datamechanism{} are
    \begin{align}
        &\subsetofprivateparamof{\privateparamindexnotation,\rvparamothernotation} = \brc{\bra{t,\rvparamothernotation}| t\in \brba{\rvmeanlowerbound + i\cdot \seclen,~ \rvmeanlowerbound + \bra{i+1}\cdot \seclen}}
        ,\label{mecheq:mean_continuous_set}\\
        &\releaseparamofindex{\privateparamindexnotation, \rvparamothernotation} = \bra{\rvmeanlowerbound + \bra{\privateparamindexnotation+0.5}\cdot \seclen,\rvparamothernotation  },\label{mecheq:mean_continuous_release}\\
        &\privateparamindexsetnotation=\brc{(\privateparamindexnotation, \rvparamothernotation): \privateparamindexnotation\in \brc{0, 1,\ldots, \secnum-1}, \rvparamothernotation\in \support{\distributionof{\expandafter\MakeUppercase\expandafter{\rvparamothernotation}}}},
    \end{align}
    where $\seclen$ is a hyper-parameter of the mechanism that divides $\bra{\rvmeanupperbound - \rvmeanlowerbound}$ and $\secnum=\frac{\rvmeanupperbound - \rvmeanlowerbound}{\seclen}\in \setofnaturalnumbers$.
\end{mechanism}

\cref{fig:mechansim_mean_continuous} shows an example when the original data distribution is Gaussian, i.e., $\rvprivate\sim \normaldistribution{\rvmean}{\rvparamothernotation}$, and $\rvmean\in \brba{\mulowerbound, \muupperbound}$. 
Intuitively, our \datamechanism{} ``quantizes'' the range of possible mean values into segments of length $\seclen$. 
It then shifts the mean of private distribution $\pdfof{\rvprivatewithparam{\rvmean,\rvparamothernotation }}$ to the midpoint of its corresponding segment, and releases the resulting distribution.
This simple deterministic mechanism is able to achieve order-optimal \privacy{}-\distortion{} tradeoff in some cases, as shown below.

\begin{figure}[t]
    \centering
    \includegraphics[width=0.6\linewidth]{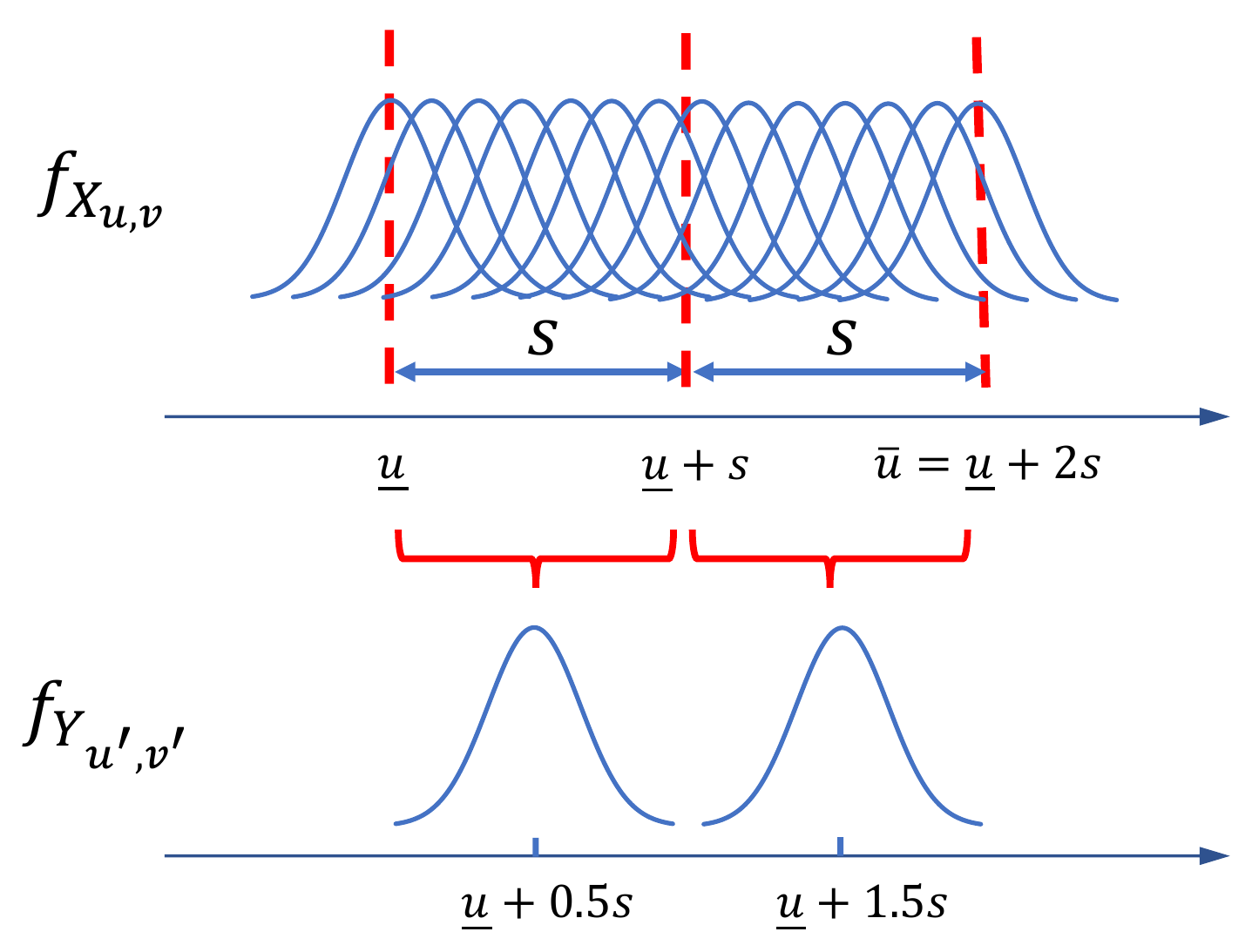}
    \caption{Illustration of the \datamechanism{} for continuous distributions when secret=mean.}
    \label{fig:mechansim_mean_continuous}
\end{figure}

\begin{proposition}
    \label{thm:trade-off_mechanism_mean_continuous}
    Under \cref{assu:mean_continuous},
    \cref{mech:mean_continuous} %
    has privacy $\privacynotation{}\leq \frac{2\privacythreshold}{\seclen}$ and distortion $\distortionnotation = \frac{\seclen}{2}< 2\distortionnotation_{\text{opt}}(\privacynotation{})$, where $\distortionnotation_{\text{opt}}(\privacynotation{})$ is the minimal \distortion{} any \datamechanism{} can achieve given privacy level $\privacynotation{}$.
\end{proposition}
The proof is in \cref{sec:proof_trade-off_mechanism_mean_continuous}.
The two takeaways from this proposition are that: (1) the data holder can use $\seclen$ to control the trade-off between \distortion{} and \privacy{}, and 
(2) the mechanism achieves an order-optimal distortion with multiplicative factor $2$.

\subsection{\Secret{} = Quantiles}
\label{sec:case_study_quantiles}
In this section, we show how to protect the $\alpha$-quantile of the exponential distribution and the shifted exponential distribution. 
We analyze the Gaussian and uniform distributions in \cref{sec:case_study_Gaussian_uniform_quantiles}.
We choose these distributions as a starting point of our analysis as many distributions in real-world data can be approximated by one of these distributions.

In our analysis, the parameters of (shifted) exponential distributions are denoted by:
\begin{packeditemize}
    \item Exponential distribution: $\rvparamnotation=\lambda$, where $\lambda$ is the scale parameter: $\pdfof{\rvprivatewithparam{\lambda}}\bra{x} = \frac{1}{\lambda} e^{-x/\lambda}$.
    \item Shifted exponential distribution generalizes the exponential distribution with an additional shift parameter $h$: $\rvparamnotation=\bra{\lambda,h}$.  In other words, $\pdfof{\rvprivatewithparam{\lambda,h}}\bra{x} = \frac{1}{\lambda} e^{-\bra{x-h}/\lambda }$.
\end{packeditemize}

As before, we first present a lower bound.
\begin{corollary}[Privacy lower bound, secret = $\alpha$-quantile of a continuous distribution]
\label{thm:lowerbound_continuous_quantile}
Consider the secret function $\secretof{\rvparamnotation}=\alpha$-quantile of $\privatepdf$. For any  $\privacymetricthreshold\in\bra{0,1}$, when $\privacynotation\leq \privacymetricthreshold$, we have $\distortionnotation> \bra{\ceil{\frac{1}{\privacymetricthreshold}}-1}\cdot 2\ratio\privacythreshold$, where $\ratio$ is defined as follows:
\begin{packeditemize}
    \item Exponential: 
    \begin{align*}
    \ratio= -\frac{1}{2\ln\bra{1-\alpha}}.\end{align*} 
    \item Shifted exponential: \begin{align*}\ratio=\begin{cases}
\frac{1}{2}\babs{1+\frac{\ln(1-\alpha)+1}{W_{-1}\bra{-\frac{\ln\bra{1-\alpha}+1}{2\bra{1-\alpha}e}}}} & \alpha \in [0, 1-e^{-1})\\
\frac{1}{2}\babs{1+\frac{\ln(1-\alpha)+1}{W_{0}\bra{-\frac{\ln\bra{1-\alpha}+1}{2\bra{1-\alpha}e}}}} & \alpha \in [1-e^{-1},1)
\end{cases},\end{align*} 
where $W_{-1}$ and $W_{0}$ are Lambert $W$ functions.
\end{packeditemize}
\end{corollary}
The proof is in \cref{sec:proof_lowerbound_continuous_quantile}.
Next, we provide \datamechanisms{} for each of the distributions that achieve trade-offs close to these bounds. %

\begin{mechanism}[For secret = $\alpha$-quantile of a continuous distribution]
    \label{mech:quantile_continuous}
    We design mechanisms for each of the distributions. In both cases, $\seclen>0$ is the quantization bin size chosen by the operator to divide $\bra{\lambdaupperbound-\lambdalowerbound}$, where $\lambdaupperbound$ and $\lambdalowerbound$ are upper and lower bounds of $\lambda$.
    \begin{packeditemize}
       \item Exponential:
        \begin{align*}
            \subsetofprivateparamof{\privateparamindexnotation} &= \brba{\lambdalowerbound + \privateparamindexnotation\cdot \seclen, \lambdalowerbound + \bra{\privateparamindexnotation+1}\cdot \seclen}
            ~~,\\
            \releaseparamofindex{\privateparamindexnotation} &= \lambdalowerbound+\bra{\privateparamindexnotation+0.5}\cdot\seclen~~,\\
            \privateparamindexsetnotation &= \setofnaturalnumbers.
        \end{align*}
        \item Shifted exponential:
        \begin{align*}
            \subsetofprivateparamof{\privateparamindexnotation,h} &= \brc{\bra{\lambdalowerbound+\bra{\privateparamindexnotation+0.5}\seclen + t, h-t_0\cdot t}|  t\in\brba{-\frac{\seclen}{2}, \frac{\seclen}{2}}}
            ~~,\\
            \releaseparamofindex{\privateparamindexnotation,h} &= \bra{\lambdalowerbound+\bra{\privateparamindexnotation + 0.5}\seclen, h }~~,\\
            \privateparamindexsetnotation &= \brc{(\privateparamindexnotation, h)| \privateparamindexnotation\in\setofnaturalnumbers, h\in\setofreal},
        \end{align*}
        where 
\begin{align*}
t_0 %
&=\begin{cases}
            -1-\ln\bra{1-\alpha}-W_{-1}\bra{-\frac{\ln\bra{1-\alpha}+1}{2\bra{1-\alpha}e}} & \bra{\alpha \in [0, 1-e^{-1})}\\
            -1-\ln\bra{1-\alpha}-W_{0}\bra{-\frac{\ln\bra{1-\alpha}+1}{2\bra{1-\alpha}e}} & \bra{\alpha \in [1-e^{-1},1)}
            \end{cases}.
\end{align*}
\end{packeditemize}
\end{mechanism}

For the \privacy{}-\distortion{} trade-off analysis of \cref{mech:quantile_continuous}, we assume that the parameters of the original data are drawn from a uniform distribution with lower and upper bounds. 
Again, we relax this assumption to Lipschitz priors in \cref{sec:proof_trade-off_mechanism_quantile_continuous_relaxed}.
Precisely,
\begin{assumption}
    \label{assu:quantile_continuous}
    The prior over distribution parameters is:
    \begin{packeditemize}
        \item Exponential: ${\text{\textlambda}}$ follows the uniform distribution over $\brba{\lambdalowerbound, \lambdaupperbound}$.
        \item Shifted exponential: $\bra{\lambda,h}$ follows the uniform distribution over $\brc{(a,b)| a\in\brba{\lambdalowerbound, \lambdaupperbound},b\in\brba{\hlowerbound, \hupperbound}}$.
    \end{packeditemize}
\end{assumption}

We relax \cref{assu:quantile_continuous} and analyze the \privacy{}-\distortion{} trade-off of \cref{mech:quantile_continuous} in \cref{sec:proof_trade-off_mechanism_quantile_continuous_relaxed}.

\begin{proposition}
    \label{thm:upperbound_continuous_quantile}
    Under \cref{assu:quantile_continuous},
    \cref{mech:quantile_continuous} %
    has the following 
$\privacynotation{}$ and $\distortionnotation$ value/bound.
    \begin{packeditemize}
        \item Exponential:
        \begin{align*}
             \privacynotation=\frac{2\privacythreshold}{-\ln\bra{1-\alpha}s},\quad\quad
             \distortionnotation=\frac{1}{2}\seclen < 2\distortionnotation_{\text{opt}}. %
        \end{align*}
        \item Shifted exponential:
        \begin{align*}
            \privacynotation %
            & < \frac{2\privacythreshold}{\abs{\ln\bra{1-\alpha}+t_0}\seclen} + \frac{\abs{t_0}\seclen}{\hupperbound-\hlowerbound},\\
            \distortionnotation %
            & = \frac{\seclen}{2}\bra{t_0-1}+\seclen e^{-t_0} < \bra{2+\frac{\abs{t_0}\cdot\abs{\ln\bra{1-\alpha}+t_0}\seclen^2}{\epsilon\bra{\hupperbound-\hlowerbound}}} \distortionnotation_{\text{opt}}.
        \end{align*}
Under the high-precision regime where $ {\frac{\seclen^2}{\hupperbound-\hlowerbound}}\rightarrow 0$ as $\seclen, (\hupperbound-\hlowerbound)\to \infty$, when $\alpha \in [0.01, 0.25]\cup [0.75, 0.99]$, 
$\distortionnotation$ satisfies
\begin{align*}
    \lim \sup_{{\frac{\seclen^2}{\hupperbound-\hlowerbound}}\rightarrow 0} \distortionnotation < 3  \distortionnotation_{\text{opt}}.
\end{align*}
    \end{packeditemize}
$\distortionnotation_{\text{opt}}$ is the optimal achievable \distortion{} given the \privacy{} achieved by \cref{mech:quantile_continuous}, and $t_0$ is a constant defined in \cref{mech:quantile_continuous}.
\end{proposition}

The proof is in \cref{sec:proof_upperbound_continuous_quantile}.
 Note that the quantization bin size 
$\seclen$ cannot be too small, or the attacker can always successfully guess the secret within a tolerance $\privacythreshold$ (i.e., $\privacynotation{} = 1$). 
Therefore, for the ``high-precision'' regime, we consider the asymptotic scaling as both $\seclen$ and $\hupperbound-\hlowerbound$ grow.
{When $s>1$,  the scaling condition ${\frac{\seclen^2}{\hupperbound-\hlowerbound}}\rightarrow 0$ implies a more interpretable condition of ${\frac{\seclen}{\hupperbound-\hlowerbound}}\rightarrow 0$, which says that the bin size is small relative to the parameter space. For example, this condition is required when the secret tolerance $\epsilon > 1/2$ (i.e., we need a bin size $s>1$ to achieve non-trivial privacy guarantees).
}

\cref{thm:upperbound_continuous_quantile} shows that the \mechanismname{} is order-optimal with multiplicative factor $2$ for the exponential distribution. 
For shifted exponential distribution, order-optimality holds asymptotically in the high-precision regime. %

\subsection{Extending \DataMechanisms{} for Dataset Input/Output}
\label{sec:case_study_extend_dataset}
The \datamechanisms{} discussed in previous sections assume that data holders know the \emph{distribution parameter} of the original data. In practice, data holders often only have a dataset of samples from the data distribution and do not know the parameters of the underlying distributions. The quantization \datamechanisms{} can be easily adapted to handle dataset input/output.

The high-level idea is that the data holders can estimate the distribution parameters $\rvparamnotation$ from the data samples and find the corresponding quantization bins $\subsetofprivateparamof{\privateparamindexnotation}$ according to the estimated parameters, and then modify the original samples as if they are sampled according to the released parameter $\releaseparamofindex{\privateparamindexnotation}$.
{This may be infeasible for high-dimensional parameter vectors $\theta$; we did not explore this question in the current work.}
For brevity, we only present the concrete procedure for  secret=mean on continuous distributions as an example. For a dataset of $\privatedataset=\brc{x_1,\ldots, x_n}$, the procedure is:
\begin{enumerate}
    \item Estimate the mean from the data samples: $\hat{\mu} = \frac{1}{n}\sum_{i\in [n]}x_i$.
    \item According to \cref{mecheq:mean_continuous_set},  compute the index of the corresponding set $\privateparamindexnotation=\floor{\frac{\hat{\mu}-\mulowerbound}{\seclen}}$.
    \item According to \cref{mecheq:mean_continuous_release}, change the mean of the data samples to $\mu_{target}=\mulowerbound + \bra{\privateparamindexnotation+0.5}\cdot \seclen$. 
    This can be done by sample-wise operation $x_i'=x_i - \hat{\mu} + \mu_{target}$.
    \item The released dataset is $\mechanismof{\privatedataset} =\brc{x_1',\ldots,x_n'}$.
\end{enumerate}

Note that this mechanism applies to samples. Therefore, it can be applied either to the original data, or as an add-on to existing data sharing tools \cite{esteban2017real,lin2020using,yin2022practical,jordon2018pate,yoon2019time}.
For example, it can be used to modify synthetically-generated samples after they are generated, or to modify the training dataset for a generative model, or to directly modify the original data for releasing. 

\section{Experiments}
\label{sec:experiments}

In the previous sections, we theoretically demonstrated the \privacy{}-\distortion{} tradeoffs of our \datamechanisms{} in some special case studies. In this section, we focus on \emph{orthogonal} questions through real-world experiments: (1) how well our \datamechanisms{} perform in practice when the assumptions do not hold, and  (2) how \name{} quantitatively compares with existing privacy frameworks %
(which we explained qualitatively in \cref{sec:motivation}). The code is open-sourced at \url{https://github.com/fjxmlzn/summary_statistic_privacy}.

\myparatightestn{Datasets.} 
We use two real-world datasets to simulate the motivating scenarios.
\begin{enumerate}
    \item \wiki{} (\wikishort{}) \cite{goog-web-traffic} contains the daily page views of 145,063 Wikipedia web pages in 2015-2016. %
    To preprocess it for our experiments, we remove the web pages with empty page view record on any day (117,277 left), and compute the mean page views across all dates for each web page. 
    Our goal is to release the page views (i.e., a 117,277-dimensional vector) while protecting the \textbf{mean of the distribution} (which reveals the business scales of the company). 
    
    \item \fccmba{} (\fccmbashort{}) \cite{mba-data} contains  network statistics (including network traffic counters) collected by United States Federal Communications Commission from homes across United States. We select the average network traffic (GB/measurement) from AT\&T clients as our data. Our goal is to release a copy of this data while hiding the \textbf{0.95-quantile} (which reveals the network capability).
\end{enumerate}

\myparatightestn{Baselines.} We compare our mechanisms discussed in \cref{sec:case_study} with three popular mechanisms proposed in prior work (\cref{sec:motivation}): differentially-private density estimation \cite{wasserman2010statistical} (shortened to DP), attribute-private Gaussian mechanism \cite{zhang2022attribute} (shortened to AP), %
{and Wasserstein mechanism for distribution privacy \cite{chen2022protecting} (shortened to DistP).} %
{As these mechanisms provide different privacy guarantees than summary statistic privacy, it is difficult to do a fair comparison between these baselines and our quantization mechanism. We include them to quantitatively show the differences (and similarities) between various privacy frameworks.}

For a dataset of samples $\privatedataset=\brc{x_1,...,x_n}$, DP works by: (1) Dividing the space into $m$ bins: $B_1,...,B_m$.%
(2) Computing the histogram $C_i=\sum_{j=1}^n \indicatorof{x_j\in B_i}$. (3) Adding noise to the histograms $D_i=\max\brc{0,C_i+\laplace{0}{\beta^2}}$, where $\laplace{0}{\beta^2}$ means a random noise from Laplace distribution with mean 0 and variance $\beta^2$.  (4)  Normalizing the histogram $p_i=\frac{D_i}{\sum_{j=1}^m D_j}$. We can then draw $y_i$ according to the histogram and release $\releasedataset=\brc{y_1,...,y_n}$ with differential privacy guarantees.
AP works by releasing $\releasedataset=\brc{x_i+\normaldistribution{0}{\beta^2}}_{i=1}^{n}$.%
{DistP works by releasing $\releasedataset=\brc{x_i+\laplace{0}{\beta^2}}_{i=1}^{n}$. }%
\revision{Note that for each of these mechanisms, normally their noise parameters would be set carefully to match the desired privacy guarantees (e.g., differential privacy). In our case, since our privacy metric is different, it is unclear how to set the noise parameters for a fair privacy comparison. For this reason, we evaluate different settings of the noise parameters, and measure the empirical tradeoffs. }

\myparatightestn{Metrics.} 
Our \privacy{} and \distortion{} metrics depend on the prior distribution of the original data $\rvparamnotation\sim \paramdistribution$ (though the mechanism does not). In practice (and also in these experiments), the data holder only has one dataset. Therefore, we cannot empirically evaluate the proposed \privacy{} and \distortion{} metrics, and resort to surrogate metrics to bound our true privacy and 
distortion.  

\myparaemphtightestn{Surrogate privacy metric.} For an original dataset $\privatedataset=\brc{x_1,...,x_n}$ and the released dataset $\releasedataset=\brc{y_1,...,y_n}$, we define the surrogate privacy metric $\tilde \Pi_\epsilon$ as the error of an attacker who guesses the secret of the released dataset as the true secret: 
$
\tilde \Pi_\epsilon \triangleq -\brd{\secretof{\privatedataset} - \secretof{\releasedataset}}
$, where $\secretof{\calD}=$ mean of $\calD$ and $0.95$-quantile of $\calD$ in 
\wikishort{} and \fccmbashort{} datasets
respectively.
Note that in the definition of $\tilde \Pi_\epsilon$, a minus sign is added so that  a smaller value indicates stronger privacy, as in privacy metric \cref{eq:privacy}.
This simple \attackerstrategy{} is in fact a good proxy for evaluating the privacy $\privacynotation$ due to the following facts. (1) For our \datamechanisms{} for these secrets \cref{mech:mean_continuous,mech:quantile_continuous,mech:fraction_discrete_general}, %
 when the prior distribution is uniform, this strategy is actually optimal, %
 so there is a direct mapping between $\tilde \Pi_\epsilon$ and $\privacynotation$. 
(2) For AP applied on protecting mean of the data (i.e., \wiki{} experiments), this strategy gives an unbiased estimator of the secret. (3) For DP and AP on other cases, this mechanism may not be an unbiased estimator of the secret, but it gives an \emph{upper bound} on the attacker's error.

\myparaemphtightestn{Surrogate distortion metric.} We define our surrogate distortion metric as the Wasserstein-1 
distance between the two datasets: $\tilde{\distortionnotation}\triangleq\distanceof{p_\privatedataset}{p_\releasedataset}$ where $p_D$ denotes the empirical distribution of a dataset $D$. %
This metric evaluates how much the mechanism distorts the dataset.

In fact, we can deduce a theoretical lower bound for the surrogate \privacy{} and \distortion{} metrics for secret = mean (shown later in \cref{fig:results}) using similar techniques as the proofs in the main paper (see \cref{sec:proof_lowerbound_surrogate}).

\begin{figure*}[th]
    \centering
    \begin{subfigure}{0.45\textwidth}
         \centering
        \includegraphics[width=1\linewidth]{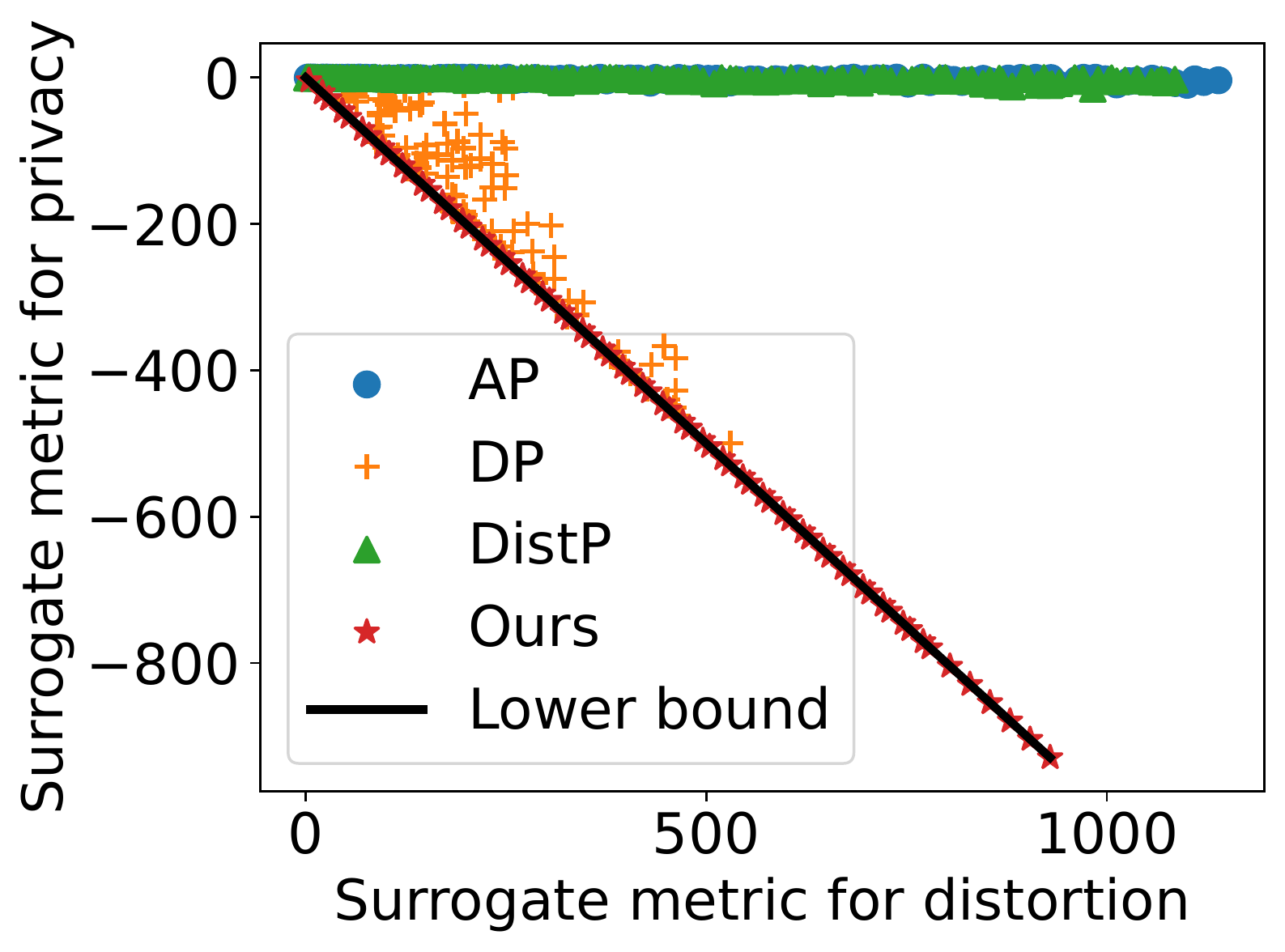}
         \caption{\wiki{}. \\(secret=mean)}
         \label{fig:wiki}
     \end{subfigure}
     \hfill
    \begin{subfigure}{0.45\textwidth}
         \centering
        \includegraphics[width=1\linewidth]{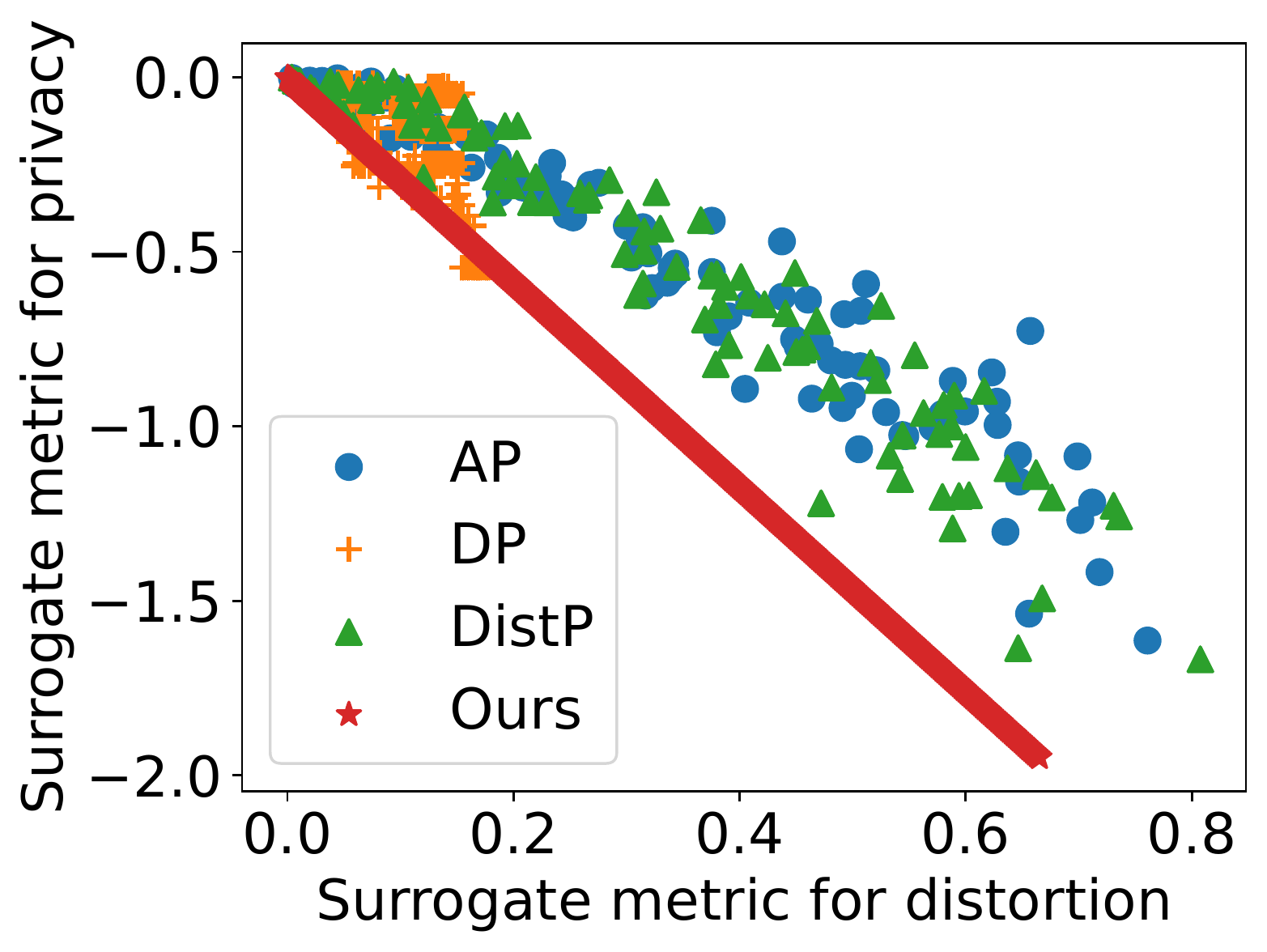}
         \caption{\fccmba{}. \\ (secret=quantile)}
         \label{fig:fccmba}
     \end{subfigure}
    \caption{Privacy (lower is better) and distortion (lower is better) of AP, DP, DistP, and ours. Each point represents one instance of \datamechanism{} with one hyper-parameter. 
    ``Lower bound'' is the theoretical lower bound of the achievable region.  
    Our \datamechanisms{} achieve better \privacy{}-\distortion{} tradeoff  than AP, DP, and DistP. 
    }
    \label{fig:results}
\end{figure*}

\subsection{Results}
We enumerate the hyper-parameters of each method (bin size and $\beta$ for DP, $\beta$ for AP %
{and DistP}, and $\seclen$ for ours). For each method and each hyper-parameter, we  compute their surrogate \privacy{} and \distortion{} metrics. The results are shown in \cref{fig:results} (bottom left is best); each data point represents one realization of mechanism $\mechanismnotation$ under a distinct hyperparameter setting. Two takeaways are below.

\noindent \emph{(1) The proposed quantization \datamechanisms{} has a good surrogate \privacy{}-\distortion{} trade-off, even when the assumptions do not hold. } 
In practical scenarios, the data distributions %
analyzed in \cref{sec:case_study} and in the Appendices may not always match real data exactly. 
Our \datamechanism{} for mean (i.e., \cref{mech:mean_continuous} used in \wikishort{}) %
supports general continuous distributions, %
and therefore, there is no such a distribution gap. Indeed, even for these surrogate metrics, our \cref{mech:mean_continuous} %
is also optimal (see \cref{sec:proof_lowerbound_surrogate}). This is visualized in \cref{fig:wiki} where we can see that our \datamechanism{} matches the theoretical lower bound of the trade-off.
However, the quantization \datamechanisms{} for quantiles (i.e., \cref{mech:quantile_continuous} used in \cref{fig:fccmba}) are order-optimal only when the distributions are within certain classes (\cref{sec:case_study_quantiles}). Observing that network traffic in \fccmbashort{} follows a one-side fat-tailed distribution (not shown), we apply the \datamechanism{} for the exponential distribution (\cref{mech:quantile_continuous}) for this dataset (which is not heavy-tailed). Despite the distribution mismatch, the quantization \datamechanism{} still achieves a good \privacy{}-\distortion{} compared to DP, AP, and DistP (\cref{fig:fccmba}). More discussions are below.

\noindent \emph{(2) The quantization \datamechanisms{} achieve better \privacy{}-\distortion{} trade-off than DP, AP, and DistP.}
AP %
{and DistP} directly add Gaussian/Laplace noise to each sample. This process does not change the mean of the distribution on expectation. Therefore, Figure \ref{fig:results} shows that AP {and DistP have} a bad \privacy{}-\distortion{} tradeoff.
DP quantizes (bins) the samples before adding noise. Quantization has a better property in terms of protecting the mean of the distribution, and therefore we see that DP has a better \privacy{}-\distortion{} tradeoff than AP {and DistP}, but still worse than the quantization mechanism. 
Note that in \cref{fig:fccmba}, a few of the DP instances have better \privacy{}-\distortion{} trade-offs than ours. This is not an indication that DP is fundamentally better. Due to the randomness in DP (from the added Laplace noise), some realizations of the noise in this experiment happened to lead to a better trade-off. Another instance of the DP algorithm could lead to a bad trade-off, and therefore, DP's achievable trade-off points are widespread.

In summary, these empirical results confirm the intuition in \cref{sec:motivation} that DP, AP, and DistP may not achieve good privacy-utility tradeoffs for our problem. This is expected---they are designed for a different objective. 
\revision{Additional results on downstream tasks are in \cref{app:additional_results}.}

\section{Discussion and Future Work}
\label{sec:discussions}

This work introduces a framework for \emph{\name{}} concerns in data sharing applications. 
This framework can be used to analyze the leakage of statistical information and the \privacy{}-\distortion{} trade-offs of \datamechanisms{} (\cref{sec:framework,sec:general_lower_bound}).
The quantization \datamechanisms{} can be used to protect statistical information (\cref{sec:general_upper_bound,sec:case_study}).
However, many interesting open questions for future work remain.

\revision{\myparatightestn{Composition guarantees.}
A limitation of the current \privacy{} metric $\privacynotation$ is that it does not provide composition guarantees; in other words, if one applies a summary statistic-private mechanism $\upsilon$ times, we cannot easily bound the privacy parameter of the $\upsilon$-fold composed mechanism. 
In contrast, composition is an important and desirable property exhibited by differential privacy \cite{dwork2014algorithmic}. 
The lack of composition can be problematic in situations where a data holder wants to release a dataset (or correlated datasets) multiple times. Understanding how to alter the definition to provide composition may be useful.}

\revision{\myparatightestn{Number of secrets.} In this work, we studied the case where the data holder only wishes to hide a single secret. In practice, data holders often want to hide multiple properties of their underlying data. It would be useful to understand how to best extend the analysis to such a setting.}

\revision{\myparatightestn{The dimension and the type of data distributions.} 
Although the proof for the lower bound in Section 
\ref{sec:general_lower_bound}
applies to general prior distributions, we analyze the \mechanismname{} under a limited set of one-dimensional distributions (\cref{tbl:summary}), assuming different parameters of the distribution are drawn independently of each other. An interesting direction for future work is to define mechanisms that have good tradeoffs under prior distributions with correlated parameters and priors. }

\myparatightestn{Relation to Differential Privacy}
Figure \ref{fig:results} suggests that despite being designed for a different threat model, the DP mechanism does fairly well. As mentioned, this is because the mchanism first bins data points, which is similar to quantization. However, this raises an important question: under what conditions on the true data, the secret quantity, and the mechanism do differentially-private mechanisms achieve a good privacy-utility tradeoff for our problem? 

\myparatightestn{Approximation error.} We studied a number of data distributions and prior distributions in this work. However, an interesting question is to bound the error in \privacy{} and \distortion{} metrics as a function of approximation error when describing either the original data distribution or the prior.

\myparatightestn{Extensions.} 
Finally, one limitation of the current \privacy{} metric $\privacynotation$ is that it depends on the prior distribution of the parameters $\paramdistribution$, which is unknown in many applications. Motivated by maximal leakage \cite{issa2019operational}, one possibility is to consider a \emph{normalized} \privacy{} metric:
\begin{align*}
   \privacynotation' \triangleq \sup_{\paramdistribution} ~\log \frac{\privacynotation}{\sup_{\secretestimatenotation} ~\probof{ \secretestimateof{\paramdistribution}\in\brb{ \secretofparam - \privacythreshold, \secretofparam + \privacythreshold } }} ,
\end{align*}
where $\secretestimateof{\paramdistribution}$ is an attacker that knows the prior distribution but does not see the released data, and the denominator is the probability that the strongest attacker guesses the secret within tolerance $\privacythreshold$. 
Similar to maximal leakage, 
we consider the worst-case leakage among all possible priors. This \emph{normalized} $\privacynotation'$ 
considers how much additional ``information'' that the released data provides to the attacker in the worst-case (see also inferential privacy \cite{ghosh2016inferential}).

\section*{Acknowledgments}
The authors gratefully acknowledge the support of NSF grants CIF-1705007 and RINGS-2148359, as well as support from the Sloan Foundation, Intel, J.P. Morgan Chase, Siemens, Bosch, and Cisco. This material is based upon work supported by the U.S. Army Research Office and the U.S. Army Futures Command under Contract No. W911NF20D0002. The content of the information does not necessarily reflect the position or the policy of the government and no official endorsement should be inferred.


\newpage
\begin{appendices}
\section{Analysis of the Alternative Formulation}
\label{sec:analysis_alternative}

In this section, we present the alternative formulation of minimizing privacy metric $\privacynotation{}$ subject to a constraint on distortion $\distortionnotation{}$:
\begin{align}
    \begin{split}
    \min_{\mechanismnotation} ~~\privacynotation{} \quad\quad\quad
    \text{subject to} ~~\distortionnotation \leq T 
    \end{split}
\end{align}

\begin{theorem}[Lower bound of \privacy{}-\distortion{} tradeoff]
\label{thm:trade_off_general_alternative}
Let $\auxdistance{\rvprivatewithparam{\rvparamnotation_1}}{\rvprivatewithparam{\rvparamnotation_2}} \triangleq \distanceformula{\rvprivatewithparam{\rvparamnotation_1}}{\rvprivatewithparam{\rvparamnotation_2}}$, 
where $\distanceof{\cdot}{\cdot}$ is defined in \cref{eq:distortion}.
Further, let
$\auxrange{\rvprivatewithparam{\rvparamnotation_1}}{\rvprivatewithparam{\rvparamnotation_2}} \triangleq \rangeformula{{\rvparamnotation_1}}{{\rvparamnotation_2}}$, 
and let $\ratio \triangleq$ $\  \inf_{\rvparamnotation_1, \rvparamnotation_2 \in\support{\paramdistribution}}
\frac{\auxdistance{\rvprivatewithparam{\rvparamnotation_1}}{\rvprivatewithparam{\rvparamnotation_2}}}{\auxrange{\rvprivatewithparam{\rvparamnotation_1}}{\rvprivatewithparam{\rvparamnotation_2}}}$.
For any $\distortionthreshold > 0$, 
when $\distortionnotation\leq \distortionthreshold$, we have $\privacynotation\geq \ceil{\frac{\distortionthreshold}{2\ratio\privacythreshold}}^{-1}$.
\end{theorem}

\begin{proof}
For any $\releaservparamnotation$, we have
\begin{align*}
\distortionthreshold &\geq \distortionnotation \\
&\geq \sup_{\rvparamnotation\in\support{\paramdistribution},\mechanismnoisenotation\in\support{\distributionofmechanismnoise}: \mechanismof{\rvparamnotation}=\releaservparamnotation}\distanceof{\privatedistribution} {\releasedistribution}\\
&\geq \sup_{\rvparamnotation_i \in
\support{\paramdistribution},\mechanismnoisenotation_i:
\mechanismofwithnoise{\rvparamnotation_i}{\mechanismnoisenotation_i}= \releaservparamnotation}
\auxdistance{\rvprivatewithparam{\rvparamnotation_1}}{\rvprivatewithparam{\rvparamnotation_2}}\numberthis \label{eq:lowerbound_tri_alternative}\\
&\geq \ratio\cdot \sup_{\rvparamnotation_i \in
\support{\paramdistribution},\mechanismnoisenotation_i:
\mechanismofwithnoise{\rvparamnotation_i}{\mechanismnoisenotation_i}= \releaservparamnotation}
\auxrange{\rvprivatewithparam{\rvparamnotation_1}}{\rvprivatewithparam{\rvparamnotation_2}}
\end{align*}
where \cref{eq:lowerbound_tri_alternative}  comes from triangle inequality. 

Let $$L_{\releaservparamnotation} \triangleq \inf_{\rvparamnotation \in
\support{\paramdistribution},\mechanismnoisenotation:
\mechanismofwithnoise{\rvparamnotation}{\mechanismnoisenotation}= \releaservparamnotation} \secretofparam~,$$
$$R_{\releaservparamnotation} \triangleq \sup_{\rvparamnotation \in
\support{\paramdistribution},\mechanismnoisenotation:
\mechanismofwithnoise{\rvparamnotation}{\mechanismnoisenotation}= \releaservparamnotation} \secretofparam~.$$
From the above result, we know that $R_{\releaservparamnotation} - L_{\releaservparamnotation} \leq \frac{\distortionthreshold}{\ratio}$. 
We can define a sequence of attackers such that $\secretestimatesubscriptof{i}{\releaservparamnotation} = L_{\releaservparamnotation} + 
\bra{i+0.5}\cdot 2\privacythreshold $ for $i\in\brc{0,1,\ldots,\ceil{\frac{\distortionthreshold}{2\ratio\privacythreshold}}-1}$ (\cref{fig:proof_lower_bound_alternative}). %
\begin{figure}[h]
    \centering
    \includegraphics[width=1\linewidth]{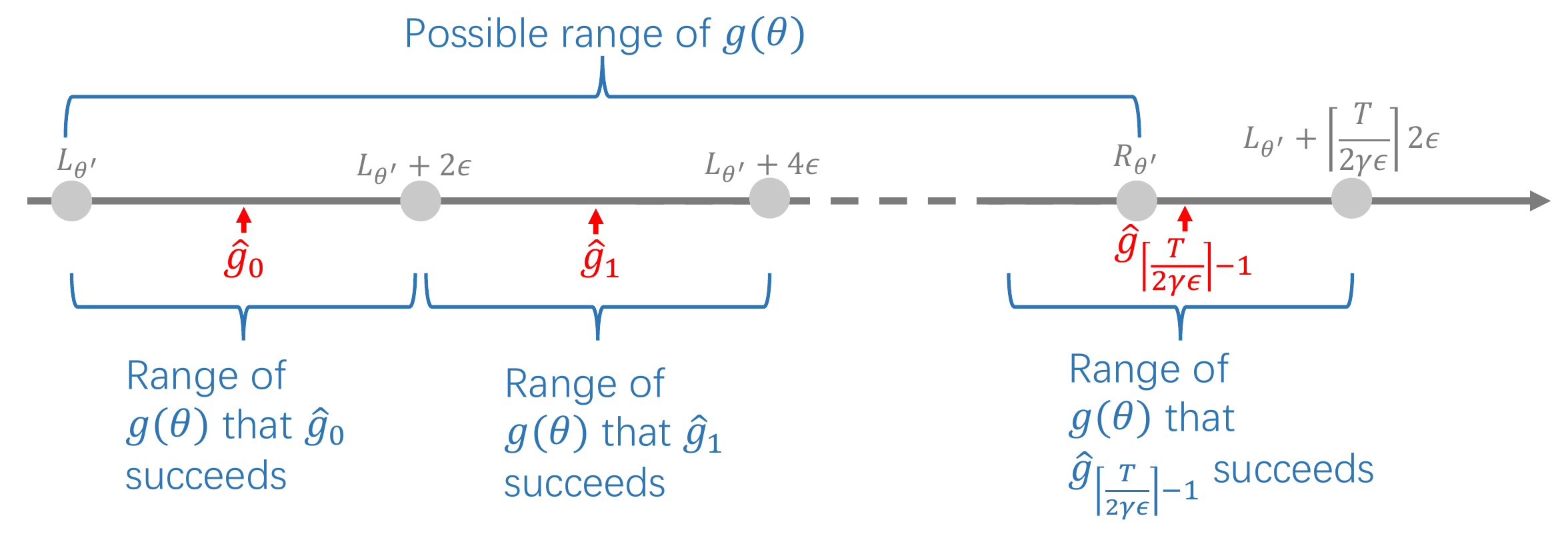}
    \caption{The construction of attackers for proof of \cref{thm:trade_off_general_alternative}. The $2\privacythreshold$ ranges of $\secretestimatenotation_0, ..., \secretestimatenotation_{\ceil{\frac{\distortionthreshold}{2\ratio\privacythreshold}}-1}$ jointly cover the entire range of possible secret $\brb{L_{\releaservparamnotation}, R_{\releaservparamnotation}}$. Therefore, there exists one attacker whose probability of guessing the secret correctly within $\privacythreshold$ is $\geq \ceil{\frac{\distortionthreshold}{2\ratio\privacythreshold}}^{-1}$ (\cref{eq:proof_lower_bound_given_thetap}).}
    \label{fig:proof_lower_bound_alternative}
\end{figure}
We have
\begin{align*}
    \sum_i \probof{ \secretestimatesubscriptof{i}{\releaservparamnotation}\in\brb{ \secretofparam - \privacythreshold, \secretofparam + \privacythreshold } \bigg|\releaservparamnotation} \geq 1,
\end{align*}
and therefore, 
\begin{align}
    \max_i \probof{ \secretestimatesubscriptof{i}{\releaservparamnotation}\in\brb{ \secretofparam - \privacythreshold, \secretofparam + \privacythreshold } \bigg|\releaservparamnotation} \geq \ceil{\frac{\distortionthreshold}{2\ratio\privacythreshold}}^{-1},
    \label{eq:proof_lower_bound_given_thetap}
\end{align}
which implies that
\begin{align*}
    \sup_{\secretestimatenotation} \probof{ \secretestimateof{\releaservparamnotation}\in\brb{ \secretofparam - \privacythreshold, \secretofparam + \privacythreshold } \bigg|\releaservparamnotation}\geq \ceil{\frac{\distortionthreshold}{2\ratio\privacythreshold}}^{-1}.
\end{align*}

Therefore, we have 
\begin{align*}
    \privacynotation{} 
    &= \sup_{\secretestimatenotation} \probof{ \secretestimateof{\releaservparamnotation}\in\brb{ \secretofparam - \privacythreshold, \secretofparam + \privacythreshold } }\\
    &=\sup_{\secretestimatenotation} \expectationof{\probof{ \secretestimateof{\releaservparamnotation}\in\brb{ \secretofparam - \privacythreshold, \secretofparam + \privacythreshold }\bigg| \releaservparamnotation } }\\
    &= \expectationof{\sup_{\secretestimatenotation}\probof{ \secretestimateof{\releaservparamnotation}\in\brb{ \secretofparam - \privacythreshold, \secretofparam + \privacythreshold }\bigg| \releaservparamnotation } }\\
    &\geq \ceil{\frac{\distortionthreshold}{2\ratio\privacythreshold}}^{-1}.
\end{align*}
\end{proof}
\section{Binary Search and Greedy Algorithms for Designing \MechanismName{}}
\label{sec:algorithm_appendix}

We use the binary search algorithm in \cref{alg:binary_search} to search for the \distortion{} budget that matches the \privacy{} budget under the optimal \datamechanism{}.

\begin{algorithm}[htpb]
    \LinesNumbered
	\BlankLine
	\SetKwInOut{Input}{Input}
	\caption{\Datamechanism{} with privacy budget.}
    \label{alg:binary_search}
	\Input{Parameter range: $\brba{\rvparamlowerbound, \rvparamupperbound}$\\
	\Privacy{} budget: $\privacymetricthreshold$\\
        \Distortion{} budget search range: $[\distortionbudgetlowerbound, \distortionbudgetupperbound]$\\
	Step size: $\seclen$ (which divides $\rvparamupperbound - \rvparamlowerbound$)\\

    Precision: $\binarysearchprecision$
	}
	\BlankLine

    \While{$\overline{\distortionbudgetnotation}-\underline{\distortionbudgetnotation}\geq \binarysearchprecision$}{
    $pri, \mathcal{S}, \releasepdfofindex{} \leftarrow \text{Algorithm-1}\bra{\brba{\rvparamlowerbound, \rvparamupperbound}, \frac{\overline{\distortionbudgetnotation}+\underline{\distortionbudgetnotation}}{2}, \precision}$\\

    \If{$pri > \privacymetricthreshold$}{$\underline{B}\leftarrow \frac{\overline{\distortionbudgetnotation}+\underline{\distortionbudgetnotation}}{2}$}
    \Else{$\overline{B}\leftarrow \frac{\overline{\distortionbudgetnotation}+\underline{\distortionbudgetnotation}}{2}$}
    }

\Return{\Datamechanism{} parameters: $\mathcal{S}, \releasepdfofindex{}$}
\end{algorithm}

We provide the greedy algorithm in \cref{alg:greedy}. 
In this algorithm, we greedily select the ranges of $\rvparamnotation$ for each $\subsetofprivateparamof{\privateparamindexnotation}$ in order. The left end point of the first range is the parameter lower bound (\cref{ln:greedy_left_end_of_first_seg}).
We then scan across all possible right end point such that the \distortion{} for this range will not exceed the budget $\distortionbudgetnotation$ (\cref{ln:greedy_scan_right_end}), and pick the one that gives the minimal attacker confidence (\cref{ln:greedy_pick_min}). After deciding the range of $\rvparamnotation$, we will set of the released distribution for this range %
(\cref{ln:greedy_set_released_dist}), and then move on to the next range (\cref{ln:greedy_next}).
The time complexity of this algorithm is $\calO\bra{\bra{\nicefrac{\rvparamupperbound-\rvparamlowerbound}{\precision}}^2 \cdot \timecomplexitynotation_D \cdot \timecomplexitynotation_P}$, the same as the dynamic programming algorithm.

\begin{algorithm}[htpb]
    \LinesNumbered
	\BlankLine
	\SetKwInOut{Input}{Input}
	\caption{Greedy-based \datamechanism{} for single-parameter distributions.}
    \label{alg:greedy}
	\Input{Parameter range: $\brab{\rvparamlowerbound, \rvparamupperbound}$\\
 Prior over parameter: $\pdfof{\Theta}$\\
	\Distortion{} budget: $\distortionbudgetnotation$\\
	Step size: $\precision$ (which divides $\rvparamupperbound - \rvparamlowerbound$)
	}
	\BlankLine
	$\privateparamindexsetnotation \leftarrow \emptyset$\\
	$L \leftarrow \rvparamlowerbound$\label{ln:greedy_left_end_of_first_seg}\\
    $privacy \leftarrow 0$ \\
	\While{$L < \rvparamupperbound$} 
	{ 
	  $min\_p\leftarrow \infty$\\
	  $min\_R \leftarrow$ NULL\\
	  $R \leftarrow L$\\
	  \While{$R \leq \rvparamupperbound$ and $\mathcal{D}\bra{L, R} \leq \distortionbudgetnotation$}
	  { \label{ln:greedy_scan_right_end}
	    $p \leftarrow \mathcal{P}\bra{L, R}$ \\
	   \If{$p \leq min\_p$} 
	    {\label{ln:greedy_pick_min}
	      $min\_p \leftarrow p$\\
	      $min\_R \leftarrow R$\\
	    }
	    $R \leftarrow R+\precision$\\
	  }
	  \If{$min\_R$ is not NULL}
	  {
	    $\subsetofprivateparamof{L} \leftarrow \brc{\rvprivatewithparam{\rvparamnotation}: \rvparamnotation \in \brab{L,~min\_R}}$\\
	    $\releasepdfofindex{L}\leftarrow \mathcal{D}\bra{L, min\_R}$%
\label{ln:greedy_set_released_dist}\\
	    $\privateparamindexsetnotation \leftarrow \privateparamindexsetnotation \cup \brc{L}$\\
        $privacy \leftarrow \frac{\int_{\rvparamlowerbound}^{L} \pdfof{\Theta}\bra{t} \mathrm{d} t}{\int_{\rvparamlowerbound}^{min\_R} \pdfof{\Theta}\bra{t} \mathrm{d} t} \cdot privacy + \frac{\int_{L}^{min\_R} \pdfof{\Theta}\bra{t} \mathrm{d} t}{\int_{\rvparamlowerbound}^{min\_R} \pdfof{\Theta}\bra{t} \mathrm{d} t}\cdot min\_p$
	  }
	  \Else{ERROR: No answer}
	  $L\leftarrow min\_R$ \label{ln:greedy_next}
	}
	\Return{ 
	$privacy$, $\brc{\subsetofprivateparamof{\privateparamindexnotation}: \privateparamindexnotation\in\privateparamindexsetnotation}, \brc{\releasepdfofindex{\privateparamindexnotation}: \privateparamindexnotation\in\privateparamindexsetnotation}$}
\end{algorithm}
\section{Proofs}
\label{section/proofs}

\subsection{Proof of %
\cref{thm:continuous_mean}}
\label{sec:proof_lowerbound_mean}

\begin{proof}

For any $\rvprivatewithparam{\rvparamnotation_1}, \rvprivatewithparam{\rvparamnotation_2}$, we have
\begin{align*}
\auxdistance{\rvprivatewithparam{\rvparamnotation_1}}{\rvprivatewithparam{\rvparamnotation_2}} 
&= \distanceformulawass{\rvprivatewithparam{\rvparamnotation_1}}{\rvprivatewithparam{\rvparamnotation_2}} \\
& \geq \frac{1}{2}
\abs{\secretof{\rvparamnotation_1}-\secretof{\rvparamnotation_2}}\numberthis \label{eq:lowerbound_mean_was}\\
&= \frac{1}{2} \auxrange{\rvprivatewithparam{\rvparamnotation_1}}{\rvprivatewithparam{\rvparamnotation_2}} .
\end{align*}
where \cref{eq:lowerbound_mean_was} comes from  Jensen's inequality.
Therefore, we have $\ratio = \inf_{\rvparamnotation_1, \rvparamnotation_2 \in\support{\paramdistribution}}
\frac{\auxdistance{\rvprivatewithparam{\rvparamnotation_1}}{\rvprivatewithparam{\rvparamnotation_2}}}{\auxrange{\rvprivatewithparam{\rvparamnotation_1}}{\rvprivatewithparam{\rvparamnotation_2}}}\geq \frac{1}{2}$. The result then follows from \cref{thm:trade_off_general}.
\end{proof}
\subsection{Proof of %
\cref{thm:trade-off_mechanism_mean_continuous}}
\label{sec:proof_trade-off_mechanism_mean_continuous}

\begin{proof}
For any released parameter $\releaservparamnotation = (\rvmean', \rvparamothernotation')$,
there exists $i \in \brc{0,...,N-1}$
such that
$\rvmean' =\rvmeanlowerbound + \bra{i+0.5}\cdot \seclen$. We have
\begin{equation}
\begin{aligned}
\nonumber
\sup_{\secretestimatenotation} \  &  \probof{ \secretestimateof{\releaservparamnotation}\in\brb{ \secretofparam - \privacythreshold, \secretofparam + \privacythreshold } \big| \releaservparamnotation }\\
    &= \sup_{\secretestimatenotation} \int_{\rvmeanlowerbound+ i\cdot \seclen}^{\rvmeanlowerbound + \bra{i+1}\cdot \seclen} \pdfof{U|U' }\bra{\rvmean|\rvmean'} \cdot 
    \int_{\rvmean-\privacythreshold}^{\rvmean+\privacythreshold} %
    \pdfof{\secretestimateof{\rvmean',\rvparamothernotation' }}\bra{h} \ \mathrm{d} h \ \mathrm{d} \rvmean\\
    &= \sup_{\secretestimatenotation} \int_{
    \rvmeanlowerbound+ i\cdot \seclen - \privacythreshold
    }^{\rvmeanlowerbound + \bra{i+1}\cdot \seclen + \privacythreshold
    }
    \pdfof{\secretestimateof{\rvmean',\rvparamothernotation'}}(h)
    \cdot 
    \int_{\secretestimateof{\pdfof{\rvreleasewithparam{\rvmean',\rvparamothernotation' }}}-\privacythreshold}^{\secretestimateof{\pdfof{\rvreleasewithparam{\rvmean',\rvparamothernotation' }}}+\privacythreshold}
    \pdfof{U|U' }\bra{\rvmean|\rvmean'}
    \ \mathrm{d} \rvmean \ \mathrm{d} h\\
    &\leq \sup_{\secretestimatenotation} \int_{
    \rvmeanlowerbound+ i\cdot \seclen - \privacythreshold
    }^{\rvmeanlowerbound + \bra{i+1}\cdot \seclen + \privacythreshold
    }
    \frac{2\privacythreshold}{\seclen}\cdot
    \pdfof{\secretestimateof{\rvmean',\rvparamothernotation'}}(h)
    \ \mathrm{d} h\\
    &\leq \frac{2\privacythreshold}{\seclen}.
\end{aligned}
\end{equation}

Therefore, we have 
\begin{align*}
    \privacynotation{} 
    &= \sup_{\secretestimatenotation} \probof{ \secretestimateof{\releaservparamnotation}\in\brb{ \secretofparam - \privacythreshold, \secretofparam + \privacythreshold } }\\
    &=\sup_{\secretestimatenotation} \expectationof{\probof{ \secretestimateof{\releaservparamnotation}\in\brb{ \secretofparam - \privacythreshold, \secretofparam + \privacythreshold }\bigg| \releaservparamnotation } }\\
    &= \expectationof{\sup_{\secretestimatenotation}\probof{ \secretestimateof{\releaservparamnotation}\in\brb{ \secretofparam - \privacythreshold, \secretofparam + \privacythreshold }\bigg| \releaservparamnotation } }\\
    &\leq \frac{2\privacythreshold}{\seclen}.
\end{align*}

For the \distortion{}, we can easily get that $\distortionnotation = \frac{\seclen}{2}$.
 According to \cref{thm:continuous_mean}, we have $\distortionnotation_{\text{opt}}>\bra{\ceil{\frac{1}{\privacynotation{}}}-1} \privacythreshold \geq \privacythreshold$. We can get that
    \begin{align*}
\distortionnotation & = \distortionnotation_{\text{opt}} + \distortionnotation - \distortionnotation_{\text{opt}}\\
& < \distortionnotation_{\text{opt}} + \distortionnotation - \bra{\ceil{\frac{1}{\privacynotation}}-1}\cdot \privacythreshold\\
&\leq \distortionnotation_{\text{opt}} + \privacythreshold + \distortionnotation - {{\frac{\privacythreshold}{\privacynotation}}}\\
&\leq \distortionnotation_{\text{opt}} + \privacythreshold\\
& \leq 2\distortionnotation_{\text{opt}}.
    \end{align*}
\end{proof}
\subsection{Proof of %
\cref{thm:lowerbound_continuous_quantile}}
\label{sec:proof_lowerbound_continuous_quantile}

\subsubsection{Exponential Distribution}
\begin{proof}
Let $\rvprivatewithparam{\lambda_1},\rvprivatewithparam{\lambda_2}$ be two exponential random variables.
We have
\begin{align*}
\frac{\auxdistance{\rvprivatewithparam{\lambda_1}}{\rvprivatewithparam{\lambda_2}}}{\auxrange{\rvprivatewithparam{\lambda_1}}{\rvprivatewithparam{\lambda_2}}} &= \frac{\frac{1}{2}\bra{{
\lambda_1}-{\lambda_2}}}{-\ln\bra{1-\alpha}\bra{{
\lambda_1}-{\lambda_2}}}
=-\frac{1}{2\ln\bra{1-\alpha}}.\numberthis \label{eq:lowerbound_exp_quantile_d_r}
\end{align*}
Therefore we can get that
\begin{align*}
\ratio= -\frac{1}{2\ln\bra{1-\alpha}}.
\end{align*}
\end{proof}

\subsubsection{Shifted Exponential Distribution}

\begin{proof}
Let $\rvprivatewithparam{\lambda_1,h_1}, \rvprivatewithparam{\lambda_2,h_2}$ be random variables from shifted exponential distributions.
Let $\lambda_2\leq \lambda_1$ without loss of generality.
Let $a=\frac{\lambda_1}{\lambda_2}$ and $b=\bra{h_1/\lambda_1-h_2/\lambda_2}{\lambda_2}$.
We can get that $\pdfof{\rvprivatewithparam{\lambda_1,h_1}}\bra{x}=a\pdfof{\rvprivatewithparam{\lambda_2,h_2}}\bra{a\bra{x+b}}$, and 
\begin{align*}
&\auxdistance{\rvprivatewithparam{\lambda_1,h_1}}{\rvprivatewithparam{\lambda_2,h_2}}= 
\distanceformulawass{\rvprivatewithparam{\lambda_1,h_1}}{\rvprivatewithparam{\lambda_2,h_2}}\\
&= \frac{1}{2}\int_{h_1}^{+\infty}\brd{x-\bra{\frac{x}{a}-b}}\pdfof{\rvprivatewithparam{\lambda_1,h_1}}\bra{x}\mathrm{d}x\\
&=\frac{\lambda_2}{2\lambda_1}\int_{h_1}^{+\infty}\brd{\bra{1/\lambda_2-1/\lambda_1}x+h_1/\lambda_1-h_2/\lambda_2}e^{-\frac{1}{\lambda_1} \bra{x-h_1}}\mathrm{d}x\\
&=\begin{cases}
\frac{1}{2}\bra{h_2-h_1+{\lambda_2}-{\lambda_1}}-e^{\frac{h_2-h_1}{{
\lambda_2}-{\lambda_1}}}\bra{{
\lambda_2}-{\lambda_1}} & \bra{h_1<h_2}\\
\frac{1}{2}\bra{h_1-h_2+{
\lambda_1}-{\lambda_2}} & \bra{h_1\geq h_2}
\end{cases},
\numberthis\label{eq:lowerbound_shifted_exp_quantile_d}
\end{align*}
\begin{align*}
\auxrange{\rvprivatewithparam{\lambda_1,h_1}}{\rvprivatewithparam{\lambda_2,h_2}}
= \brd{\ln\bra{1-\alpha}\bra{{\lambda_1}-{\lambda_2}}+h_2-h_1}.
\end{align*}

When $h_1< h_2$, let $t = \frac{h_2-h_1}{{
\lambda_1}-{\lambda_2}}\in(0,+\infty)$. 
We have
\begin{align*}
&\quad\frac{\auxdistance{\rvprivatewithparam{\lambda_1,h_1}}{\rvprivatewithparam{\lambda_2,h_2}}}{\auxrange{\rvprivatewithparam{\lambda_1,h_1}}{\rvprivatewithparam{\lambda_2,h_2}}} \\
&= \frac{h_2-h_1+{\lambda_2}-{\lambda_1}-2e^{\frac{h_2-h_1}{{
\lambda_2}-{\lambda_1}}}\bra{{
\lambda_2}-{\lambda_1}}}{2\left|\ln\bra{1-\alpha}\bra{{\lambda_1}-{\lambda_2}}+h_2-h_1\right|}\\
& = \frac{t+2e^{-t}-1}{2\left|\ln\bra{1-\alpha}+t\right|}\\
& \geq 
\begin{cases}
\frac{1}{2}\babs{1+\frac{\ln(1-\alpha)+1}{W_{-1}\bra{-\frac{\ln\bra{1-\alpha}+1}{2\bra{1-\alpha}e}}}} & \alpha \in [0, 1-e^{-1})\\
\frac{1}{2}\babs{1+\frac{\ln(1-\alpha)+1}{W_{0}\bra{-\frac{\ln\bra{1-\alpha}+1}{2\bra{1-\alpha}e}}}} & \alpha \in [1-e^{-1},1)
\end{cases},
\end{align*}
where $W_{-1}$ and $W_{0}$ are Lambert $W$ functions. ``$=$'' achieves when
\begin{align*}
t = t_0 \triangleq
\begin{cases}
-1-\ln\bra{1-\alpha}-W_{-1}\bra{-\frac{\ln\bra{1-\alpha}+1}{2\bra{1-\alpha}e}} & \bra{\alpha \in [0, 1-e^{-1})}\\
-1-\ln\bra{1-\alpha}-W_{0}\bra{-\frac{\ln\bra{1-\alpha}+1}{2\bra{1-\alpha}e}} & \bra{\alpha \in [1-e^{-1},1)}
\end{cases}.
\end{align*} 

When $h_1\geq h_2$, let $t = \frac{h_1-h_2}{{
\lambda_1}-{\lambda_2}}\in(0,+\infty)$. We have
\begin{align*}
\frac{\auxdistance{\rvprivatewithparam{\lambda_1,h_1}}{\rvprivatewithparam{\lambda_2,h_2}}}{\auxrange{\rvprivatewithparam{\lambda_1,h_1}}{\rvprivatewithparam{\lambda_2,h_2}}} &= \frac{h_1-h_2+{
\lambda_1}-{\lambda_2}}{2\babs{\ln\bra{1-\alpha}\bra{{\lambda_1}-{\lambda_2}}+h_2-h_1}}\\
&= \frac{t+1}{2\babs{\ln\bra{1-\alpha}-t}}\\
&\geq \min\brc{\frac{1}{2}, -\frac{1}{2\ln\bra{1-\alpha}}}.
\end{align*}

Therefore we can get that
\begin{align*}
\ratio= \begin{cases}
\frac{1}{2}\babs{1+\frac{\ln(1-\alpha)+1}{W_{-1}\bra{-\frac{\ln\bra{1-\alpha}+1}{2\bra{1-\alpha}e}}}} & \alpha \in [0, 1-e^{-1})\\
\frac{1}{2}\babs{1+\frac{\ln(1-\alpha)+1}{W_{0}\bra{-\frac{\ln\bra{1-\alpha}+1}{2\bra{1-\alpha}e}}}} & \alpha \in [1-e^{-1},1)
\end{cases}.
\end{align*}
\end{proof}
\subsection{Proof of %
\cref{thm:upperbound_continuous_quantile}}
\label{sec:proof_upperbound_continuous_quantile}

\subsubsection{Exponential Distribution}
\label{sec:proof_upperbound_continuous_quantile_EXP}
\begin{proof}
The proof of $\distortionnotation$ and $\privacynotation$ is the same as \cref{sec:proof_trade-off_mechanism_mean_continuous}, except that we use the $\auxdistance{\cdot}{\cdot}$ and $\auxrange{\cdot}{\cdot}$ from \cref{eq:lowerbound_exp_quantile_d_r}.

    For $\distortionnotation_{\text{opt}}$, we have $\distortionnotation_{\text{opt}}>\bra{\ceil{\frac{1}{\privacynotation{}}}-1}\cdot 2\ratio\privacythreshold\geq 2\ratio\privacythreshold$, where $\ratio = -\frac{1}{2\ln(1-\alpha)}$. We can get that
    \begin{align*}
\distortionnotation & = \distortionnotation_{\text{opt}} + \distortionnotation - \distortionnotation_{\text{opt}}\\
& < \distortionnotation_{\text{opt}} + \distortionnotation - \bra{\ceil{\frac{1}{\privacynotation}}-1}\cdot 2\ratio\privacythreshold\\
&\leq \distortionnotation_{\text{opt}} + 2\ratio\privacythreshold + \distortionnotation - {{\frac{2\ratio\privacythreshold}{\privacynotation}}}\\
&= \distortionnotation_{\text{opt}} + 2\ratio\privacythreshold\\
& \leq 2\distortionnotation_{\text{opt}}.
    \end{align*}
\end{proof}

\subsubsection{Shifted Exponential Distribution}
\label{sec:proof_upperbound_continuous_quantile_SFHITED_EXP}
\begin{proof}

We first focus on the proof for $\privacynotation$.

\begin{figure}[h]
    \centering
    \includegraphics[width=0.8\linewidth]{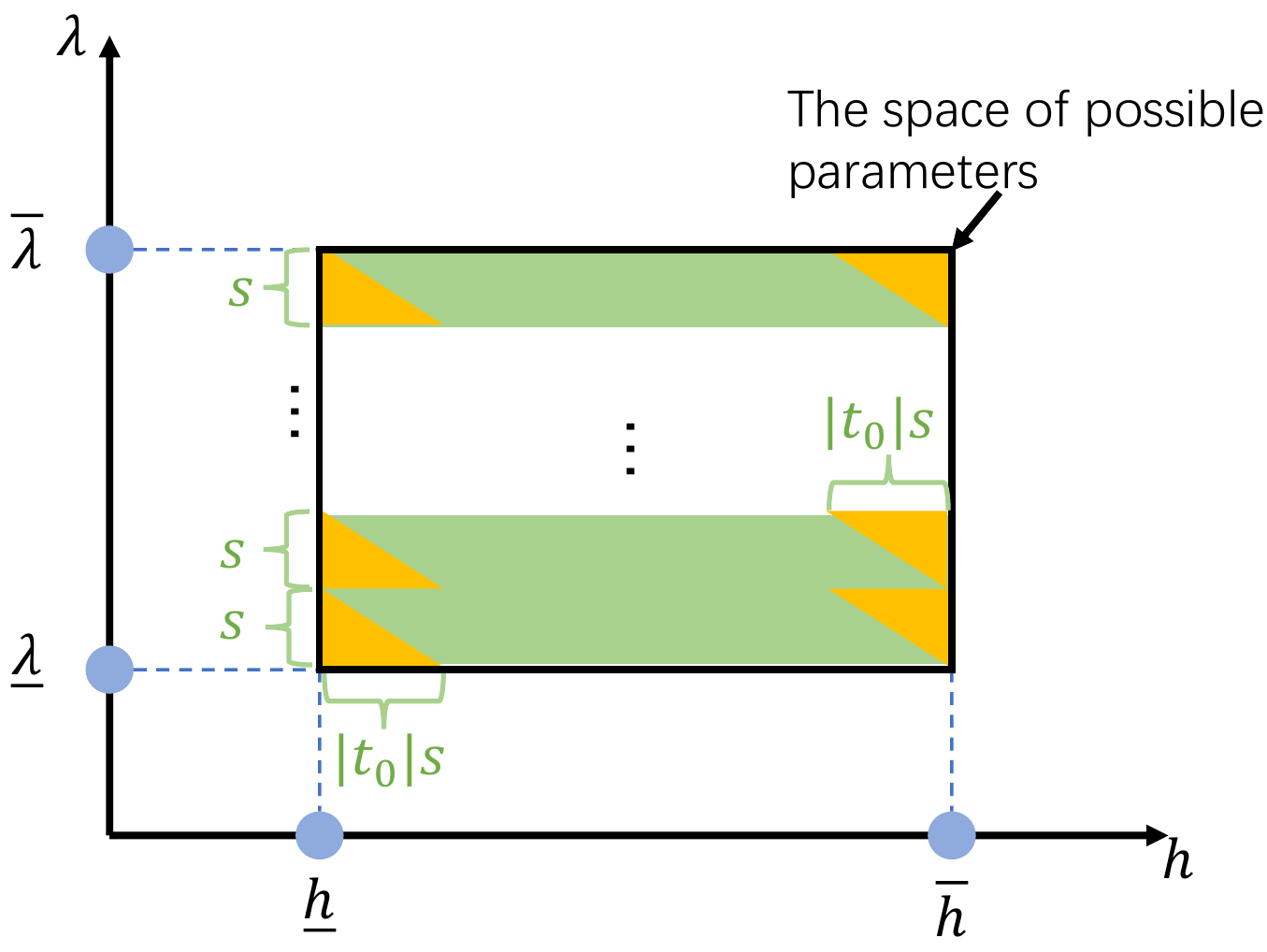}
    \caption{The construction for proof of \cref{thm:upperbound_continuous_quantile} for shifted exponential distributions. We separate the space of possible parameters into two regions (yellow and green) and bound the attacker's success rate on each region separately.}
    \label{fig:upperbound_shifted_exp_quantile}
\end{figure}

In \cref{fig:upperbound_shifted_exp_quantile}, we separate the space of possible data parameters into two regions represented by yellow and green colors. The yellow regions $S_{yellow}$ constitute right triangles with height $\seclen$ and width $\abs{t_0}\seclen$.  The green region $S_{green}$ is the rest of the parameter space. 
The high-level idea of our proof is as follows.
Note that for any parameter $\rvparamnotation\in S_{green}$, there exists a $\subsetofprivateparamof{\privateparamindexnotation,h}$ s.t. $\rvparamnotation \in \subsetofprivateparamof{\privateparamindexnotation,h}$ and $\subsetofprivateparamof{\mu,\privateparamindexnotation}\subset S_{green}$. Therefore, we can bound the attack success rate if  $\rvparamnotation\in S_{green}$. At the same time, the probability of $\rvparamnotation\in S_{yellow}$ is bounded.
Therefore, we can bound the overall attacker's success rate (i.e., $\privacynotation$). More specifically, let the optimal attacker be $\secretestimatestarnotation$. We have
\begin{align*}
    \privacynotation{} 
    &= \probof{ \secretestimatestarof{\releaservparamnotation}\in\brb{ \secretofparam - \privacythreshold, \secretofparam + \privacythreshold } }\\
    &= \int_{ \rvparamnotation\in S_{green}}p(\rvparamnotation)\probof{ \secretestimatestarof{\releaservparamnotation}\in\brb{ \secretofparam - \privacythreshold, \secretofparam + \privacythreshold } }d\rvparamnotation \\
    &\quad+  \int_{ \rvparamnotation\in S_{yellow}}p(\rvparamnotation)\probof{ \secretestimatestarof{\releaservparamnotation}\in\brb{ \secretofparam - \privacythreshold, \secretofparam + \privacythreshold }} d\rvparamnotation\\
    &<\frac{2\privacythreshold}{\abs{\ln\bra{1-\alpha}+t_0}\seclen} + \frac{\abs{t_0}\seclen}{\hupperbound-\hlowerbound}.
\end{align*}

For the \distortion{},
it is straightforward to get that $\distortionnotation = \frac{\seclen}{2}\bra{t_0-1}+\seclen e^{-t_0}$ from \cref{eq:lowerbound_shifted_exp_quantile_d}, and $\distortionnotation_{\text{opt}}>\bra{\ceil{\frac{1}{\privacynotation{}}}-1}\cdot 2\ratio\privacythreshold \geq 2\ratio\privacythreshold$, where $\ratio$ is defined in \cref{thm:lowerbound_continuous_quantile}.
Denote $\zeta = \frac{2\privacythreshold}{\abs{\ln\bra{1-\alpha}+t_0}\seclen} + \frac{\abs{t_0}\seclen}{\hupperbound-\hlowerbound}-\privacynotation$,
we can get that $\bra{\privacynotation + \zeta - \frac{\abs{t_0}\seclen}{\hupperbound-\hlowerbound}} \cdot \distortionnotation = 2\ratio\privacythreshold$ and 
\begin{align*}
\distortionnotation & = \distortionnotation_{\text{opt}} + \distortionnotation - \distortionnotation_{\text{opt}}\\
& < \distortionnotation_{\text{opt}} + \distortionnotation - \bra{\ceil{\frac{1}{\privacynotation}}-1}\cdot 2\ratio\privacythreshold\\
&\leq \distortionnotation_{\text{opt}} + 2\ratio\privacythreshold + \distortionnotation - {{\frac{2\ratio\privacythreshold}{\privacynotation}}}\\
&=\distortionnotation_{\text{opt}} + 2\ratio\privacythreshold + \frac{\frac{\abs{t_0}\seclen}{\hupperbound-\hlowerbound}-\zeta}{\frac{2\privacythreshold}{\abs{\ln\bra{1-\alpha}+t_0}\seclen} + \frac{\abs{t_0}\seclen}{\hupperbound-\hlowerbound}-\zeta}\cdot\distortionnotation\\
&<\distortionnotation_{\text{opt}} + 2\ratio\privacythreshold + \frac{\frac{\abs{t_0}\seclen}{\hupperbound-\hlowerbound}}{\frac{2\privacythreshold}{\abs{\ln\bra{1-\alpha}+t_0}\seclen} + \frac{\abs{t_0}\seclen}{\hupperbound-\hlowerbound}}\cdot\distortionnotation.
\end{align*}
Therefore, 
\begin{align*}
\distortionnotation&<\bra{1+\frac{\abs{t_0}\cdot\abs{\ln\bra{1-\alpha}+t_0}\seclen^2}{2\epsilon\bra{\hupperbound-\hlowerbound}}} \bra{\distortionnotation_{\text{opt}} + 2\ratio\privacythreshold } \\
&\leq \bra{2+\frac{\abs{t_0}\cdot\abs{\ln\bra{1-\alpha}+t_0}\seclen^2}{\epsilon\bra{\hupperbound-\hlowerbound}}} \distortionnotation_{\text{opt}}.
\end{align*}

$t_0$ is bounded when $\alpha \in \brb{0, c_1}\cup \brb{1-\frac{1}{e}, c_2}$, where $c_1\in \brba{0, 1-\frac{1}{e}}, c_2\in \brba{1-\frac{1}{e}, 1}
$. Therefore, when $\alpha \in [0.01, 0.25]\cup [0.75, 0.99]$, we can get that
\begin{align*}
    \lim{\sup_{{\frac{\seclen^2}{\hupperbound-\hlowerbound}}\rightarrow 0}} \distortionnotation 
    < \lim{\sup_{{\frac{\seclen^2}{\hupperbound-\hlowerbound}}\rightarrow 0}}\bra{2+\frac{\abs{t_0}\cdot\abs{\ln\bra{1-\alpha}+t_0}\seclen^2}{\epsilon\bra{\hupperbound-\hlowerbound}}} \distortionnotation_{\text{opt}}
    < 3  \distortionnotation_{\text{opt}}.
\end{align*}

\end{proof}
\subsection{Proofs for the Surrogate Metrics}
\label{sec:proof_lowerbound_surrogate}

\textit{Secret=Mean:}
For any $p_\releasedataset$, we have
\begin{align*}
\tilde{\distortionnotation}=\wassersteinof{p_\privatedataset}{p_\releasedataset} \geq \brd{\frac{1}{n} \sum_{i=1}^{n}x_i - \frac{1}{n} \sum_{i=1}^{n}y_i}=-\tilde \Pi_\epsilon.
\end{align*}

For  $p_\releasedataset$ released from our mechanism (\cref{sec:case_study_extend_dataset}), we have \begin{align*}
\tilde{\distortionnotation}=\wassersteinof{p_\privatedataset}{p_\releasedataset} = \brd{\frac{1}{n} \sum_{i=1}^{n}x_i - \frac{1}{n} \sum_{i=1}^{n}y_i}=-\tilde \Pi_\epsilon.
\end{align*}

\section{\Privacy{}-\Distortion{} Performance of  \DataMechanism{} with Relaxed Assumption}

\subsection{\Privacy{}-\Distortion{} Performance of \cref{mech:mean_continuous} with Relaxed Assumption}
\label{sec:proof_trade-off_mechanism_mean_continuous_relaxed}

We relax \cref{assu:mean_continuous} as follows.
\begin{assumption}
\label{assu:mean_continuous_relaxed}
The distribution parameter vector $\rvparamnotation$ can be written as $(\rvmean, \rvparamothernotation)$, where $\rvmean\in \setofreal$, $\rvparamothernotation\in \setofreal^{\paramdim-1}$, and for any $\rvmean\neq \rvmean'$, $\pdfof{\rvprivatewithparam{\rvmean,\rvparamothernotation}}\bra{x}=\pdfof{\rvprivatewithparam{\rvmean',\rvparamothernotation}}\bra{x-\rvmean'+\rvmean}$.
The prior over distribution parameters is $\pdfof{\RVmean,\RVparamothernotation}\bra{a,b}=\pdfof{\RVmean}\bra{a}\cdot \pdfof{\RVparamothernotation}\bra{b}$, where $\support{{\RVmean}} = \brba{\rvmeanlowerbound, \rvmeanupperbound}$, and $\pdfof{\RVmean}$ is $\mathcal{L}$-Lipschitz continuous and has lower bound $\underline{c}$.
\end{assumption}

Based on \cref{assu:mean_continuous_relaxed}, the \Privacy{}-\distortion{} performance of \cref{mech:mean_continuous} is shown below.

\begin{proposition}
    \label{thm:trade-off_mechanism_mean_continuous_relax}
    Under \cref{assu:mean_continuous_relaxed},
    \cref{mech:mean_continuous} %
    has $\distortionnotation = \frac{\seclen}{2}$ and $\privacynotation{}\leq \frac{2\privacythreshold\brb{\underline{c}+\mathcal{L}\bra{\seclen-x^*-\privacythreshold}}}{\underline{c}\seclen+\frac{\mathcal{L}}{2}\bra{\seclen-x^*}^2}$, where $x^* = \seclen + \frac{\underline{c}}{\mathcal{L}} - \privacythreshold - \sqrt{\bra{\frac{\underline{c}}{\mathcal{L}} - \privacythreshold}^2 + \frac{2\underline{c}\seclen}{\mathcal{L}}}$.
\end{proposition}

\begin{proof}
We first provide the following lemma.
\begin{lemma}
\label{lemma:max_fraction_Lipschitz_mean}
For a $\mathcal{L}$-Lipschitz continuous function $f(x), x\in \brb{\underline{x}, \overline{x}}$, $\inf_{x\in \brb{\underline{x}, \overline{x}}} f(x) \geq \underline{c}\geq 0$, it satisfies
$$
\sup_{x'\in\brb{\underline{x}, \overline{x}-\delta}} \frac{\int_{x'}^{x'+\delta}f(x)\mathrm{d}x}{\int_{\underline{x}}^{\overline{x}}f(x)\mathrm{d}x} \leq 
\frac{\delta\brb{\underline{c}+\mathcal{L}\bra{\overline{x}-x^*-\frac{\delta}{2}}}}{\underline{c}\bra{\overline{x}-\underline{x}}+\frac{\mathcal{L}}{2}\bra{\overline{x}-x^*}^2},
$$
where $x^* = \overline{x} + \frac{\underline{c}}{\mathcal{L}} - \frac{\delta}{2} - \sqrt{\bra{\frac{\underline{c}}{\mathcal{L}} - \frac{\delta}{2}}^2 + \frac{2\underline{c}\bra{\overline{x}-\underline{x}}}{\mathcal{L}}}$.
\end{lemma}

For any released parameter $\releaservparamnotation = (\rvmean', \rvparamothernotation')$,
there exists $i \in \brc{0,...,N-1}$
such that
$\rvmean' =\rvmeanlowerbound + \bra{i+0.5}\cdot \seclen$. We have
\begin{equation}
\begin{aligned}
\nonumber
\sup_{\secretestimatenotation} & \  \probof{ \secretestimateof{\releaservparamnotation}\in\brb{ \secretofparam - \privacythreshold, \secretofparam + \privacythreshold } \big| \releaservparamnotation }\\
    &= \sup_{\secretestimatenotation} \int_{\rvmeanlowerbound+ i\cdot \seclen}^{\rvmeanlowerbound + \bra{i+1}\cdot \seclen} \pdfof{U|U' }\bra{\rvmean|\rvmean'} \cdot 
    \int_{\rvmean-\privacythreshold}^{\rvmean+\privacythreshold} %
    \pdfof{\secretestimateof{\rvmean',\rvparamothernotation' }}\bra{h} \ \mathrm{d} h \ \mathrm{d} \rvmean\\
    &= \sup_{\secretestimatenotation} \int_{
    \rvmeanlowerbound+ i\cdot \seclen - \privacythreshold
    }^{\rvmeanlowerbound + \bra{i+1}\cdot \seclen + \privacythreshold
    }
    \pdfof{\secretestimateof{\rvmean',\rvparamothernotation'}}(h)
    \cdot 
    \int_{\secretestimateof{\pdfof{\rvreleasewithparam{\rvmean',\rvparamothernotation' }}}-\privacythreshold}^{\secretestimateof{\pdfof{\rvreleasewithparam{\rvmean',\rvparamothernotation' }}}+\privacythreshold}
    \pdfof{U|U' }\bra{\rvmean|\rvmean'}
    \ \mathrm{d} \rvmean \ \mathrm{d} h.
\end{aligned}
\end{equation}

For $\int_{\secretestimateof{\pdfof{\rvreleasewithparam{\rvmean',\rvparamothernotation' }}}-\privacythreshold}^{\secretestimateof{\pdfof{\rvreleasewithparam{\rvmean',\rvparamothernotation' }}}+\privacythreshold}
    \pdfof{U|U' }\bra{\rvmean|\rvmean'}
    \ \mathrm{d} \rvmean$, denote
\begin{align*}
    &x_1 = \max\bra{0, \secretestimateof{\pdfof{\rvreleasewithparam{\rvmean',\rvparamothernotation' }}}-\privacythreshold-\rvmeanlowerbound- i\cdot \seclen}, \\
    &x_2 = \min\bra{ \secretestimateof{\pdfof{\rvreleasewithparam{\rvmean',\rvparamothernotation' }}}+\privacythreshold-\rvmeanlowerbound- i\cdot \seclen, \seclen},
\end{align*}
 we have
\begin{align*}
\int_{\secretestimateof{\pdfof{\rvreleasewithparam{\rvmean',\rvparamothernotation' }}}-\privacythreshold}^{\secretestimateof{\pdfof{\rvreleasewithparam{\rvmean',\rvparamothernotation' }}}+\privacythreshold}
    \pdfof{U|U' }\bra{\rvmean|\rvmean'}
    \ \mathrm{d} \rvmean
= \frac{\int_{x_1}^{x_2}\pdfof{U }\bra{\rvmeanlowerbound+ i\cdot \seclen + x}
    \ \mathrm{d} x}
{\int_{
    0
    }^{\seclen
    }\pdfof{U }\bra{\rvmeanlowerbound+ i\cdot \seclen + x}
    \ \mathrm{d} x}.
\end{align*}

$\pdfof{U }\bra{\rvmeanlowerbound+ i\cdot \seclen + x}$ is $\mathcal{L}$-Lipschitz and has lower bound $\underline{c}$. $x_2 - x_1 \leq 2\privacythreshold$ and $x_1, x_2 \in \brb{0, \seclen}$. According to \cref{lemma:max_fraction_Lipschitz_mean}, we have
\begin{align*}
\int_{\secretestimateof{\pdfof{\rvreleasewithparam{\rvmean',\rvparamothernotation' }}}-\privacythreshold}^{\secretestimateof{\pdfof{\rvreleasewithparam{\rvmean',\rvparamothernotation' }}}+\privacythreshold}
    \pdfof{U|U' }\bra{\rvmean|\rvmean'}
    \ \mathrm{d} \rvmean
&= \frac{\int_{x_1}^{x_2}\pdfof{U }\bra{\rvmeanlowerbound+ i\cdot \seclen + x}
    \ \mathrm{d} x}
{\int_{
    0
    }^{\seclen
    }\pdfof{U }\bra{\rvmeanlowerbound+ i\cdot \seclen + x}
    \ \mathrm{d} x}\\
&\leq \frac{2\privacythreshold\brb{\underline{c}+\mathcal{L}\bra{\seclen-x^*-\privacythreshold}}}{\underline{c}\seclen+\frac{\mathcal{L}}{2}\bra{\seclen-x^*}^2},
\end{align*}
where $x^* = \seclen + \frac{\underline{c}}{\mathcal{L}} - \privacythreshold - \sqrt{\bra{\frac{\underline{c}}{\mathcal{L}} - \privacythreshold}^2 + \frac{2\underline{c}\seclen}{\mathcal{L}}}$.

Therefore, we can get that 
\begin{align*}
    \sup_{\secretestimatenotation} \   &\probof{ \secretestimateof{\releaservparamnotation}\in\brb{ \secretofparam - \privacythreshold, \secretofparam + \privacythreshold } \big| \releaservparamnotation }\\
    &\leq \sup_{\secretestimatenotation} \int_{
    \rvmeanlowerbound+ i\cdot \seclen - \privacythreshold
    }^{\rvmeanlowerbound + \bra{i+1}\cdot \seclen + \privacythreshold
    }
    \frac{2\privacythreshold\brb{\underline{c}+\mathcal{L}\bra{\seclen-x^*-\privacythreshold}}}{\underline{c}\seclen+\frac{\mathcal{L}}{2}\bra{\seclen-x^*}^2}\cdot
    \pdfof{\secretestimateof{\rvmean',\rvparamothernotation'}}(h)
    \ \mathrm{d} h\\
    &\leq \frac{2\privacythreshold\brb{\underline{c}+\mathcal{L}\bra{\seclen-x^*-\privacythreshold}}}{\underline{c}\seclen+\frac{\mathcal{L}}{2}\bra{\seclen-x^*}^2}.
\end{align*}

Therefore, we have 
\begin{align*}
    \privacynotation{} 
    &= \sup_{\secretestimatenotation} \probof{ \secretestimateof{\releaservparamnotation}\in\brb{ \secretofparam - \privacythreshold, \secretofparam + \privacythreshold } }\\
    &=\sup_{\secretestimatenotation} \expectationof{\probof{ \secretestimateof{\releaservparamnotation}\in\brb{ \secretofparam - \privacythreshold, \secretofparam + \privacythreshold }\bigg| \releaservparamnotation } }\\
    &= \expectationof{\sup_{\secretestimatenotation}\probof{ \secretestimateof{\releaservparamnotation}\in\brb{ \secretofparam - \privacythreshold, \secretofparam + \privacythreshold }\bigg| \releaservparamnotation } }\\
    &\leq \frac{2\privacythreshold\brb{\underline{c}+\mathcal{L}\bra{\seclen-x^*-\privacythreshold}}}{\underline{c}\seclen+\frac{\mathcal{L}}{2}\bra{\seclen-x^*}^2}.
\end{align*}

For the \distortion{}, we can easily get that $\distortionnotation = \frac{\seclen}{2}$.
\end{proof}

\subsubsection{Proof of \cref{lemma:max_fraction_Lipschitz_mean}}
\label{sec:proof_of_max_fraction_lemma}

Without loss of generality, we assume that $f(\overline{x}) \geq f(\underline{x})$. 
Based on simple geometric analysis, we can get that when $\frac{\int_{x'}^{x'+\delta}f(x)\mathrm{d}x}{\int_{\underline{x}}^{\overline{x}}f(x)\mathrm{d}x}$ achieves supremum, as illustrated in \cref{fig:fraction_derivation_illustration}, $f(\underline{x}) = \underline{c}$, $x' = \overline{x}-\delta$, and $f(\overline{x}) = \underline{x}+\mathcal{L}\bra{\overline{x}-x''}$, where $x''\in \brb{\underline{x}, x'}$.

\begin{figure}[htbp]
    \centering
    \includegraphics[width=0.6\linewidth]{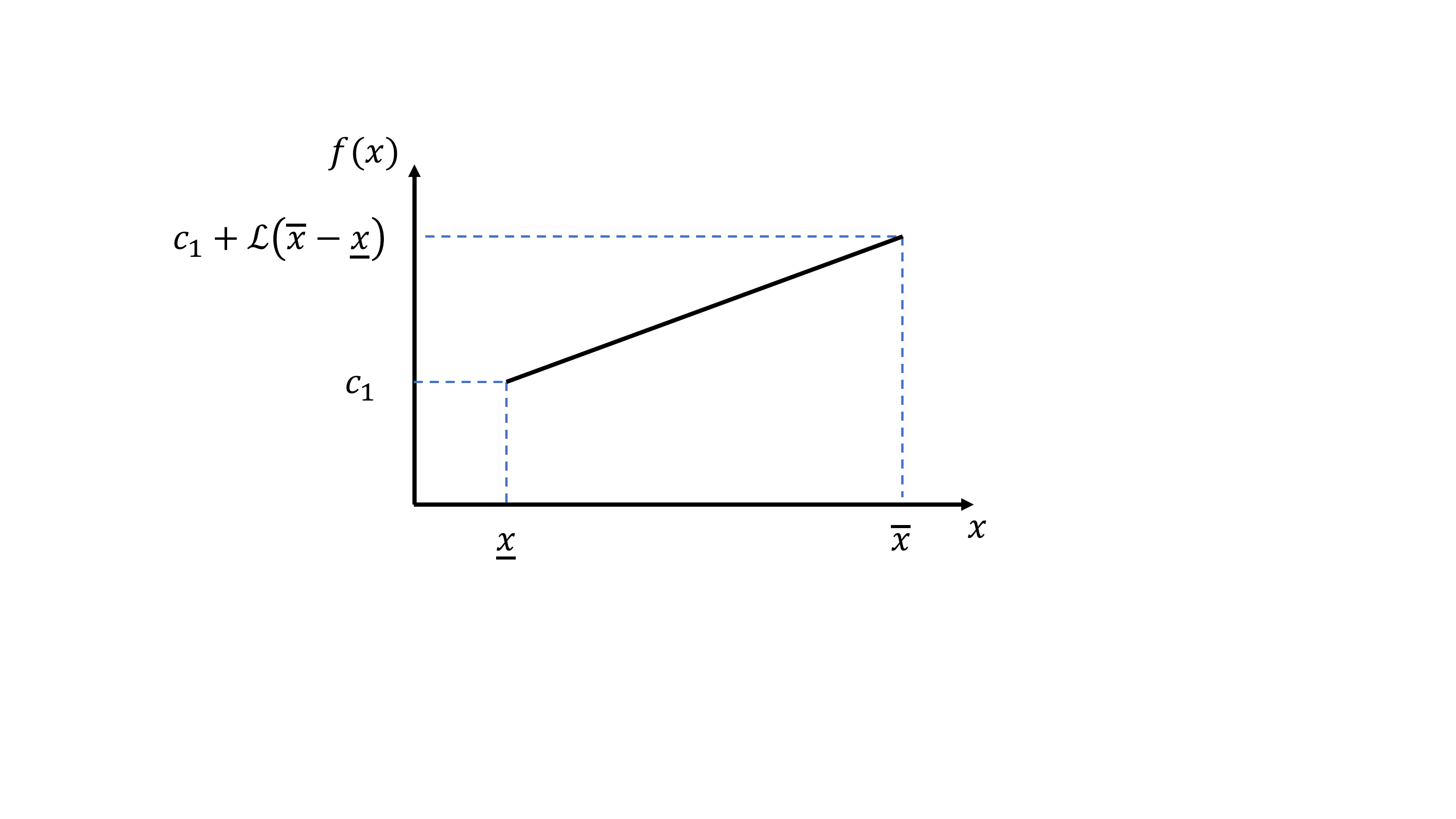}
    \caption{Illustration of $f(x)$ when $\frac{\int_{x'}^{x'+\delta}f(x)\mathrm{d}x}{\int_{\underline{x}}^{\overline{x}}f(x)\mathrm{d}x}$ achieves supremum.}
    \label{fig:fraction_derivation_illustration}
\end{figure}

In this case, we can get that 
$$
\frac{\int_{\overline{x}-\delta}^{\overline{x}}f(x)\mathrm{d}x}{\int_{\underline{x}}^{\overline{x}}f(x)\mathrm{d}x} = 
\frac{\delta\brb{\underline{c}+\mathcal{L}\bra{\overline{x}-x''-\frac{\delta}{2}}}}{\underline{c}\bra{\overline{x}-\underline{x}}+\frac{\mathcal{L}}{2}\bra{\overline{x}-x''}^2}
\triangleq
h\bra{x''},
$$
where $x''\in \brb{\underline{x}, x'}$. When $x'' =  \overline{x} + \frac{\underline{c}}{\mathcal{L}} - \frac{\delta}{2} - \sqrt{\bra{\frac{\underline{c}}{\mathcal{L}} - \frac{\delta}{2}}^2 + \frac{2\underline{c}\bra{\overline{x}-\underline{x}}}{\mathcal{L}}} \triangleq x^*$, $h(x'')$ achieves supremum. Therefore, we have
\begin{align*}
\sup_{x'\in\brb{\underline{x}, \overline{x}-\delta}} \frac{\int_{x'}^{x'+\delta}f(x)\mathrm{d}x}{\int_{\underline{x}}^{\overline{x}}f(x)\mathrm{d}x}
&\leq 
\sup_f \sup_{x'\in\brb{\underline{x}, \overline{x}-\delta}} \frac{\int_{x'}^{x'+\delta}f(x)\mathrm{d}x}{\int_{\underline{x}}^{\overline{x}}f(x)\mathrm{d}x}\\
&= 
\frac{\delta\brb{\underline{c}+\mathcal{L}\bra{\overline{x}-x^*-\frac{\delta}{2}}}}{\underline{c}\bra{\overline{x}-\underline{x}}+\frac{\mathcal{L}}{2}\bra{\overline{x}-x^*}^2}.
\end{align*}

\subsection{\Privacy{}-\Distortion{} Performance of \cref{mech:quantile_continuous} with Relaxed Assumption}
\label{sec:proof_trade-off_mechanism_quantile_continuous_relaxed}

We relax \cref{assu:quantile_continuous} as follows.
\begin{assumption}
    \label{assu:quantile_continuous_relaxed}
    The prior over distribution parameters as specified below.
    \begin{packeditemize}
        \item Exponential: $\support{{\lambda}} = \brba{\lambdalowerbound, \lambdaupperbound}$, and $\pdfof{\lambda}$ is $\mathcal{L}$-Lipschitz continuous and has lower bound $\underline{c}$.
        \item Shifted exponential: $\support{{\lambda,h}} = \brc{(a,b)| a\in\brba{\lambdalowerbound, \lambdaupperbound},b\in\brba{\hlowerbound, \hupperbound}}$,
    $\pdfof{\lambda,h}\bra{a, b} = \pdfof{\lambda}\bra{a}\cdot \pdfof{h}\bra{b}$, and $\pdfof{\lambda}$ (resp. $\pdfof{h}$) is $\mathcal{L}_{\lambda}$-Lipschitz (resp. $\mathcal{L}_{h}$-Lipschitz) and has lower bound $\frac{k_{\lambda}}{\muupperbound-\mulowerbound}$ with $k_{\lambda} \in (0,1]$ (resp. $\frac{k_{h}}{\sigmaupperbound-\sigmalowerbound}$ with $k_{h} \in (0,1]$).
    \end{packeditemize}
\end{assumption}

Based on \cref{assu:quantile_continuous_relaxed}, the \Privacy{}-\distortion{} performance of \cref{mech:quantile_continuous} is shown below.

\begin{proposition}
    \label{thm:upperbound_continuous_quantile_relaxed}
    Under \cref{assu:quantile_continuous_relaxed},
    \cref{mech:quantile_continuous} %
    has the following $\distortionnotation$ and $\privacynotation{}$ value/bound.
    \begin{packeditemize}
        
        \item Exponential:
        \begin{align*}
            \distortionnotation&=\frac{1}{2}\seclen, \\\privacynotation&\leq \frac{\frac{2\privacythreshold}{-\ln\bra{1-\alpha}}\cdot\brb{\underline{c}+\mathcal{L}\bra{\seclen-x^*+\frac{\privacythreshold}{\ln\bra{1-\alpha}}}}}{\underline{c}\seclen+\frac{\mathcal{L}}{2}\bra{\seclen-x^*}^2},
\end{align*}
where $x^* = \seclen + \frac{\underline{c}}{\mathcal{L}} + \frac{\privacythreshold}{\ln\bra{1-\alpha}} - \sqrt{\bra{\frac{\underline{c}}{\mathcal{L}} + \frac{\privacythreshold}{\ln\bra{1-\alpha}}}^2 + \frac{2\underline{c}\seclen}{\mathcal{L}}}$.
        \item Shifted exponential:
        \begin{align*}
            \distortionnotation %
            & = \frac{\seclen}{2}\bra{t_0-1}+\seclen e^{-t_0},\\
            \privacynotation %
            & < \frac{\frac{2\privacythreshold}{\abs{\ln\bra{1-\alpha}+t_0}}
\cdot \brb{\underline{c}+\mathcal{L}_{\lambda, h}\bra{\frac{\seclen}{2}-t^*-\frac{\privacythreshold}{\abs{\ln\bra{1-\alpha}+t_0}}}}}{\underline{c}\seclen+\frac{\mathcal{L}_{\lambda, h}}{2}\bra{\frac{\seclen}{2}-t^*}^2}+\\ 
&  M\bra{ 
 \hupperbound-\hlowerbound, \frac{k_{h}}{\hupperbound-\hlowerbound}, \mathcal{L}_h, 1}\cdot M\bra{ \lambdaupperbound-\lambdalowerbound, \frac{k_{\lambda}}{\lambdaupperbound-\lambdalowerbound}, \mathcal{L}_{\lambda}, 1}\cdot
 \bra{\lambdaupperbound-\lambdalowerbound}\abs{t_0}\seclen,
        \end{align*}
where $\underline{c} = \frac{k_{h}k_{\lambda}}{\bra{\hupperbound-\hlowerbound}\cdot\bra{\lambdaupperbound-\lambdalowerbound}}$, function $M$ satisfies
$$
 M\bra{x,  c, \mathcal{L}, \mathcal{A}} = 
\begin{cases}
\frac{\mathcal{A}}{x} + \frac{\mathcal{L}x}{2}, & \text{if } c \leq \frac{\mathcal{A}}{x} - \frac{\mathcal{L}x}{2}\\
c + \sqrt{2\mathcal{L}\bra{\mathcal{A}-cx}}, & \text{if } c > \frac{\mathcal{A}}{x} - \frac{\mathcal{L}x}{2}
\end{cases},
$$ 
$
\mathcal{L}_{\lambda, h} = \mathcal{L}_{\lambda} M\bra{ \frac{\hupperbound-\hlowerbound}{\abs{t_0}}, \frac{k_{h}}{\hupperbound-\hlowerbound}, \abs{t_0}\mathcal{L}_h, \frac{1}{\abs{t_0}}} + \abs{t_0} \mathcal{L}_{h}  M\bra{ \lambdaupperbound-\lambdalowerbound, \frac{k_{\lambda}}{\lambdaupperbound-\lambdalowerbound}, \mathcal{L}_{\lambda}, 1}
$,
and \\
$t^* = \frac{\seclen}{2} + \frac{\underline{c}}{\mathcal{L}_{\lambda, h}} - \frac{\privacythreshold}{\abs{\ln\bra{1-\alpha}+t_0}} - \sqrt{\bra{\frac{\underline{c}}{\mathcal{L}_{\lambda, h}} - \frac{\privacythreshold}{\abs{\ln\bra{1-\alpha}+t_0}}}^2 + \frac{2\underline{c}\seclen}{\mathcal{L}_{\lambda, h}}}$.
    \end{packeditemize}
    The $t_0$ parameter is defined in \cref{mech:quantile_continuous}.
\end{proposition}

\subsubsection{Proof of \cref{thm:upperbound_continuous_quantile_relaxed} for Exponential Distribution}
\label{sec:upperbound_continuous_quantile_exponential_relaxed}

It is straightforward to get the formula for $\distortionnotation$ from \cref{eq:lowerbound_exp_quantile_d_r}. Here we focus on the proof for $\privacynotation$.

Similar to the proof in \cref{sec:proof_trade-off_mechanism_mean_continuous_relaxed}, according to \cref{lemma:max_fraction_Lipschitz_mean}, we have 
\begin{align*}
    \privacynotation{} 
    &= \expectationof{\sup_{\secretestimatenotation}\probof{ \secretestimateof{\releaservparamnotation}\in\brb{ \secretofparam - \privacythreshold, \secretofparam + \privacythreshold }\bigg| \releaservparamnotation } }\\
    &\leq \sup_{i\in \mathbb{N}, t'\in \mathbb{R}%
    }
    \frac{\int_{\max\brc{0, t'}}^{\min\brc{\seclen, t'-\frac{2\privacythreshold}{\ln\bra{1-\alpha}}}}\pdfof{\lambda}\bra{\lambdalowerbound+i\cdot \seclen +t}\mathrm{d}t}{\int_{0}^{\seclen}\pdfof{\lambda}\bra{\lambdalowerbound+i\cdot \seclen +t}\mathrm{d}t}\\
    &\leq \frac{\frac{2\privacythreshold}{-\ln\bra{1-\alpha}}\cdot\brb{\underline{c}+\mathcal{L}\bra{\seclen-x^*+\frac{\privacythreshold}{\ln\bra{1-\alpha}}}}}{\underline{c}\seclen+\frac{\mathcal{L}}{2}\bra{\seclen-x^*}^2},
\end{align*}
where $x^* = \seclen + \frac{\underline{c}}{\mathcal{L}} + \frac{\privacythreshold}{\ln\bra{1-\alpha}} - \sqrt{\bra{\frac{\underline{c}}{\mathcal{L}} + \frac{\privacythreshold}{\ln\bra{1-\alpha}}}^2 + \frac{2\underline{c}\seclen}{\mathcal{L}}}$.

\subsubsection{Proof of \cref{thm:upperbound_continuous_quantile_relaxed} for Shifted Exponential Distribution}
\label{sec:upperbound_continuous_quantile_shifted_exponential_relaxed}

It is straightforward to get the formula for $\distortionnotation$ from \cref{eq:lowerbound_shifted_exp_quantile_d}. Here we focus on the proof for $\privacynotation$.

According to \cref{eqn:green-yellow}, we can bound the attack success rate $\privacynotation$ as 
\begin{align*}
    \privacynotation{} 
    < \sup_{\rvparamnotation\in S_{green}}\probof{ \secretestimatestarof{\releaservparamnotation}\in\brb{ \secretofparam - \privacythreshold, \secretofparam + \privacythreshold } } + \int_{ \rvparamnotation\in S_{yellow}}p(\rvparamnotation)d\rvparamnotation.
\end{align*}

As for the first term $\sup_{\rvparamnotation\in S_{green}}\probof{ \secretestimatestarof{\releaservparamnotation}\in\brb{ \secretofparam - \privacythreshold, \secretofparam + \privacythreshold } }$, we can get that 
\begin{align*}
&\sup_{\rvparamnotation\in S_{green}}\probof{ \secretestimatestarof{\releaservparamnotation}\in\brb{ \secretofparam - \privacythreshold, \secretofparam + \privacythreshold } }\\
&=
\sup_{i\in \mathbb{N}, h, t'\in \mathbb{R} %
} 
\frac{\int_{\max\brc{-\frac{\seclen}{2}, t'}}^{\min\brc{\frac{\seclen}{2}, t'+\frac{2\privacythreshold}{\abs{\ln\bra{1-\alpha}+t_0}}}} \pdfof{\lambda, h}\bra{ \lambdalowerbound + \bra{\privateparamindexnotation+0.5}\cdot\seclen+t, h-t_0\cdot t} \mathrm{d}t}
{\int_{-\frac{\seclen}{2}}^{\frac{\seclen}{2}} \pdfof{\lambda, h}\bra{ \lambdalowerbound + \bra{\privateparamindexnotation+0.5}\cdot\seclen+t, h-t_0\cdot t}\mathrm{d}t}.
\end{align*}

To analyze the above term,
we provide the following lemma.
\begin{lemma}
\label{lemma:max_value_Lipschitz}
For a $\mathcal{L}$-Lipschitz continuous function $f(x), x\in \brb{\underline{x}, \overline{x}}$, if $\int_{\underline{x}}^{\overline{x}}f(x)\mathrm{d}x = \mathcal{A}$ and $\inf_{x\in \brb{\underline{x}, \overline{x}}} f(x) \geq \underline{c}$, it satisfies
\begin{align*}
\sup_{x\in \brb{\underline{x}, \overline{x}}} f(x) &\leq 
\begin{cases}
\frac{\mathcal{A}}{\overline{x}-\underline{x}} + \frac{\mathcal{L}\bra{\overline{x}-\underline{x}}}{2}, & \text{if } \underline{c} \leq \frac{\mathcal{A}}{\overline{x}-\underline{x}} - \frac{\mathcal{L}\bra{\overline{x}-\underline{x}}}{2}\\
\underline{c} + \sqrt{2\mathcal{L}\bra{\mathcal{A}-\underline{c}\bra{\overline{x}-\underline{x}}}}, & \text{if } \underline{c} > \frac{\mathcal{A}}{\overline{x}-\underline{x}} - \frac{\mathcal{L}\bra{\overline{x}-\underline{x}}}{2}
\end{cases}\\
& \triangleq M\bra{\overline{x} - \underline{x},  \underline{c}, \mathcal{L}, \mathcal{A}}.
\end{align*}
\end{lemma}

The proof is in \cref{sec:proof_of_max_value_lemma}.

Since $\pdfof{\lambda, h}\bra{ \lambdalowerbound + \bra{\privateparamindexnotation+0.5}\cdot\seclen+t, h-t_0\cdot t} = \pdfof{\lambda}\bra{ \lambdalowerbound + \bra{\privateparamindexnotation+0.5}\cdot\seclen+t}\cdot 
\pdfof{h}\bra{h-t_0\cdot t}$, according to \cref{lemma:max_value_Lipschitz}, we can get that $\pdfof{\lambda, h}$ is $\mathcal{L}_{\lambda, h} $-Lipschitz continuous, where
\small
$$
\mathcal{L}_{\lambda, h} = \mathcal{L}_{\lambda}\cdot M\bra{ \frac{\hupperbound-\hlowerbound}{\abs{t_0}}, \frac{k_{h}}{\hupperbound-\hlowerbound}, \abs{t_0}\mathcal{L}_h, \frac{1}{\abs{t_0}}} + \abs{t_0} \mathcal{L}_{h} \cdot M\bra{ \lambdaupperbound-\lambdalowerbound, \frac{k_{\lambda}}{\lambdaupperbound-\lambdalowerbound}, \mathcal{L}_{\lambda}, 1}.
$$
\normalsize
We can also get that 
$$
\inf_{a\in\brba{\lambdalowerbound, \lambdaupperbound}, b\in\brba{\hlowerbound, \hupperbound}} \pdfof{\lambda, h}\bra{a, b} \geq \frac{k_{h}k_{\lambda}}{\bra{\hupperbound-\hlowerbound}\cdot\bra{\lambdaupperbound-\lambdalowerbound}}
\triangleq \underline{c}.
$$

Therefore, according to \cref{lemma:max_fraction_Lipschitz_mean}, we can get that

\begin{align*}
&\sup_{\rvparamnotation\in S_{green}}\probof{ \secretestimatestarof{\releaservparamnotation}\in\brb{ \secretofparam - \privacythreshold, \secretofparam + \privacythreshold } }\\
&=
\sup_{i\in \mathbb{N}, h, t'\in \mathbb{R} %
} 
\frac{\int_{\max\brc{-\frac{\seclen}{2}, t'}}^{\min\brc{\frac{\seclen}{2}, t'+\frac{2\privacythreshold}{\abs{\ln\bra{1-\alpha}+t_0}}}} \pdfof{\lambda, h}\bra{ \lambdalowerbound + \bra{\privateparamindexnotation+0.5}\cdot\seclen+t, h-t_0\cdot t} \mathrm{d}t}
{\int_{-\frac{\seclen}{2}}^{\frac{\seclen}{2}} \pdfof{\lambda, h}\bra{ \lambdalowerbound + \bra{\privateparamindexnotation+0.5}\cdot\seclen+t, h-t_0\cdot t}\mathrm{d}t}\\
& \leq 
\frac{\frac{2\privacythreshold}{\abs{\ln\bra{1-\alpha}+t_0}}
\cdot \brb{\underline{c}+\mathcal{L}_{\lambda, h}\bra{\frac{\seclen}{2}-t^*-\frac{\privacythreshold}{\abs{\ln\bra{1-\alpha}+t_0}}}}}{\underline{c}\seclen+\frac{\mathcal{L}_{\lambda, h}}{2}\bra{\frac{\seclen}{2}-t^*}^2},
\end{align*}
where $t^* = \frac{\seclen}{2} + \frac{\underline{c}}{\mathcal{L}_{\lambda, h}} - \frac{\privacythreshold}{\abs{\ln\bra{1-\alpha}+t_0}} - \sqrt{\bra{\frac{\underline{c}}{\mathcal{L}_{\lambda, h}} - \frac{\privacythreshold}{\abs{\ln\bra{1-\alpha}+t_0}}}^2 + \frac{2\underline{c}\seclen}{\mathcal{L}_{\lambda, h}}}$, $
\mathcal{L}_{\lambda, h} = \mathcal{L}_{\lambda}\cdot M\bra{ \frac{\hupperbound-\hlowerbound}{\abs{t_0}}, \frac{k_{h}}{\hupperbound-\hlowerbound}, \abs{t_0}\mathcal{L}_h, \frac{1}{\abs{t_0}}} + \abs{t_0} \mathcal{L}_{h} \cdot M\bra{ \lambdaupperbound-\lambdalowerbound, \frac{k_{\lambda}}{\lambdaupperbound-\lambdalowerbound}, \mathcal{L}_{\lambda}, 1}
$, and
$\underline{c} = \frac{k_{h}k_{\lambda}}{\bra{\hupperbound-\hlowerbound}\cdot\bra{\lambdaupperbound-\lambdalowerbound}}$.

As for $\int_{ \rvparamnotation\in S_{yellow}}p(\rvparamnotation)d\rvparamnotation$, we have
\begin{align*}
&\int_{ \rvparamnotation\in S_{yellow}}p(\rvparamnotation)d\rvparamnotation \\
&\leq 
 M\bra{ 
 \hupperbound-\hlowerbound, \frac{k_{h}}{\hupperbound-\hlowerbound}, \mathcal{L}_h, 1}\cdot M\bra{ \lambdaupperbound-\lambdalowerbound, \frac{k_{\lambda}}{\lambdaupperbound-\lambdalowerbound}, \mathcal{L}_{\lambda}, 1}\cdot
 \int_{ \rvparamnotation\in S_{yellow}}d\rvparamnotation\\
& = 
M\bra{ 
 \hupperbound-\hlowerbound, \frac{k_{h}}{\hupperbound-\hlowerbound}, \mathcal{L}_h, 1}\cdot M\bra{ \lambdaupperbound-\lambdalowerbound, \frac{k_{\lambda}}{\lambdaupperbound-\lambdalowerbound}, \mathcal{L}_{\lambda}, 1}\cdot
 \bra{\lambdaupperbound-\lambdalowerbound}\abs{t_0}\seclen.
\end{align*}

Above all, we can get that
\begin{align*}
    \privacynotation{} 
    &< \sup_{\rvparamnotation\in S_{green}}\probof{ \secretestimatestarof{\releaservparamnotation}\in\brb{ \secretofparam - \privacythreshold, \secretofparam + \privacythreshold } } + \int_{ \rvparamnotation\in S_{yellow}}p(\rvparamnotation)d\rvparamnotation.
    \\
    & \leq \frac{\frac{2\privacythreshold}{\abs{\ln\bra{1-\alpha}+t_0}}
\cdot \brb{\underline{c}+\mathcal{L}_{\lambda, h}\bra{\frac{\seclen}{2}-t^*-\frac{\privacythreshold}{\abs{\ln\bra{1-\alpha}+t_0}}}}}{\underline{c}\seclen+\frac{\mathcal{L}_{\lambda, h}}{2}\bra{\frac{\seclen}{2}-t^*}^2} +\\
&M\bra{ 
 \hupperbound-\hlowerbound, \frac{k_{h}}{\hupperbound-\hlowerbound}, \mathcal{L}_h, 1}\cdot M\bra{ \lambdaupperbound-\lambdalowerbound, \frac{k_{\lambda}}{\lambdaupperbound-\lambdalowerbound}, \mathcal{L}_{\lambda}, 1}\cdot
 \bra{\lambdaupperbound-\lambdalowerbound}\abs{t_0}\seclen,
\end{align*}
where $M(\cdot,\cdot,\cdot,\cdot), \underline{c}, \mathcal{L}_{\lambda, h}, t^*$ are defined as above.

\subsubsection{Proof of \cref{lemma:max_value_Lipschitz}}
\label{sec:proof_of_max_value_lemma}

Without loss of generality, we assume that $f(\overline{x}) \geq f(\underline{x})$. 
Based on simple geometric analysis, we can get that there are two patterns when $\sup_{x\in \brb{\underline{x}, \overline{x}}} f(x)$ achieves supremum, which are shown in
\cref{fig:derivation_illustration}.

\begin{figure*}[htbp]
    \centering
    \begin{subfigure}{0.46\textwidth}
         \centering
        \includegraphics[width=1\linewidth]{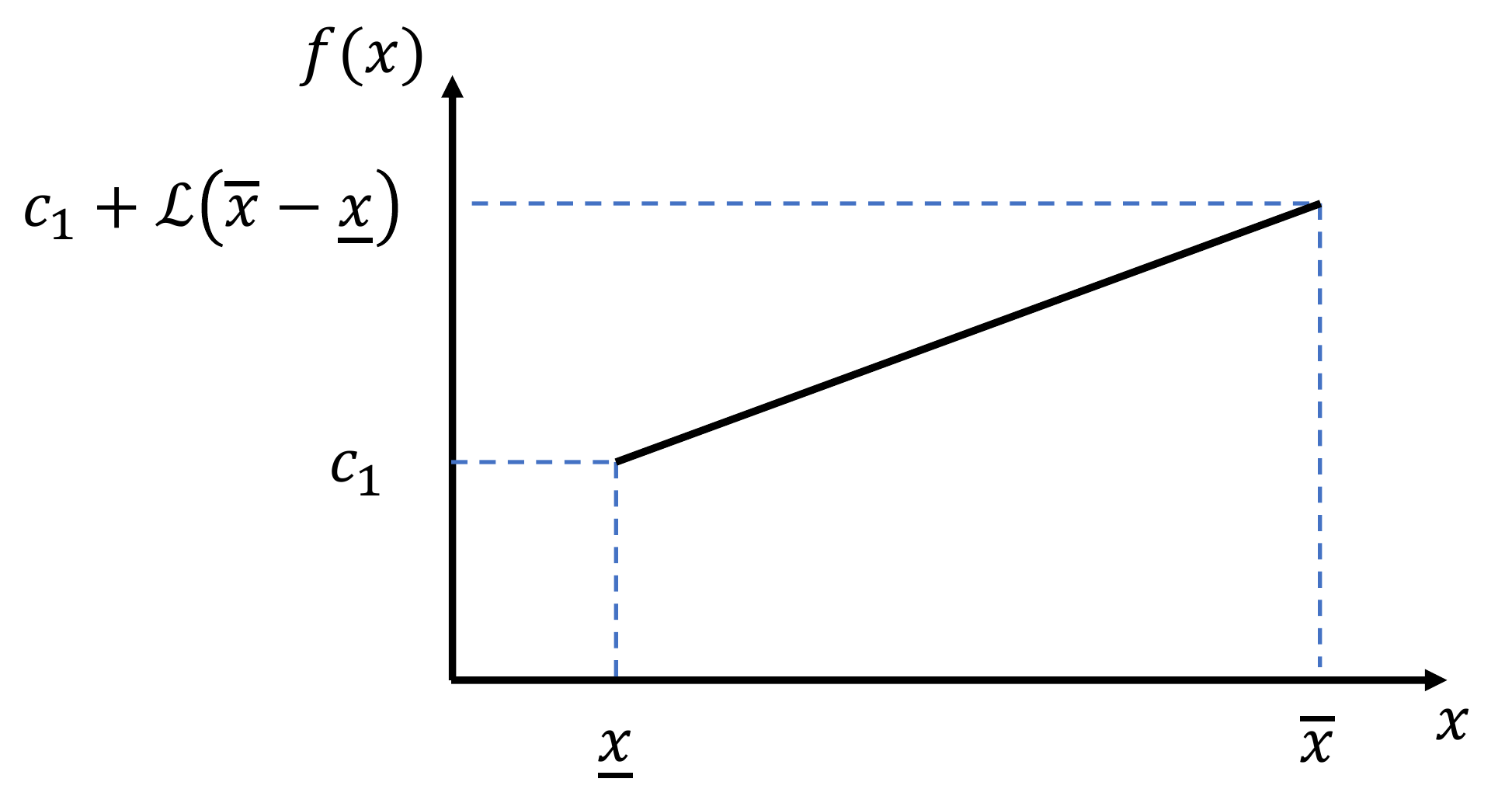}
         \caption{Pattern 1.}
     \end{subfigure}
     \hfill
    \begin{subfigure}{0.46\textwidth}
         \centering
        \includegraphics[width=1\linewidth]{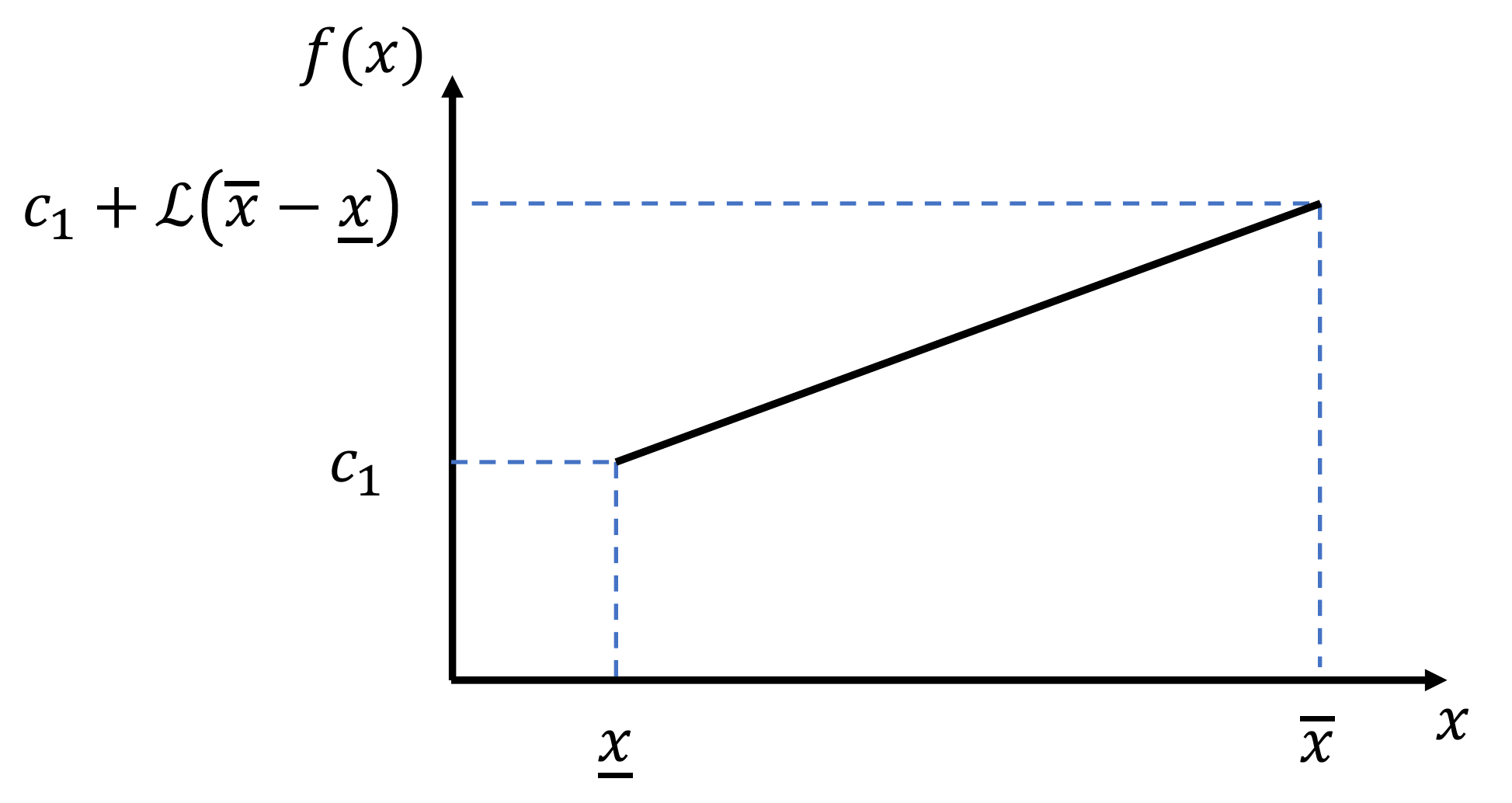}
         \caption{Pattern 2.}
     \end{subfigure}
    \caption{Two patterns when $\sup_{x\in \brb{\underline{x}, \overline{x}}} f(x)$ achieves supremum.
    }
\label{fig:derivation_illustration}
\end{figure*}

For pattern 1, $f(\underline{x}) = c_1 \geq \underline{c}$, $f(\overline{x}) = c_1 + \mathcal{L}(\overline{x}-\underline{x})$, and $\int_{\underline{x}}^{\overline{x}}f(x)\mathrm{d}x = \bra{c_1 + \frac{\mathcal{L}}{2}(\overline{x}-\underline{x})}\cdot\bra{\overline{x}-\underline{x}} = \mathcal{A}$. Therefore, when $\underline{c} \leq \frac{\mathcal{A}}{\overline{x}-\underline{x}} - \frac{\mathcal{L}\bra{\overline{x}-\underline{x}}}{2}$, we have 
$$
\sup_{f}\sup_{x\in \brb{\underline{x}, \overline{x}}} f(x) = c_1 + \mathcal{L}(\overline{x}-\underline{x}) = 
\frac{\mathcal{A}}{\overline{x}-\underline{x}} + \frac{\mathcal{L}\bra{\overline{x}-\underline{x}}}{2}.
$$

For pattern 2, $f(\underline{x}) =  \underline{c}$, $f(\overline{x}) = \underline{c} + \mathcal{L}(\overline{x}-x')$, where $x'\in \brab{\underline{x}, \overline{x}}$, and $\int_{\underline{x}}^{\overline{x}}f(x)\mathrm{d}x = \underline{c}(\overline{x}-\underline{x})+\frac{\mathcal{L}}{2}\bra{\overline{x}-x'}^2 = \mathcal{A}$. Therefore, when $\underline{c} > \frac{\mathcal{A}}{\overline{x}-\underline{x}} - \frac{\mathcal{L}\bra{\overline{x}-\underline{x}}}{2}$, we have 
$$
\sup_{f}\sup_{x\in \brb{\underline{x}, \overline{x}}} f(x) = \underline{c} + \mathcal{L}(\overline{x}-x') = 
\underline{c} + \sqrt{2\mathcal{L}\bra{\mathcal{A}-\underline{c}\bra{\overline{x}-\underline{x}}}}.
$$

Above all, we can get that 
\begin{align*}
\sup_{x\in \brb{\underline{x}, \overline{x}}} f(x) &\leq
\sup_f \sup_{x\in \brb{\underline{x}, \overline{x}}} f(x)\\
&= 
\begin{cases}
\frac{\mathcal{A}}{\overline{x}-\underline{x}} + \frac{\mathcal{L}\bra{\overline{x}-\underline{x}}}{2}, & \text{if } \underline{c} \leq \frac{\mathcal{A}}{\overline{x}-\underline{x}} - \frac{\mathcal{L}\bra{\overline{x}-\underline{x}}}{2}\\
\underline{c} + \sqrt{2\mathcal{L}\bra{\mathcal{A}-\underline{c}\bra{\overline{x}-\underline{x}}}}, & \text{if } \underline{c} > \frac{\mathcal{A}}{\overline{x}-\underline{x}} - \frac{\mathcal{L}\bra{\overline{x}-\underline{x}}}{2}
\end{cases}.
\end{align*}

\section{Discrete Distribution with \Secret{} = Mean}

\label{sec:case_study_mean_discrete}

Here, we consider three typical examples of discrete distributions: geometric distributions,  binomial distributions, and Poisson distributions with parameter $\rvparamnotation$. More specifically, the original distribution is 
\begin{align*}
    \probof{\rvprivatewithparam{\rvparamnotation} = k  } = 
    \begin{cases}
    \bra{1-\rvparamnotation}^k\rvparamnotation & \text{(geometric distribution)}\\
    \binom{n}{k}\rvparamnotation^k\bra{1-\rvparamnotation}^{n-k} & \text{(binomial distribution)}\\
    \frac{\rvparamnotation^k e^{-\rvparamnotation}}{k!} & \text{(Poisson distribution)}
    \end{cases}
\end{align*}
where $n$ standards for the number of trials in binomial distribution. The support of the parameter is 
$\setofprivateparam=\brc{\rvprivatewithparam{\rvparamnotation}: \rvparamnotation\in \brab{\rvparamlowerbound, \rvparamupperbound} } $
where $\brab{\rvparamlowerbound, \rvparamupperbound} \subseteq\bra{0, 1}$ for geometric distribution and binomial distribution, and $\brab{\rvparamlowerbound, \rvparamupperbound} \subseteq\bra{0, \infty}$ for Poisson distribution.

We first analyze the lower bound.
\begin{corollary}[Privacy lower bound, secret = mean of a discrete distribution]
\label{thm:discrete_mean}
Consider the secret function $\secretof{\rvparamnotation}=\sum_x x\privatepdf\bra{x}$. For any  $\privacymetricthreshold\in\bra{0,1}$, when $\privacynotation\leq \privacymetricthreshold$, we have $\distortionnotation> \bra{\ceil{\frac{1}{\privacymetricthreshold}}-1}\cdot 2\ratio\privacythreshold$, where the value of $\ratio$ depends on the type of the distributions:
\begin{packeditemize}
    \item Geometric: 
    \begin{align*}\ratio=\inf_{\rvparamlowerbound< \rvparamnotation_1<\rvparamnotation_2\leq \rvparamupperbound}\frac{\bra{1-\rvparamnotation_2}^{h\bra{\rvparamnotation_1,\rvparamnotation_2}} - \bra{1-\rvparamnotation_1}^{h\bra{\rvparamnotation_1,\rvparamnotation_2}}}{2\bra{\frac{1}{\rvparamnotation_2} - \frac{1}{\rvparamnotation_1}}}~,
    \end{align*} where $h\bra{\rvparamnotation_1,\rvparamnotation_2} = \floor{  \frac{\log\bra{\rvparamnotation_2} - \log\bra{\rvparamnotation_1}}{ \log\bra{1 - \rvparamnotation_1} - \log\bra{1 - \rvparamnotation_2 }  }  } + 1$.
    \item Binomial: 
    \begin{align*}
    &\ratio=\inf_{\rvparamlowerbound< \rvparamnotation_1<\rvparamnotation_2\leq \rvparamupperbound} \\
    &~~\scalebox{1}{$
    \frac{I_{1-\rvparamnotation_2}\bra{n-h\bra{\rvparamnotation_1,\rvparamnotation_2}, 1+h\bra{\rvparamnotation_1,\rvparamnotation_2}}-I_{1-\rvparamnotation_1}\bra{n-h\bra{\rvparamnotation_1,\rvparamnotation_2}, 1+h\bra{\rvparamnotation_1,\rvparamnotation_2}}}{2n\bra{\rvparamnotation_1-\rvparamnotation_2}},
    $}
    \end{align*}
    where $h\bra{\rvparamnotation_1,\rvparamnotation_2}=\floor{k'}$, $k' = n\ln\bra{\frac{1-\theta_2}{1-\theta_1}} \Big/ \ln\bra{{\frac{\theta_1\bra{1-\theta_2}}{\theta_2\bra{1-\theta_1}}}}$, %
    and $I$ represents the regularized incomplete beta function.
    \item Poisson: \begin{align*}\ratio=\inf_{\rvparamlowerbound< \rvparamnotation_1<\rvparamnotation_2\leq \rvparamupperbound}\frac{Q\bra{h\bra{\rvparamnotation_1,\rvparamnotation_2}, \rvparamnotation_2} - Q\bra{h\bra{\rvparamnotation_1,\rvparamnotation_2}, \rvparamnotation_1}}{2\bra{\rvparamnotation_1-\rvparamnotation_2}},\end{align*} where $h\bra{\rvparamnotation_1,\rvparamnotation_2} = \floor{\frac{\rvparamnotation_1 - \rvparamnotation_2}{\ln\bra{\rvparamnotation_1}-\ln\bra{\rvparamnotation_2}}} + 1$ and $Q$ is the regularized gamma function.
\end{packeditemize}
\end{corollary}
The proof is in \cref{sec:proof_lowerbound_discrete_mean}. The above lower bounds can be computed numerically.

Since these distributions only have one parameter, we can use \cref{alg:dp} and \cref{alg:greedy} to derive a \datamechanism{}. The performance of greedy-based and dynamic-programming-based \datamechanisms{} for each distribution is shown in \cref{fig:algorithm_mean}.

\begin{figure*}[th]
    \centering
    \begin{subfigure}{0.3\textwidth}
         \centering
        \includegraphics[width=1\linewidth]{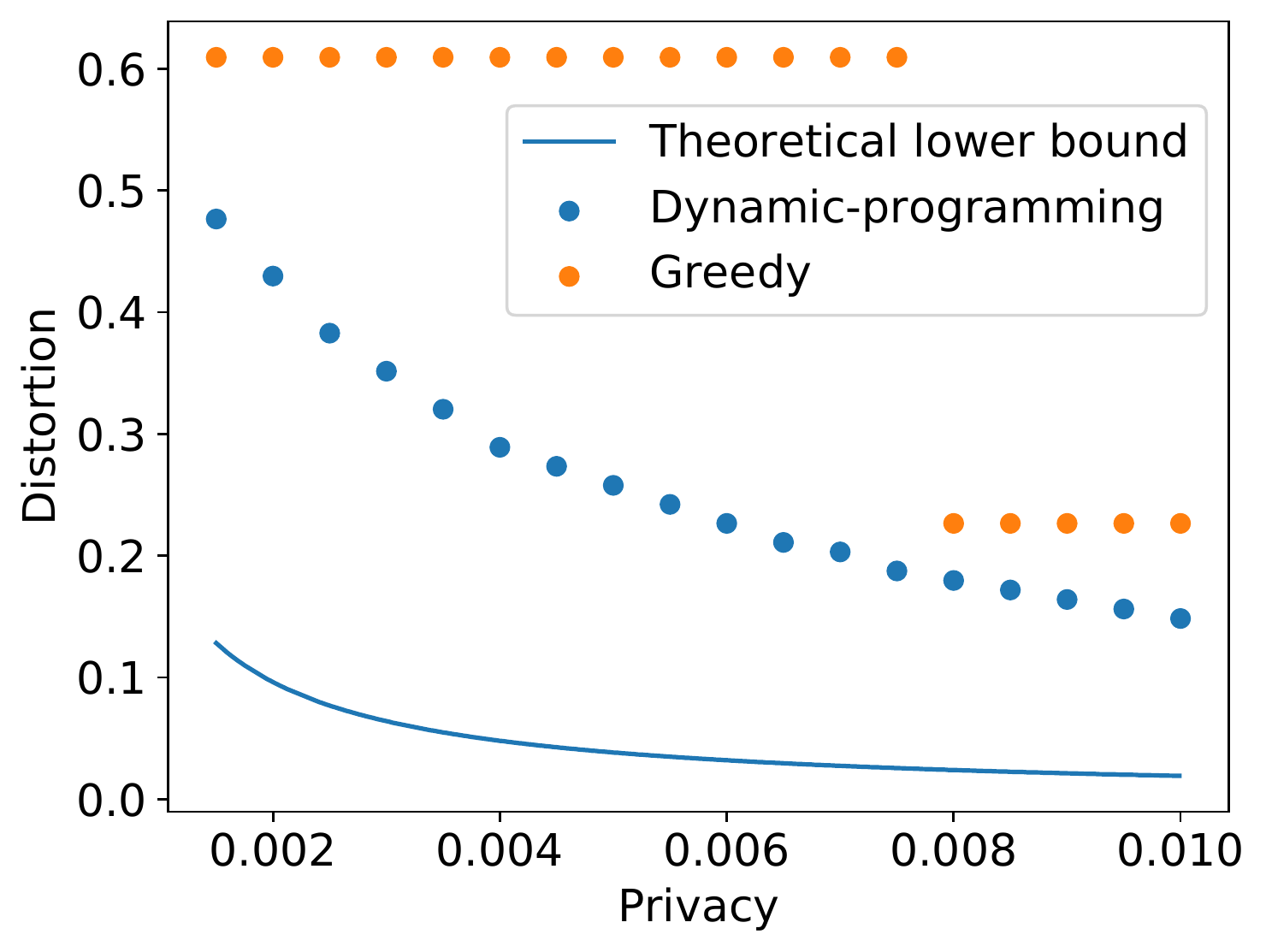}
         \caption{Distribution = Geometric}
         \label{fig:Geometric_mean}
     \end{subfigure}
     \hfill
    \begin{subfigure}{0.3\textwidth}
         \centering
        \includegraphics[width=1\linewidth]{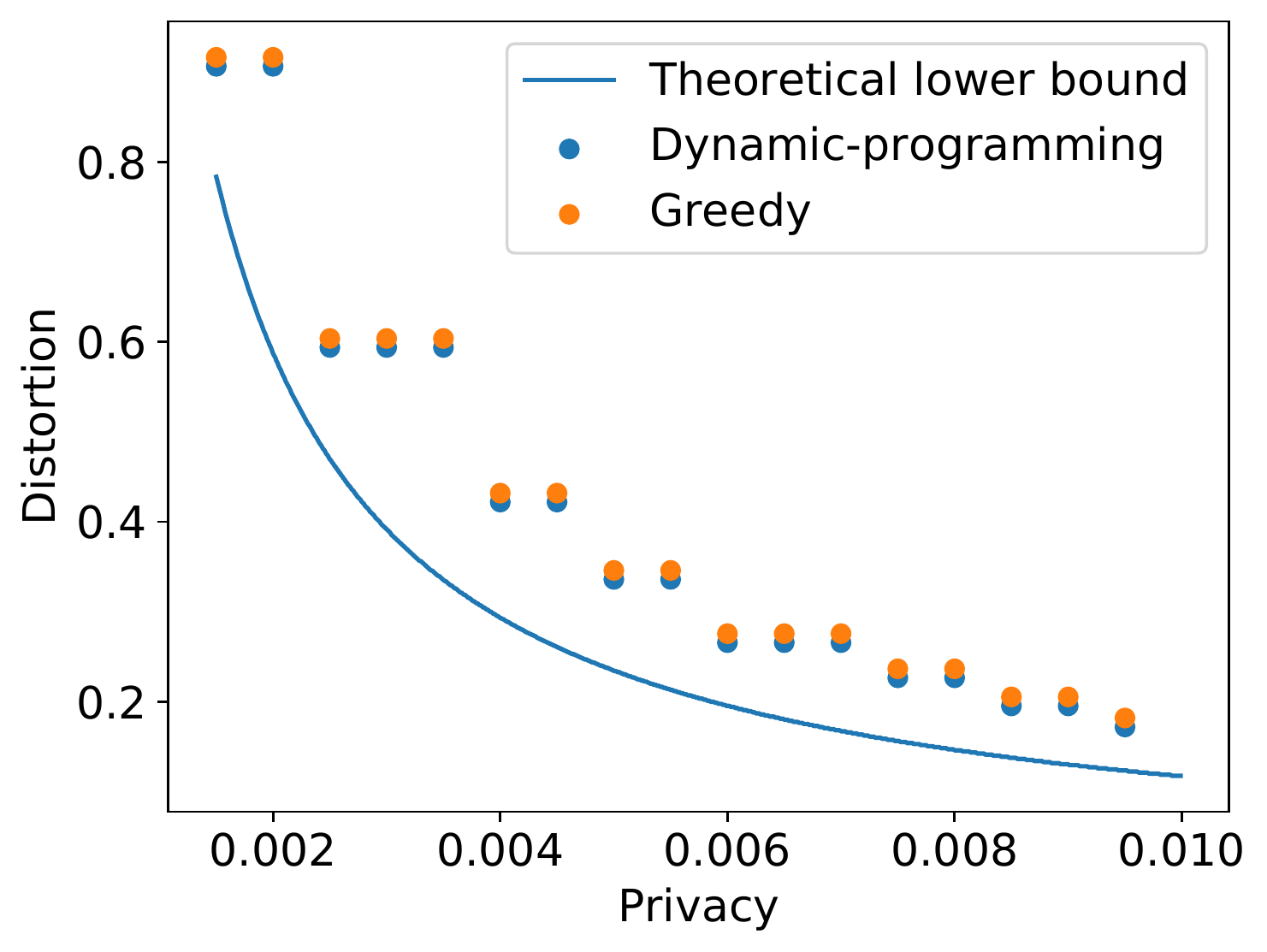}
         \caption{Distribution = Binomial}
         \label{fig:Binomial_mean}
     \end{subfigure}
     \hfill
    \begin{subfigure}{0.3\textwidth}
         \centering
        \includegraphics[width=1\linewidth]{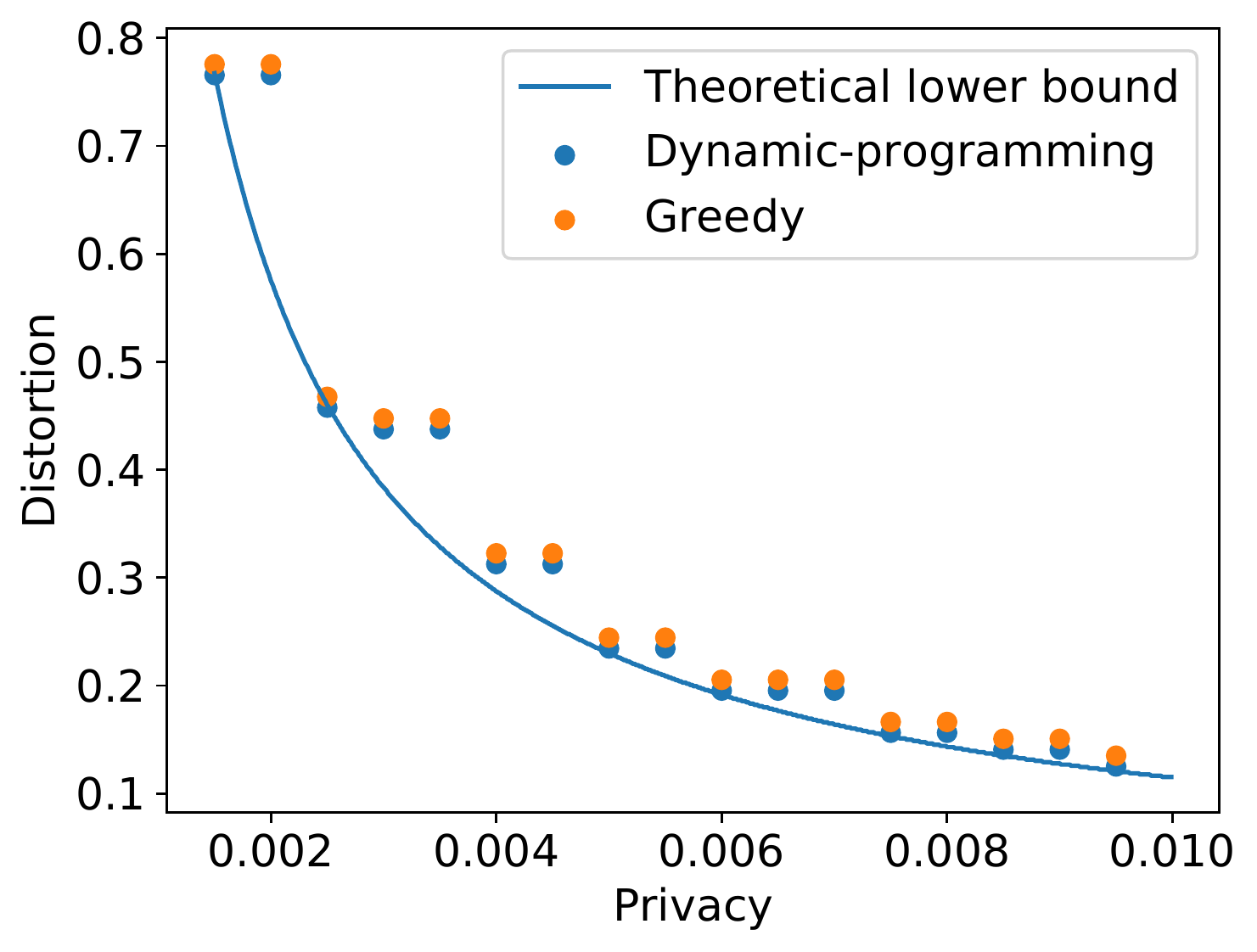}
         \caption{Distribution = Poisson}
         \label{fig:Poisson_mean}
     \end{subfigure}
    \caption{ \Privacy{}-\distortion{} performance of \cref{alg:dp} and \cref{alg:greedy} for geometric, binomial and Poisson distribution when secret = mean.
    }
    \label{fig:algorithm_mean}
\end{figure*}

As we can observe, the distortion that dynamic-programming-based \datamechanism{} achieves it is always smaller than or equal to that of the greedy-based \datamechanism{}.

\subsection{Proof of \cref{thm:discrete_mean}}%
\label{sec:proof_lowerbound_discrete_mean}
\subsubsection{Geometric Distribution}
\begin{proof}
Let $\rvprivatewithparam{\rvparamnotation_1}$ and $\rvprivatewithparam{\rvparamnotation_2}$  be two Geometric random variables with parameters $\rvparamnotation_1$ and $\rvparamnotation_2$ respectively.
We assume that $\rvparamnotation_1 > \rvparamnotation_2$ without loss of generality. Let $k'$ satisfy $\bra{1-\rvparamnotation_1}^{k'} \rvparamnotation_1 = \bra{1-\rvparamnotation_2}^{k'} \rvparamnotation_2$ and $k_0 = \floor{k'}  + 1$. Then we can get that
\begin{align*}
\auxdistance{\rvprivatewithparam{\rvparamnotation_1}}{\rvprivatewithparam{\rvparamnotation_2}}
&= \distanceformulaTV{\rvprivatewithparam{\rvparamnotation_1}}{\rvprivatewithparam{\rvparamnotation_2}}\\
&= \frac{1}{2}\bra{1-\rvparamnotation_2}^{k_0} - \frac{1}{2}\bra{1-\rvparamnotation_1}^{k_0},
\\
\auxrange{\rvprivatewithparam{\rvparamnotation_1}}{\rvprivatewithparam{\rvparamnotation_2}}
&= \frac{1}{\rvparamnotation_2} - \frac{1}{\rvparamnotation_1}.
\end{align*}

Therefore, we have
\begin{align*}
\ratio=
\inf_{\rvparamlowerbound< \rvparamnotation_1<\rvparamnotation_2\leq \rvparamupperbound} \frac{\bra{1-\rvparamnotation_2}^{k_0} - \bra{1-\rvparamnotation_1}^{k_0}}{2\bra{\frac{1}{\rvparamnotation_2} - \frac{1}{\rvparamnotation_1}}}~~.
\end{align*}

The rest follows from \cref{thm:trade_off_general}.
\end{proof}

\subsubsection{Binomial Distribution}
\begin{proof}
Let $\rvprivatewithparam{\rvparamnotation_1}$ and $\rvprivatewithparam{\rvparamnotation_2}$  be two binomial random variables with parameters $\rvparamnotation_1$ and $\rvparamnotation_2$ respectively with fixed number of trials $n$.
We assume that $\rvparamnotation_1 > \rvparamnotation_2$ without loss of generality.
Let $k'$ satisfy $\binom{n}{k'}\rvparamnotation_1^{k'}\bra{1-\rvparamnotation_1}^{n-k'} = \binom{n}{k'}\rvparamnotation_2^{k'}\bra{1-\rvparamnotation_1}^{n-k'}$ and $k_0=\lfloor k' \rfloor$. We can get that
\begin{align*}
\auxdistance{\rvprivatewithparam{\rvparamnotation_1}}{\rvprivatewithparam{\rvparamnotation_2}}&= 
\distanceformulaTV{\rvprivatewithparam{\rvparamnotation_1}}{\rvprivatewithparam{\rvparamnotation_2}}\\
&= \frac{1}{2}I_{1-\rvparamnotation_2}\bra{n-k_0, 1+k_0}-\frac{1}{2}I_{1-\rvparamnotation_1}\bra{n-k_0, 1+k_0},
\\
\auxrange{\rvprivatewithparam{\rvparamnotation_1}}{\rvprivatewithparam{\rvparamnotation_2}}
&= n\bra{\rvparamnotation_1-\rvparamnotation_2},
\end{align*}
where $I$ represents the regularized incomplete beta function.

Therefore, we have
\begin{align*}
\ratio =
\inf_{\rvparamlowerbound< \rvparamnotation_1<\rvparamnotation_2\leq \rvparamupperbound}\frac{I_{1-\rvparamnotation_2}\bra{n-k_0, 1+k_0}-I_{1-\rvparamnotation_1}\bra{n-k_0, 1+k_0}}{2n\bra{\rvparamnotation_1-\rvparamnotation_2}}.
\end{align*}

The rest follows from \cref{thm:trade_off_general}.
\end{proof}

\subsubsection{Poisson Distribution}
\begin{proof}
Let $\rvprivatewithparam{\rvparamnotation_1}$ and $\rvprivatewithparam{\rvparamnotation_2}$  be two Poisson random variables with parameters $\rvparamnotation_1$ and $\rvparamnotation_2$ respectively. We assume that $\rvparamnotation_1>\rvparamnotation_2$ without loss of generality.
Let $k'$ satisfy $\rvparamnotation_1^{k'} e^{-\rvparamnotation_1} = \rvparamnotation_2^{k'} e^{-\rvparamnotation_2}$ and $k_0 = \lfloor k' \rfloor + 1$. Then we can get that
\begin{align*}
\auxdistance{\rvprivatewithparam{\rvparamnotation_1}}{\rvprivatewithparam{\rvparamnotation_2}}&= 
\distanceformulaTV{\rvprivatewithparam{\rvparamnotation_1}}{\rvprivatewithparam{\rvparamnotation_2}}\\
&= \frac{1}{2}Q\bra{k_0, \rvparamnotation_2} - \frac{1}{2}Q\bra{k_0, \rvparamnotation_1},
\\
\auxrange{\rvprivatewithparam{\rvparamnotation_1}}{\rvprivatewithparam{\rvparamnotation_2}}
&= \rvparamnotation_1 - \rvparamnotation_2,
\end{align*}
where $Q$ is the regularized gamma function.

Therefore, we have
\begin{align*}
\ratio=
\inf_{\rvparamlowerbound< \rvparamnotation_1<\rvparamnotation_2\leq \rvparamupperbound}\frac{Q\bra{k_0, \rvparamnotation_2} - Q\bra{k_0, \rvparamnotation_1}}{2\bra{\rvparamnotation_1-\rvparamnotation_2}}.
\end{align*}

The rest follows from \cref{thm:trade_off_general}.
\end{proof}
\section{More Distributions with \Secret{} = Quantiles}
\label{sec:case_study_Gaussian_uniform_quantiles}
In this section, we discuss how to protect the quantiles for typical examples of continuous distributions: Gaussian distributions and uniform distributions. %
In our analysis, their parameters are denoted by:
\begin{packeditemize}
    \item Gaussian distributions: $\rvparamnotation=\bra{\mu,\sigma}$, where $\mu,\sigma$ are the mean and the standard deviation of the Gaussian distribution.
    \item Uniform distributions: $\rvparamnotation=\bra{m,n}$, where $m,n$ denote the lower and upper bound of the uniform distribution. In other words, $\rvprivatewithparam{m,n}$ is a random variable from uniform distribution $\uniformdistributionvar{m}{n}$.
\end{packeditemize}

As before, we first present the lower bound.
\begin{corollary}[Privacy lower bound, secret = $\alpha$-quantile of a continuous distribution]
\label{thm:lowerbound_continuous_quantile_more}
Consider the secret function $\secretof{\rvparamnotation}=\alpha$-quantile of $\privatepdf$. For any  $\privacymetricthreshold\in\bra{0,1}$, when $\privacynotation\leq \privacymetricthreshold$, we have $\distortionnotation> \bra{\ceil{\frac{1}{\privacymetricthreshold}}-1}\cdot 2\ratio\privacythreshold$, where the value of $\ratio$ depends on the type of the distributions:
\begin{packeditemize}
    \item Gaussian: \begin{align*}\ratio=\min_{t}\frac{\sqrt{\frac{1}{2\pi}}e^{-\frac{1}{2}t^2}-t\bra{\frac{1}{2}-\Phi\bra{t}}}{\abs{t+Q_{\alpha}}},
    \end{align*}
    where $\Phi$ denotes the CDF of the standard Gaussian distribution and
$ Q_{\alpha}\triangleq\Phi^{-1}(\alpha)$.
    \item Uniform: 
    \begin{align*}
        \ratio=\begin{cases}
\sqrt{\alpha^2-\alpha+\frac{1}{2}}+\alpha-\frac{1}{2} & \alpha \leq 0.5\\
\sqrt{\alpha^2-\alpha+\frac{1}{2}}-\alpha+\frac{1}{2} & \alpha>0.5
\end{cases}.
    \end{align*}
\end{packeditemize}
\end{corollary}
The proof is in \cref{sec:proof_lowerbound_continuous_quantile_more}.
The bound for uniform is in closed form, while the bound for Gaussian can be computed numerically.

Next, we provide \datamechanisms{} for each of the distributions. Here, we assume that the parameters of the original data are drawn from a uniform distribution with lower and upper bounds. In more details, we make the following assumptions.
\begin{assumption}
    \label{assu:quantile_continuous_Gaussian_uniform}
    The prior over distribution parameters as specified below.
    \begin{packeditemize}
        \item Gaussian: $\bra{\text{\textmugreek}, \text{\textsigma{}}}$ follows the uniform distribution over $\Big\{(a,b)|$ $\ a\in\brba{\mulowerbound, \muupperbound},b\in\brba{\sigmalowerbound, \sigmaupperbound}\Big\}$.
        \item Uniform: $\bra{M,N}$ follows the uniform distribution over $\big\{(a,b)| $ $\ a\in\brba{\mlowerbound, \mupperbound},b\in\brba{\mlowerbound, \mupperbound},a<b\big\}$.
    \end{packeditemize}
\end{assumption}

\begin{mechanism}[For secret = quantile of a continuous distribution]
    \label{mech:quantile_Gaussian_uniform_Gaussian_uniform}
    We design mechanisms for each of the distributions.
    \begin{packeditemize}
        \item Gaussian: 
            \begin{align*}
                \subsetofprivateparamof{\mu,\privateparamindexnotation} &= \brc{\bra{\mu+t_0\cdot t, \sigmalowerbound + \bra{\privateparamindexnotation+0.5}\cdot\seclen+t}| t\in \brba{-\frac{\seclen}{2}, \frac{\seclen}{2}}}
                ~~,\\
                \releaseparamofindex{\mu,\privateparamindexnotation} &= \bra{\mu,\sigmalowerbound+\bra{\privateparamindexnotation+0.5}\cdot\seclen} ~~,\\
                \privateparamindexsetnotation &=  \brc{\bra{\mu,\privateparamindexnotation}: \privateparamindexnotation\in\setofnaturalnumbers, \mu\in\setofreal},
            \end{align*}
             where $\seclen$ is a hyper-parameter of the mechanism that divides $\bra{\sigmaupperbound - \sigmalowerbound}$ and 
             \begin{align*}
                 t_0 = \arg\min_t \frac{\sqrt{\frac{1}{2\pi}}e^{-\frac{1}{2}t^2}-t\bra{\frac{1}{2}-\Phi\bra{t}}}{\abs{t+Q_{\alpha}}}.
             \end{align*}. %
        \item Uniform:
        \begin{align*}
            \subsetofprivateparamof{m,\privateparamindexnotation} &= 
            \scalebox{0.8}{$
            \brc{\bra{m-t_0\cdot t, m+\bra{\privateparamindexnotation+0.5}\cdot\seclen+t}| t\in \brab{-\frac{\seclen}{2\bra{t_0+1}}, \frac{\seclen}{2\bra{t_0+1}}}}
            $}
            ~~,\\
            \releaseparamofindex{m,\privateparamindexnotation} &= \bra{m,m+\bra{\privateparamindexnotation+0.5}\cdot\seclen} ~~,\\
            \privateparamindexsetnotation &=  \brc{\bra{m,\privateparamindexnotation}| \privateparamindexnotation\in\setofpositiveintegers, m\in\setofreal},
        \end{align*}
        where $t_0=\frac{1}{\frac{1}{l}-1}$ for 
        \begin{align*}
        l = 
        \begin{cases}
        \alpha + \sqrt{\alpha^2 - \alpha + \frac{1}{2}} & \alpha \leq 0.5\\
        \alpha - \sqrt{\alpha^2 - \alpha + \frac{1}{2}} & \alpha>0.5
        \end{cases}.
        \end{align*}
        and $\seclen>0$ is a hyper-parameter of the mechanism that divides $\bra{\mupperbound-\mlowerbound}$.
    \end{packeditemize}
\end{mechanism}

These \datamechanisms{} achieve the following $\distortionnotation$ and $\privacynotation$.
\begin{proposition}
    \label{thm:upperbound_continuous_quantile_Gaussian_uniform}
    Under \cref{assu:quantile_continuous_Gaussian_uniform},
    \cref{mech:quantile_Gaussian_uniform_Gaussian_uniform} %
    has the following $\distortionnotation$ and $\privacynotation{}$ value/bound.
    \begin{packeditemize}
        \item Gaussian: 
        \begin{align*}
        \privacynotation %
        & < \frac{2\privacythreshold}{\abs{t_0 + Q_\alpha}\seclen} + \frac{\abs{t_0}\seclen}{\muupperbound-\mulowerbound},\\
        \distortionnotation %
        & = \frac{\seclen}{2}\sqrt{\frac{2}{\pi}}e^{-\frac{1}{2}t_0^2} - \frac{t_0 \seclen}{2}\bra{1-2\Phi\bra{t_0}}< \bra{2+\frac{\abs{t_0}\cdot\abs{t_0 + Q_\alpha}\seclen^2}{\bra{\muupperbound-\mulowerbound}\epsilon}} \distortionnotation_{\text{opt}}.
        \end{align*}
        Under the ``high-precision'' regime where $ {\frac{\seclen^2}{\muupperbound-\mulowerbound}}\rightarrow 0$ as $\seclen, (\muupperbound-\mulowerbound)\to\infty$, $\distortionnotation$ satisfies
\begin{align*}
    \lim{\sup_{{\frac{\seclen^2}{\muupperbound-\mulowerbound}}\rightarrow 0}} \distortionnotation < 3  \distortionnotation_{\text{opt}}.
\end{align*}
        \item Uniform:
        \begin{align*}
            \privacynotation &<  \frac{2\privacythreshold\bra{t_0+1}}{\abs{\bra{1-\alpha}t_0-\alpha}\seclen} +  \frac{2\seclen \cdot t_0}{\bra{t_0+1}\bra{\mupperbound-\mlowerbound}} + \frac{\seclen^2}{2\bra{\mupperbound-\mlowerbound}^2},\\
            \distortionnotation &= \frac{\bra{t_0^2+1}\seclen}{4(t_0+1)^2}\\
            &< \bra{2+\frac{\abs{\bra{1-\alpha}t_0-\alpha}\seclen}{\privacythreshold\bra{t_0+1}}\cdot \bra{  \frac{2\seclen \cdot t_0}{\bra{t_0+1}\bra{\mupperbound-\mlowerbound}} + \frac{\seclen^2}{2\bra{\mupperbound-\mlowerbound}^2}}} \distortionnotation_{\text{opt}}.
        \end{align*}
    Under the ``high-precision'' regime where ${\frac{\seclen^2}{\mupperbound-\mlowerbound}}\rightarrow 0$ as $\seclen, (\mupperbound-\mlowerbound)\to \infty$, $\distortionnotation$ satisfies
\begin{align*}
    \lim{\sup_{{\frac{\seclen^2}{\mupperbound-\mlowerbound}}\rightarrow 0}} \distortionnotation < 3  \distortionnotation_{\text{opt}}.
\end{align*}
    \end{packeditemize}
    The $t_0$ parameter is defined in \cref{mech:quantile_Gaussian_uniform_Gaussian_uniform} for each distribution. 
\end{proposition}

The proof is in \cref{sec:proof_upperbound_continuous_quantile_Gaussian_uniform}. For Gaussian distribution, we relax \cref{assu:quantile_continuous_Gaussian_uniform} and analyze the \privacy{}-\distortion{} performance of \cref{mech:quantile_Gaussian_uniform_Gaussian_uniform} in \cref{sec:proof_trade-off_mechanism_quantile_continuous_relaxed_Gaussian}. 
For both distributions, we consider the ``high-precision'' regime. %
The two takeaways are that: (1) data holder can use $\seclen$ to control the trade-off between \distortion{} and \privacy{}, and (2) the mechanism is order-optimal with multiplicative factor $3$.

\subsection{Proof of \cref{thm:lowerbound_continuous_quantile_more}}
\label{sec:proof_lowerbound_continuous_quantile_more}

\subsubsection{Gaussian Distribution}
\begin{proof}
Let $\rvprivatewithparam{\mu_1,\sigma_2}, \rvprivatewithparam{\mu_2,\sigma_2} $ be two Gaussian random variables with means $\mu_1,\mu_2$ and sigmas $\sigma_1,\sigma_2$ respectively. 
Let $\Phi$ denotes the CDF of the standard Gaussian distribution and let
$\Phi^{-1}(\alpha)\triangleq Q_{\alpha}$.

When $\sigma_1=\sigma_2$, we have
\begin{align*}
\frac{\auxdistance{\rvprivatewithparam{\mu_1,\sigma_1}}{\rvprivatewithparam{\mu_2,\sigma_2}}}{\auxrange{\rvprivatewithparam{\mu_1,\sigma_1}}{\rvprivatewithparam{\mu_2,\sigma_2}}} = \frac{\frac{1}{2}\abs{\mu_1-\mu_2}}{\abs{\mu_1+\sigma Q_{\alpha}-\bra{\mu_2+\sigma Q_{\alpha}}}} = \frac{1}{2}.
\end{align*}

When $\sigma_1\not=\sigma_2$, we assume $\sigma_2>\sigma_1$ without loss of generality. Let $a=\frac{\sigma_1}{\sigma_2}$ and $b=\frac{\sigma_2}{\sigma_1}\mu_1-\mu_2$. Let $a=\frac{\sigma_1}{\sigma_2}$ and $b=\frac{\sigma_2}{\sigma_1}\mu_1-\mu_2$.
We can get that $\pdfof{\rvprivatewithparam{\mu_1,\sigma_1}}\bra{x}=a\pdfof{\rvprivatewithparam{\mu_2,\sigma_2}}\bra{a\bra{x+b}}$, and %
\begin{align*}
\auxdistance{\rvprivatewithparam{\mu_1,\sigma_1}}{\rvprivatewithparam{\mu_2,\sigma_2}}
&= \distanceformulawass{\rvprivatewithparam{\mu_1,\sigma_1}}{\rvprivatewithparam{\mu_2,\sigma_2}}\\
&= \frac{1}{2}\int_{-\infty}^{+\infty}\left|x-\bra{\frac{x}{a}-b}\right|\pdfof{\rvprivatewithparam{\mu_1,\sigma_1}}\bra{x}\mathrm{d}x\\
&=\bra{\mu_1-\mu_2}\bra{\Phi\bra{\frac{\mu_1-\mu_2}{\sigma_2-\sigma_1}}-\frac{1}{2}}\\
&\quad +\sqrt{\frac{1}{2\pi}}\bra{\sigma_2-\sigma_1}e^{-\frac{1}{2}\bra{\frac{\mu_1-\mu_2}{\sigma_2-\sigma_1}}^2}\numberthis\label{eq:lowerbound_gaussian_quantile_d},
\\
\auxrange{\rvprivatewithparam{\mu_1,\sigma_1}}{\rvprivatewithparam{\mu_2,\sigma_2}}
&= \abs{\mu_1+\sigma_1 Q_{\alpha}-\bra{\mu_2+\sigma_2 Q_{\alpha}}}\\
&= \abs{\bra{\mu_1-\mu_2}+\bra{\sigma_1-\sigma_2}Q_{\alpha}}.
\end{align*}

Let $\frac{\mu_1-\mu_2}{\sigma_1-\sigma_2}\triangleq t$, we can get that
\begin{align*}
\frac{\auxdistance{\rvprivatewithparam{\mu_1,\sigma_1}}{\rvprivatewithparam{\mu_2,\sigma_2}}}{\auxrange{\rvprivatewithparam{\mu_1,\sigma_1}}{\rvprivatewithparam{\mu_2,\sigma_2}}} &= \frac{\sqrt{\frac{1}{2\pi}}e^{-\frac{1}{2}t^2}-t\bra{\frac{1}{2}-\Phi\bra{t}}}{\abs{t+Q_{\alpha}}}\triangleq h\bra{t}.
\end{align*}
Since $\lim_{t\rightarrow\infty} = \frac{1}{2}$, we have $\min\brc{\min_{t}h\bra{t}, \frac{1}{2}} = \min_{t}h\bra{t}$, and therefore we can get that
\begin{align*}
\ratio= \min_{t}h\bra{t}.
\end{align*}
\end{proof}

\subsubsection{Uniform Distribution}
\begin{proof}

Let $\rvprivatewithparam{m_1,n_1}, \rvprivatewithparam{m_2,n_2}$ be two uniform random variables. Let $\cdfof{\rvprivatewithparam{m_1,n_1}}, \cdfof{\rvprivatewithparam{m_2,n_2}}$ be their CDFs, and let $m_2\geq m_1$ without loss of generality. We can get that
\begin{align*}
\auxdistance{\rvprivatewithparam{m_1,n_1}}{\rvprivatewithparam{m_2,n_2}}
&= \distanceformulawass{\rvprivatewithparam{m_1,n_1}}{\rvprivatewithparam{m_2,n_2}}\\
&= \frac{1}{2} \int_{-\infty}^{+\infty} \abs{\cdfof{\rvprivatewithparam{m_1,n_1}}\bra{x}-\cdfof{\rvprivatewithparam{m_2,n_2}}\bra{x}} \mathrm{d}x\\
&=\begin{cases}
\frac{m_2-m_1+n_2-n_1}{4} & n_2\geq n_1\\
\frac{\bra{m_2-m_1}^2+\bra{n_1-n_2}^2}{4\bra{m_2-m_1+\bra{n_1-n_2}}} & n_2<n_1
\end{cases},\numberthis\label{eq:lowerbound_uniform_quantile_d}
\\
\auxrange{\rvprivatewithparam{m_1,n_1}}{\rvprivatewithparam{m_2,n_2}}&= 
\abs{m_2+\alpha\bra{n_2-m_2}-\brb{m_1+\alpha\bra{n_1-m_1}}}\\
&= \abs{\bra{1-\alpha}\bra{m_2-m_1}+\alpha\bra{n_2-n_1}}.
\end{align*}

When $n_2=n_1$, we have
\begin{align*}
\frac{\auxdistance{\rvprivatewithparam{m_1,n_1}}{\rvprivatewithparam{m_2,n_2}}}{\auxrange{\rvprivatewithparam{m_1,n_1}}{\rvprivatewithparam{m_2,n_2}}} = \frac{m_2-m_1}{4\bra{1-\alpha}\bra{m_2-m_1}}=\frac{1}{4\bra{1-\alpha}}.
\end{align*}

When $n_2 > n_1$, let $t_1 = \frac{m_2-m_1}{n_2-n_1}\in \brba{0, +\infty}$, we have
\begin{align*}
\frac{\auxdistance{\rvprivatewithparam{m_1,n_1}}{\rvprivatewithparam{m_2,n_2}}}{\auxrange{\rvprivatewithparam{m_1,n_1}}{\rvprivatewithparam{m_2,n_2}}} &= \frac{1}{4}\frac{m_2-m_1+n_2-n_1}{\bra{1-\alpha}\bra{m_2-m_1}+\alpha\bra{n_2-n_1}}\\
& =\frac{1}{4}\frac{t_1+1}{\bra{1-\alpha}t_1+\alpha}\\
& = \frac{1}{4\bra{1-\alpha}}\bra{1+\frac{1-2\alpha}{1-\alpha}\cdot\frac{1}{t_1+\frac{\alpha}{1-\alpha}}}\\
& \geq \begin{cases}
\frac{1}{4\bra{1-\alpha}} & \alpha \leq 0.5\\
\frac{1}{4\alpha} & \alpha>0.5
\end{cases}.
\end{align*}

When $n_2<n_1$, let $t_2 = \frac{m_2-m_1}{n_1-n_2}\in (0, +\infty)$, we have
\begin{align*}
\frac{\auxdistance{\rvprivatewithparam{m_1,n_1}}{\rvprivatewithparam{m_2,n_2}}}{\auxrange{\rvprivatewithparam{m_1,n_1}}{\rvprivatewithparam{m_2,n_2}}} &=
\frac{1}{4} \frac{\bra{m_2-m_1}^2+\bra{n_1-n_2}^2}{\bra{m_2-m_1+\bra{n_1-n_2}}}\cdot\\
&\quad\ \frac{1}{\abs{\bra{1-\alpha}\bra{m_2-m_1}-\alpha\bra{n_1-n_2}}}\\
& = \frac{1}{4}\frac{t_2^2+1}{\bra{t_2+1}\abs{\bra{1-\alpha}t_2-\alpha}}\\
& \geq \begin{cases}
\sqrt{\alpha^2-\alpha+\frac{1}{2}}+\alpha-\frac{1}{2} & \alpha \leq 0.5\\
\sqrt{\alpha^2-\alpha+\frac{1}{2}}-\alpha+\frac{1}{2} & \alpha>0.5
\end{cases}.
\end{align*}
``$=$'' achieves when $t_2 = \frac{1}{\frac{1}{l}-1} \triangleq t_0$, where 
\begin{align*}
l = 
\begin{cases}
\alpha + \sqrt{\alpha^2 - \alpha + \frac{1}{2}} & \alpha \leq 0.5\\
\alpha - \sqrt{\alpha^2 - \alpha + \frac{1}{2}} & \alpha>0.5
\end{cases}.
\end{align*}

Therefore we can get that
\begin{equation}
\begin{aligned}
\nonumber
\ratio= \begin{cases}
\sqrt{\alpha^2-\alpha+\frac{1}{2}}+\alpha-\frac{1}{2} & \alpha \leq 0.5\\
\sqrt{\alpha^2-\alpha+\frac{1}{2}}-\alpha+\frac{1}{2} & \alpha>0.5
\end{cases}.
\end{aligned}
\end{equation}
\end{proof}

\subsection{Proof of \cref{thm:upperbound_continuous_quantile_Gaussian_uniform}}
\label{sec:proof_upperbound_continuous_quantile_Gaussian_uniform}
\subsubsection{Gaussian Distribution}
\label{sec:proof_upperbound_continuous_quantile_GAUSSIAN}

\begin{proof}

We first focus on the proof for $\privacynotation$.

\begin{figure}[h]
    \centering
    \includegraphics[width=0.8\linewidth]{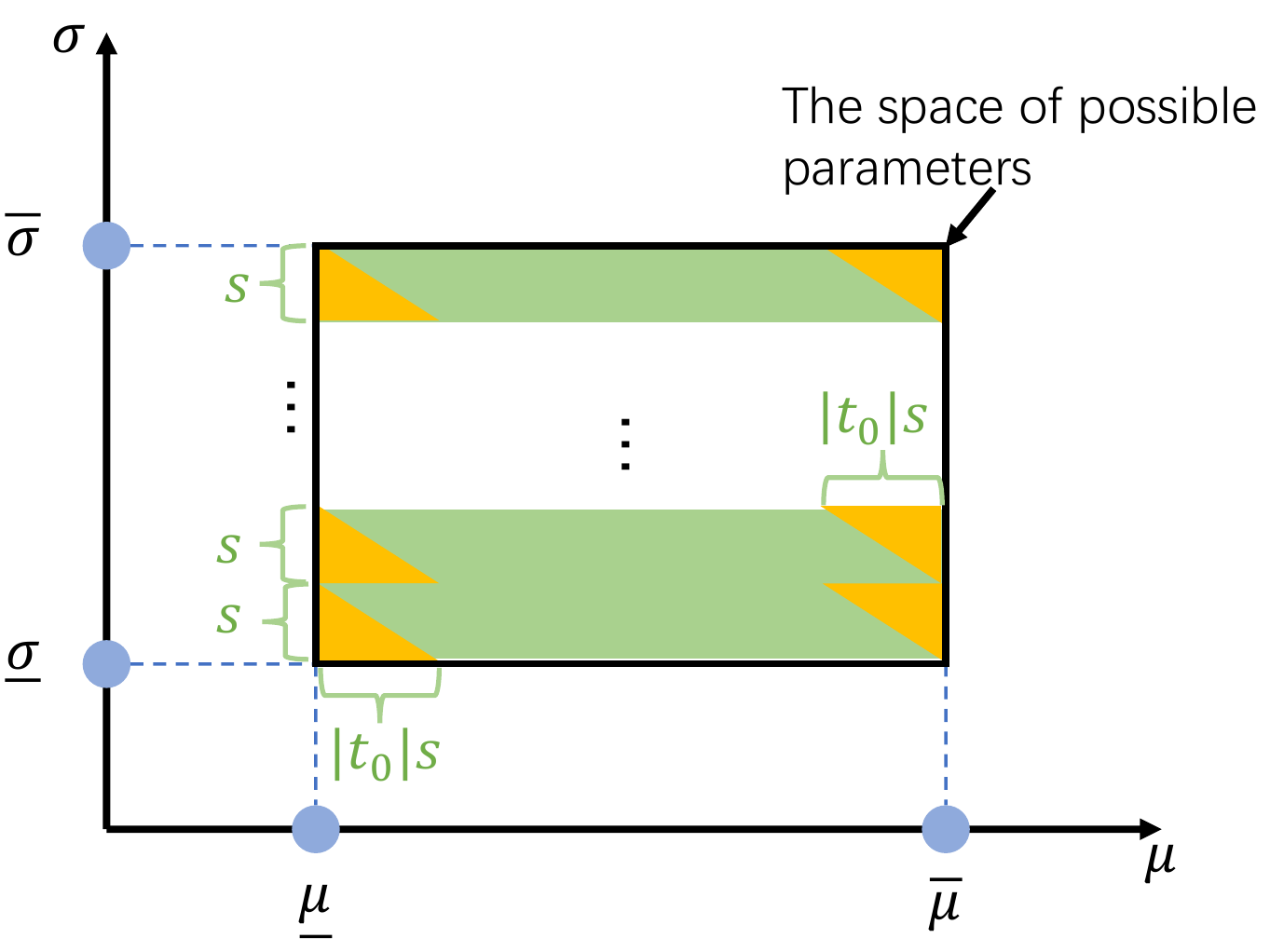}
    \caption{The construction for proof of \cref{thm:upperbound_continuous_quantile_Gaussian_uniform} for Gaussian distributions. We separate the space of possible parameters into two regions (yellow and green) and bound the attacker's success rate on each region separately.}
    \label{fig:upperbound_gaussian_quantile}
\end{figure}

In \cref{fig:upperbound_gaussian_quantile}, we separate the space of possible data parameters into two regions represented by yellow and green colors. The yellow regions $S_{yellow}$ constitute right triangles with height $\seclen$ and width $\abs{t_0}\seclen$.  The green region $S_{green}$ is the rest of the parameter space. 
The high-level idea of our proof is as follows.
Note that for any parameter $\rvparamnotation\in S_{green}$, there exists a $\subsetofprivateparamof{\mu,\privateparamindexnotation}$ s.t. $\rvparamnotation \in \subsetofprivateparamof{\mu,\privateparamindexnotation}$ and $\subsetofprivateparamof{\mu,\privateparamindexnotation}\subset S_{green}$. Therefore, we can bound the attack success rate if  $\rvparamnotation\in S_{green}$. At the same time, the probability of $\rvparamnotation\in S_{yellow}$ is bounded.
Therefore, we can bound the overall attacker's success rate (i.e., $\privacynotation$). More specifically, let the optimal attacker be $\secretestimatestarnotation$. We have
\begin{align*}
    \privacynotation{} 
    &= \probof{ \secretestimatestarof{\releaservparamnotation}\in\brb{ \secretofparam - \privacythreshold, \secretofparam + \privacythreshold } }\\
    &= \int_{ \rvparamnotation\in S_{green}}p(\rvparamnotation)\probof{ \secretestimatestarof{\releaservparamnotation}\in\brb{ \secretofparam - \privacythreshold, \secretofparam + \privacythreshold } }d\rvparamnotation \\
    &\quad+  \int_{ \rvparamnotation\in S_{yellow}}p(\rvparamnotation)\probof{ \secretestimatestarof{\releaservparamnotation}\in\brb{ \secretofparam - \privacythreshold, \secretofparam + \privacythreshold }} d\rvparamnotation\\
    &<\frac{2\privacythreshold}{\abs{t_0 + Q_\alpha}\seclen} + \frac{\abs{t_0}\seclen}{\muupperbound-\mulowerbound}.
\end{align*}

For the \distortion{},
it is straightforward to get that $\distortionnotation = \frac{\seclen}{2}\sqrt{\frac{2}{\pi}}e^{-\frac{1}{2}t_0^2} - \frac{t_0 \seclen}{2}\bra{1-2\Phi\bra{t_0}}$ from \cref{eq:lowerbound_gaussian_quantile_d}, and $\distortionnotation_{\text{opt}}>\bra{\ceil{\frac{1}{\privacynotation{}}}-1}\cdot 2\ratio\privacythreshold \geq 2\ratio\privacythreshold$, where $\ratio$ is defined in \cref{thm:lowerbound_continuous_quantile_more}.
We can get that $\bra{\privacynotation - \frac{\abs{t_0}\seclen}{\muupperbound-\mulowerbound}} \cdot \distortionnotation < 2\ratio\privacythreshold$ and 
\begin{align*}
\distortionnotation & = \distortionnotation_{\text{opt}} + \distortionnotation - \distortionnotation_{\text{opt}}\\
& < \distortionnotation_{\text{opt}} + \distortionnotation - \bra{\ceil{\frac{1}{\privacynotation}}-1}\cdot 2\ratio\privacythreshold\\
&\leq \distortionnotation_{\text{opt}} + 2\ratio\privacythreshold + \distortionnotation - {{\frac{2\ratio\privacythreshold}{\privacynotation}}}\\
&< \distortionnotation_{\text{opt}} + 2\ratio\privacythreshold + \frac{\frac{\abs{t_0}\seclen}{\muupperbound-\mulowerbound}}{\frac{2\privacythreshold}{\abs{t_0 + Q_\alpha}\seclen} + \frac{\abs{t_0}\seclen}{\muupperbound-\mulowerbound}}\cdot\distortionnotation\\
&=\bra{1+\frac{\abs{t_0}\cdot\abs{t_0 + Q_\alpha}\seclen^2}{2\epsilon\bra{\muupperbound-\mulowerbound}}}\bra{\distortionnotation_{\text{opt}} + 2\ratio\privacythreshold } \\
&\leq \bra{2+\frac{\abs{t_0}\cdot\abs{t_0 + Q_\alpha}\seclen^2}{\epsilon\bra{\muupperbound-\mulowerbound}}} \distortionnotation_{\text{opt}}.
\end{align*}

\end{proof}

\subsubsection{Uniform Distribution}
\label{sec:proof_upperbound_continuous_quantile_UNIFORM}

\begin{proof}

We first focus on the proof for $\privacynotation$.

\begin{figure}[h]
    \centering
    \includegraphics[width=0.65\linewidth]{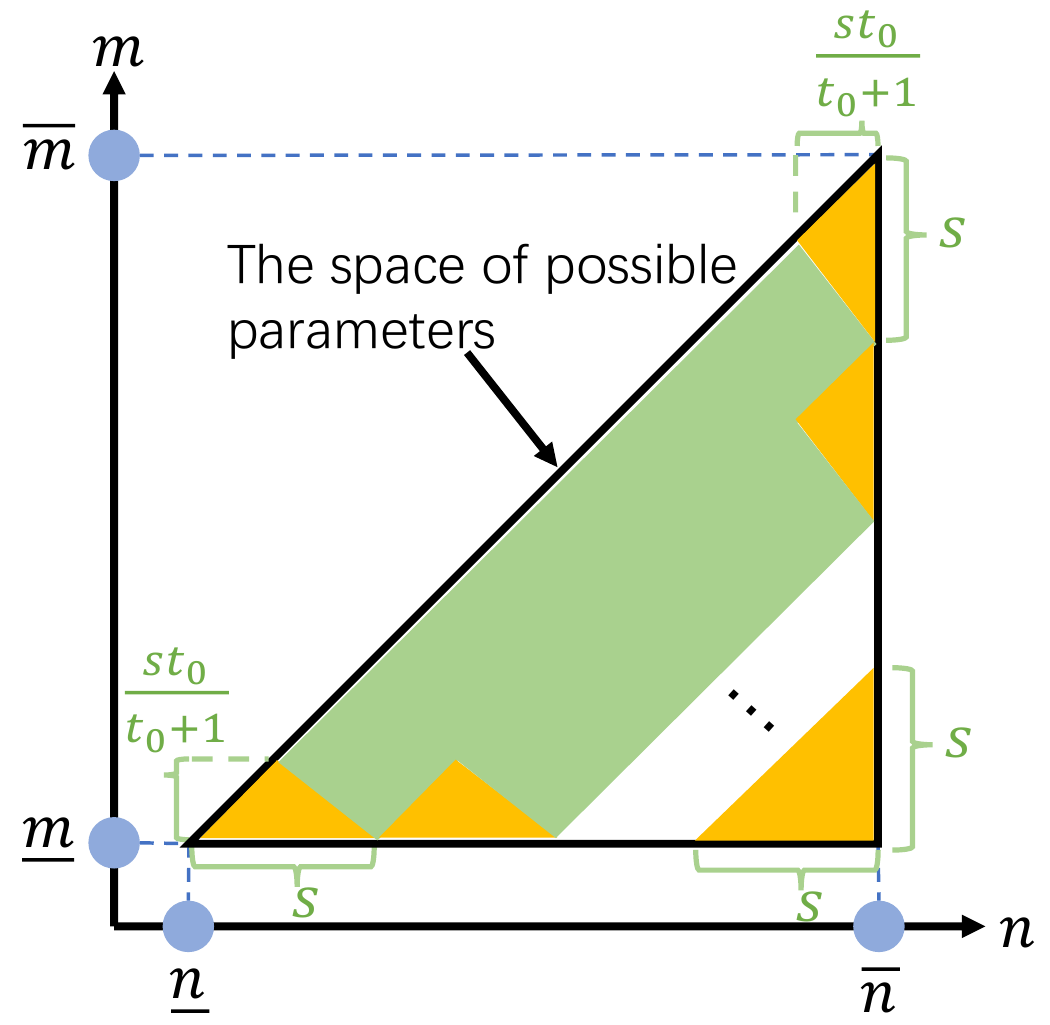}
    \caption{The construction for proof of \cref{thm:upperbound_continuous_quantile_Gaussian_uniform} for uniform distributions. We separate the space of possible parameters into two regions (yellow and green) and bound the attacker's success rate on each region separately.}
    \label{fig:upperbound_uniform_quantile}
\end{figure}

In \cref{fig:upperbound_uniform_quantile}, we separate the space of possible data parameters into two regions represented by yellow and green colors. The yellow regions $S_{yellow}$ constitute triangles with height $\frac{\seclen t_0}{t_0+1}$ and width $\seclen$ (except for the right-bottom triangle with height and width $\seclen$).  The green region $S_{green}$ is the rest of the parameter space. 
The high-level idea of our proof is as follows.
Note that for any parameter $\rvparamnotation\in S_{green}$, there exists a $\subsetofprivateparamof{\mu,\privateparamindexnotation}$ s.t. $\rvparamnotation \in \subsetofprivateparamof{\mu,\privateparamindexnotation}$ and $\subsetofprivateparamof{\mu,\privateparamindexnotation}\subset S_{green}$. Therefore, we can bound the attack success rate if  $\rvparamnotation\in S_{green}$. At the same time, the probability of $\rvparamnotation\in S_{yellow}$ is bounded.
Therefore, we can bound the overall attacker's success rate (i.e., $\privacynotation$). More specifically, let the optimal attacker be $\secretestimatestarnotation$. We have
\begin{align*}
    \privacynotation{} 
    &= \probof{ \secretestimatestarof{\releaservparamnotation}\in\brb{ \secretofparam - \privacythreshold, \secretofparam + \privacythreshold } }\\
    &= \int_{ \rvparamnotation\in S_{green}}p(\rvparamnotation)\probof{ \secretestimatestarof{\releaservparamnotation}\in\brb{ \secretofparam - \privacythreshold, \secretofparam + \privacythreshold } }d\rvparamnotation \\
    &\quad+  \int_{ \rvparamnotation\in S_{yellow}}p(\rvparamnotation)\probof{ \secretestimatestarof{\releaservparamnotation}\in\brb{ \secretofparam - \privacythreshold, \secretofparam + \privacythreshold }} d\rvparamnotation\\
    &<\frac{2\privacythreshold\bra{t_0+1}}{\abs{\bra{1-\alpha}t_0-\alpha}\seclen} +  \frac{2\seclen \cdot t_0}{\bra{t_0+1}\bra{\mupperbound-\mlowerbound}} + \frac{\seclen^2}{2\bra{\mupperbound-\mlowerbound}^2}.
\end{align*}
The second term $\frac{2\seclen \cdot t_0}{\bra{t_0+1}\bra{\mupperbound-\mlowerbound}}$ bounds the probability of the yellow region except for the right-bottom triangle, and the last term $\frac{\seclen^2}{2\bra{\mupperbound-\mlowerbound}^2}$ is the probability of the right-bottom triangle.

For the \distortion{},
it is straightforward to get that $\distortionnotation = \frac{\bra{t_0^2+1}\seclen}{4(t_0+1)^2}$ from \cref{eq:lowerbound_uniform_quantile_d}, and $\distortionnotation_{\text{opt}}>\bra{\ceil{\frac{1}{\privacynotation{}}}-1}\cdot 2\ratio\privacythreshold \geq 2\ratio\privacythreshold$, where $\ratio$ is defined in \cref{thm:lowerbound_continuous_quantile_more}.
We can get that $\bra{\privacynotation - \frac{2\seclen \cdot t_0}{\bra{t_0+1}\bra{\mupperbound-\mlowerbound}} - \frac{\seclen^2}{2\bra{\mupperbound-\mlowerbound}^2}} \cdot \distortionnotation < 2\ratio\privacythreshold$ and 
\small
\begin{align*}
&\distortionnotation  = \distortionnotation_{\text{opt}} + \distortionnotation - \distortionnotation_{\text{opt}}\\
& < \distortionnotation_{\text{opt}} + \distortionnotation - \bra{\ceil{\frac{1}{\privacynotation}}-1}\cdot 2\ratio\privacythreshold\\
&\leq \distortionnotation_{\text{opt}} + 2\ratio\privacythreshold + \distortionnotation - {{\frac{2\ratio\privacythreshold}{\privacynotation}}}\\
&< \distortionnotation_{\text{opt}} + 2\ratio\privacythreshold + \frac{\frac{2\seclen \cdot t_0}{\bra{t_0+1}\bra{\mupperbound-\mlowerbound}} + \frac{\seclen^2}{2\bra{\mupperbound-\mlowerbound}^2}}{\frac{2\privacythreshold\bra{t_0+1}}{\abs{\bra{1-\alpha}t_0-\alpha}\seclen} +  \frac{2\seclen \cdot t_0}{\bra{t_0+1}\bra{\mupperbound-\mlowerbound}} + \frac{\seclen^2}{2\bra{\mupperbound-\mlowerbound}^2}}\cdot\distortionnotation\\
&=\bra{1+\frac{\abs{\bra{1-\alpha}t_0-\alpha}\seclen}{2\privacythreshold\bra{t_0+1}} \bra{  \frac{2\seclen \cdot t_0}{\bra{t_0+1}\bra{\mupperbound-\mlowerbound}} + \frac{\seclen^2}{2\bra{\mupperbound-\mlowerbound}^2}}} \bra{\distortionnotation_{\text{opt}} + 2\ratio\privacythreshold } \\
&\leq \bra{2+\frac{\abs{\bra{1-\alpha}t_0-\alpha}\seclen}{\privacythreshold\bra{t_0+1}}\cdot \bra{  \frac{2\seclen \cdot t_0}{\bra{t_0+1}\bra{\mupperbound-\mlowerbound}} + \frac{\seclen^2}{2\bra{\mupperbound-\mlowerbound}^2}}} \distortionnotation_{\text{opt}}.
\end{align*}
\normalsize

When $ {\frac{\seclen^2}{\mupperbound-\mlowerbound}}\rightarrow 0$ as $\seclen, (\mupperbound-\mlowerbound)\to\infty$, we can get that $ {\frac{\seclen^3}{\bra{\mupperbound-\mlowerbound}^2}}\rightarrow 0$. Therefore, in this case, $\lim{\sup_{{\frac{\seclen^2}{\mupperbound-\mlowerbound}}\rightarrow 0}} \distortionnotation < 3  \distortionnotation_{\text{opt}}$.

\end{proof}

\subsection{\Privacy{}-\Distortion{} Performance of \cref{mech:quantile_Gaussian_uniform_Gaussian_uniform} with Relaxed Assumption}
\label{sec:proof_trade-off_mechanism_quantile_continuous_relaxed_Gaussian}

For Gaussian distribution, we relax \cref{assu:quantile_continuous_Gaussian_uniform} as follows.

\begin{assumption}
    \label{assu:quantile_Gaussian_continuous}
    The prior over Gaussian distribution parameters  
    satisfies
    $\support{{\text{\textmugreek}, \text{\textsigma{}}}} = \brc{(a,b)| a\in\brba{\mulowerbound, \muupperbound},b\in\brba{\sigmalowerbound, \sigmaupperbound}}$,
    $\pdfof{\text{\textmugreek}, \text{\textsigma{}}}\bra{a, b} = \pdfof{\text{\textmugreek}}\bra{a}\cdot \pdfof{\text{\textsigma{}}}\bra{b}$, and $\pdfof{\text{\textmugreek}}\bra{a}$ (resp. $\pdfof{\text{\textsigma{}}}\bra{b}$) is $\mathcal{L}_{\mu}$-Lipschitz (resp. $\mathcal{L}_{\sigma}$-Lipschitz) and has lower bound $\frac{k_{\mu}}{\muupperbound-\mulowerbound}$ with $k_{\mu} \in (0,1]$ (resp. $\frac{k_{\sigma}}{\sigmaupperbound-\sigmalowerbound}$ with $k_{\sigma} \in (0,1]$).
\end{assumption}

Based on \cref{assu:quantile_Gaussian_continuous}, the \Privacy{}-\distortion{} performance of \cref{mech:quantile_Gaussian_uniform_Gaussian_uniform} is shown below.

\begin{proposition}
    \label{thm:upperbound_Gaussian_quantile}
    Under \cref{assu:quantile_Gaussian_continuous},
    \cref{mech:quantile_Gaussian_uniform_Gaussian_uniform} %
    has the following $\distortionnotation$ and $\privacynotation{}$ value/bound:
        \begin{align*}
        \distortionnotation %
        & = \frac{\seclen}{2}\sqrt{\frac{2}{\pi}}e^{-\frac{1}{2}t_0^2} - \frac{t_0 \seclen}{2}\bra{1-2\Phi\bra{t_0}},\\
        \privacynotation %
        & < \frac{\frac{2\privacythreshold}{\abs{t_0+Q_{\alpha}}}
\cdot \brb{\underline{c}+\mathcal{L}_{\mu, \sigma}\bra{\frac{\seclen}{2}-t^*-\frac{\privacythreshold}{\abs{t_0+Q_{\alpha}}}}}}{\underline{c}\seclen+\frac{\mathcal{L}_{\mu, \sigma}}{2}\bra{\frac{\seclen}{2}-t^*}^2} + \\
&M\bra{ 
 \muupperbound-\mulowerbound, \frac{k_{\mu}}{\muupperbound-\mulowerbound}, \mathcal{L}_\mu, 1}\cdot M\bra{ \sigmaupperbound-\sigmalowerbound, \frac{k_{\sigma}}{\sigmaupperbound-\sigmalowerbound}, \mathcal{L}_{\sigma}, 1}\cdot
 \bra{\sigmaupperbound-\sigmalowerbound}\abs{t_0}\seclen,
        \end{align*}
where $
\underline{c} = \frac{k_{\mu}k_{\sigma}}{\bra{\muupperbound-\mulowerbound}\cdot\bra{\sigmaupperbound-\sigmalowerbound}}
$, function $M$ satisfies
$$
 M\bra{x,  c, \mathcal{L}, \mathcal{A}} = 
\begin{cases}
\frac{\mathcal{A}}{x} + \frac{\mathcal{L}x}{2}, & \text{if } c \leq \frac{\mathcal{A}}{x} - \frac{\mathcal{L}x}{2}\\
c + \sqrt{2\mathcal{L}\bra{\mathcal{A}-cx}}, & \text{if } c > \frac{\mathcal{A}}{x} - \frac{\mathcal{L}x}{2}
\end{cases},
$$
$
\mathcal{L}_{\mu, \sigma} = \mathcal{L}_{\sigma}\cdot M\bra{ \frac{\muupperbound-\mulowerbound}{\abs{t_0}}, \frac{k_{\mu}}{\muupperbound-\mulowerbound}, \abs{t_0}\mathcal{L}_\mu, \frac{1}{\abs{t_0}}} + \abs{t_0} \mathcal{L}_{\mu} \cdot M\bra{ \sigmaupperbound-\sigmalowerbound, \frac{k_{\sigma}}{\sigmaupperbound-\sigmalowerbound}, \mathcal{L}_{\sigma}, 1}
$, and
$t^* = \frac{\seclen}{2} + \frac{\underline{c}}{\mathcal{L}_{\mu, \sigma}} - \frac{\privacythreshold}{\abs{t_0+Q_{\alpha}}} - \sqrt{\bra{\frac{\underline{c}}{\mathcal{L}_{\mu, \sigma}} - \frac{\privacythreshold}{\abs{t_0+Q_{\alpha}}}}^2 + \frac{2\underline{c}\seclen}{\mathcal{L}_{\mu, \sigma}}}$.
\end{proposition}

\begin{proof}

It is straightforward to get the formula for $\distortionnotation$ from \cref{eq:lowerbound_gaussian_quantile_d}. Here we focus on the proof for $\privacynotation$.

Similar to \cref{sec:upperbound_continuous_quantile_shifted_exponential_relaxed}, based on \cref{lemma:max_fraction_Lipschitz_mean} and \cref{lemma:max_value_Lipschitz}, we can get that 
\begin{align*}
&\sup_{\rvparamnotation\in S_{green}}\probof{ \secretestimatestarof{\releaservparamnotation}\in\brb{ \secretofparam - \privacythreshold, \secretofparam + \privacythreshold } }\\
&=
\sup_{i\in \mathbb{N}, \mu, t'\in \mathbb{R}%
} \frac{\int_{\max\brc{-\frac{\seclen}{2}, t'}}^{\min\brc{\frac{\seclen}{2}, t'+\frac{2\privacythreshold}{\abs{t_0+Q_{\alpha}}}}}  \pdfof{\text{\textmugreek}, \text{\textsigma{}}}\bra{\mu+t_0\cdot t, \sigmalowerbound + \bra{\privateparamindexnotation+0.5}\cdot\seclen+t} \mathrm{d}t}
{\int_{-\frac{\seclen}{2}}^{\frac{\seclen}{2}} \pdfof{\text{\textmugreek}, \text{\textsigma{}}}\bra{\mu+t_0\cdot t, \sigmalowerbound + \bra{\privateparamindexnotation+0.5}\cdot\seclen+t}\mathrm{d}t}\\
& \leq 
\frac{\frac{2\privacythreshold}{\abs{t_0+Q_{\alpha}}}
\cdot \brb{\underline{c}+\mathcal{L}_{\mu, \sigma}\bra{\frac{\seclen}{2}-t^*-\frac{\privacythreshold}{\abs{t_0+Q_{\alpha}}}}}}{\underline{c}\seclen+\frac{\mathcal{L}_{\mu, \sigma}}{2}\bra{\frac{\seclen}{2}-t^*}^2},
\end{align*}
where $t^* = \frac{\seclen}{2} + \frac{\underline{c}}{\mathcal{L}_{\mu, \sigma}} - \frac{\privacythreshold}{\abs{t_0+Q_{\alpha}}} - \sqrt{\bra{\frac{\underline{c}}{\mathcal{L}_{\mu, \sigma}} - \frac{\privacythreshold}{\abs{t_0+Q_{\alpha}}}}^2 + \frac{2\underline{c}\seclen}{\mathcal{L}_{\mu, \sigma}}}$, $
\mathcal{L}_{\mu, \sigma} = \mathcal{L}_{\sigma}\cdot M\bra{ \frac{\muupperbound-\mulowerbound}{\abs{t_0}}, \frac{k_{\mu}}{\muupperbound-\mulowerbound}, \abs{t_0}\mathcal{L}_\mu, \frac{1}{\abs{t_0}}} + \abs{t_0} \mathcal{L}_{\mu} \cdot M\bra{ \sigmaupperbound-\sigmalowerbound, \frac{k_{\sigma}}{\sigmaupperbound-\sigmalowerbound}, \mathcal{L}_{\sigma}, 1}
$,  and 
$
\underline{c} = \frac{k_{\mu}k_{\sigma}}{\bra{\muupperbound-\mulowerbound}\cdot\bra{\sigmaupperbound-\sigmalowerbound}}
$.

As for $\int_{ \rvparamnotation\in S_{yellow}}p(\rvparamnotation)d\rvparamnotation$, we have
\begin{align*}
&\int_{ \rvparamnotation\in S_{yellow}}p(\rvparamnotation)d\rvparamnotation \\
&\leq 
 M\bra{ 
 \muupperbound-\mulowerbound, \frac{k_{\mu}}{\muupperbound-\mulowerbound}, \mathcal{L}_\mu, 1}\cdot M\bra{ \sigmaupperbound-\sigmalowerbound, \frac{k_{\sigma}}{\sigmaupperbound-\sigmalowerbound}, \mathcal{L}_{\sigma}, 1}\cdot
 \int_{ \rvparamnotation\in S_{yellow}}d\rvparamnotation\\
& = 
M\bra{ 
 \muupperbound-\mulowerbound, \frac{k_{\mu}}{\muupperbound-\mulowerbound}, \mathcal{L}_\mu, 1}\cdot M\bra{ \sigmaupperbound-\sigmalowerbound, \frac{k_{\sigma}}{\sigmaupperbound-\sigmalowerbound}, \mathcal{L}_{\sigma}, 1}\cdot
 \bra{\sigmaupperbound-\sigmalowerbound}\abs{t_0}\seclen.
\end{align*}

Above all, we can get that
\begin{align*}
    \privacynotation{} 
    &< \sup_{\rvparamnotation\in S_{green}}\probof{ \secretestimatestarof{\releaservparamnotation}\in\brb{ \secretofparam - \privacythreshold, \secretofparam + \privacythreshold } } + \int_{ \rvparamnotation\in S_{yellow}}p(\rvparamnotation)d\rvparamnotation.
    \\
    & \leq \frac{\frac{2\privacythreshold}{\abs{t_0+Q_{\alpha}}}
\cdot \brb{\underline{c}+\mathcal{L}_{\mu, \sigma}\bra{\frac{\seclen}{2}-t^*-\frac{\privacythreshold}{\abs{t_0+Q_{\alpha}}}}}}{\underline{c}\seclen+\frac{\mathcal{L}_{\mu, \sigma}}{2}\bra{\frac{\seclen}{2}-t^*}^2} +  \\
&M\bra{ 
 \muupperbound-\mulowerbound, \frac{k_{\mu}}{\muupperbound-\mulowerbound}, \mathcal{L}_\mu, 1}\cdot M\bra{ \sigmaupperbound-\sigmalowerbound, \frac{k_{\sigma}}{\sigmaupperbound-\sigmalowerbound}, \mathcal{L}_{\sigma}, 1}\cdot
 \bra{\sigmaupperbound-\sigmalowerbound}\abs{t_0}\seclen,
\end{align*}
where $M(\cdot,\cdot,\cdot,\cdot), \underline{c}, \mathcal{L}_{\mu, \sigma}, t^*$ are defined as above.
\end{proof}
\section{Case Study with \Secret{} = Standard Deviation}
\label{sec:case_study_std_more}
In this section, we discuss how to protect standard deviation for several continuous and discrete distributions. 

\subsection{Continuous Distributions}
\label{sec:case_study_std_continuous}
We consider the same distributions discussed in \cref{sec:case_study_quantiles} and \cref{sec:case_study_Gaussian_uniform_quantiles}: Gaussian, uniform, and (shifted) exponential distributions.

\begin{corollary}[Privacy lower bound, secret = standard deviation of a continuous distribution]
\label{thm:lowerbound_continuous_std}
Consider the secret function $\secretof{\rvparamnotation}=$ standard deviation of $\privatepdf$. For any  $\privacymetricthreshold\in\bra{0,1}$, when $\privacynotation\leq \privacymetricthreshold$, we have $\distortionnotation> \bra{\ceil{\frac{1}{\privacymetricthreshold}}-1}\cdot 2\ratio\privacythreshold$, where the value of $\ratio$ depends on the type of the distributions:
\begin{packeditemize}
    \item Gaussian: 
    \begin{align*}
    \ratio = \min_{t}
    \sqrt{\frac{1}{2\pi}}e^{-\frac{1}{2}t^2}-t\bra{\frac{1}{2}-\Phi\bra{t}},
    \end{align*}
    where $\Phi$ denotes the CDF of the standard Gaussian distribution.
    \item Uniform: 
    $
        \ratio=\frac{\sqrt{3}}{4}.
    $
    \item Exponential:  
    $
        \ratio=\frac{1}{2}.
    $
    \item Shifted exponential: 
    $
    \ratio= \frac{\ln{2}}{2}.
    $
\end{packeditemize}
\end{corollary}
The proof is in \cref{sec:proof_lowerbound_continuous_std}.
The bounds for Gaussian can be computed numerically, while the bounds for all other distributions are  in closed form.

Next, we present the \datamechanism{} for these distributions and the secret under the same assumption as \cref{assu:quantile_continuous}.

\begin{mechanism}[For secret = standard deviation of a continuous distribution]
    \label{mech:std_continuous}
    We design mechanisms for each of the distributions.
    \begin{packeditemize}
        \item Gaussian: 
            \begin{align*}
\subsetofprivateparamof{\mu,\privateparamindexnotation} &= \brc{\bra{\mu+t_0\cdot t, \sigmalowerbound + \bra{\privateparamindexnotation+0.5}\cdot\seclen+t}| t\in \brba{-\frac{\seclen}{2}, \frac{\seclen}{2}}}
~~,\\
\releaseparamofindex{\mu,\privateparamindexnotation} &= \bra{\mu,\sigmalowerbound+\bra{\privateparamindexnotation+0.5}\cdot\seclen} ~~,\\
\privateparamindexsetnotation &=  \brc{\bra{\mu,\privateparamindexnotation}| \privateparamindexnotation\in\setofnaturalnumbers, \mu\in\setofreal},
\end{align*}
where $\seclen$ is a hyper-parameter of the mechanism that divides $\bra{\sigmaupperbound - \sigmalowerbound}$ and 
             \begin{align*}
                 t_0 = \arg\min_t \sqrt{\frac{1}{2\pi}}e^{-\frac{1}{2}t^2}-t\bra{\frac{1}{2}-\Phi\bra{t}} .
             \end{align*}. %
        \item Uniform:
        \begin{align*}
            \subsetofprivateparamof{m,\privateparamindexnotation} &= 
            \scalebox{0.8}{$
            \brc{\bra{m-t, m+\bra{\privateparamindexnotation+0.5}\cdot\seclen+t}| t\in \brab{-\frac{\seclen}{4}, \frac{\seclen}{4}}}
            $}
            ~~,\\
            \releaseparamofindex{m,\privateparamindexnotation} &= \bra{m,m+\bra{\privateparamindexnotation+0.5}\cdot\seclen} ~~,\\
            \privateparamindexsetnotation &=  \brc{\bra{m,\privateparamindexnotation}| \privateparamindexnotation\in\setofpositiveintegers, m\in\setofreal},
        \end{align*}
        where $\seclen>0$ is a hyper-parameter of the mechanism that divides $\bra{\mupperbound-\mlowerbound}$.
        \item Exponential:
        \begin{align*}
            \subsetofprivateparamof{\privateparamindexnotation} &= \brba{\lambdalowerbound + \privateparamindexnotation\cdot \seclen, \lambdalowerbound + \bra{\privateparamindexnotation+1}\cdot \seclen}
            ~~,\\
            \releaseparamofindex{\privateparamindexnotation} &= \lambdalowerbound+\bra{\privateparamindexnotation+0.5}\cdot\seclen~~,\\
            \privateparamindexsetnotation &= \setofnaturalnumbers,
        \end{align*}
        where $\seclen>0$ is a hyper-parameter of the mechanism that divides $\bra{\lambdaupperbound-\lambdalowerbound}$.
        \item Shifted exponential:
        \begin{align*}
            \subsetofprivateparamof{\privateparamindexnotation,h} &= \brc{\bra{\lambdalowerbound+\bra{\privateparamindexnotation+0.5}\seclen + t, h-\ln 2\cdot t}|  t\in\brba{-\frac{\seclen}{2}, \frac{\seclen}{2}}}
            ~~,\\
            \releaseparamofindex{\privateparamindexnotation,h} &= \bra{\lambdalowerbound+\bra{\privateparamindexnotation + 0.5}\seclen, h }~~,\\
            \privateparamindexsetnotation &= \brc{(\privateparamindexnotation, h)| \privateparamindexnotation\in\setofnaturalnumbers, h\in\setofreal},
        \end{align*}
        where $\seclen>0$ is a hyper-parameter of the mechanism that divides $\bra{\lambdaupperbound-\lambdalowerbound}$.
    \end{packeditemize}
\end{mechanism}

These \datamechanisms{} achieve the following $\distortionnotation$ and $\privacynotation$.
\begin{proposition}
    \label{thm:upperbound_continuous_std}
    Under \cref{assu:quantile_continuous},
    \cref{mech:std_continuous} %
    has the following $\distortionnotation$ and $\privacynotation{}$ value/bound.
    \begin{packeditemize}
        \item Gaussian: 
        \begin{align*}
        \privacynotation %
        & < \frac{2\privacythreshold}{\seclen} + \frac{\abs{t_0}\seclen}{\muupperbound-\mulowerbound},\\
        \distortionnotation %
        & = \frac{\seclen}{2}\sqrt{\frac{2}{\pi}}e^{-\frac{1}{2}t_0^2} - \frac{t_0 \seclen}{2}\bra{1-2\Phi\bra{t_0}}< \bra{2+\frac{\abs{t_0}\seclen^2}{\bra{\muupperbound-\mulowerbound}\epsilon}} \distortionnotation_{\text{opt}},
        \end{align*}
        where $t_0$ is defined in \cref{mech:std_continuous}. Under the ``high-precision'' regime where $ {\frac{\seclen^2}{\muupperbound-\mulowerbound}}\rightarrow 0$ as $\seclen, (\muupperbound-\mulowerbound)\to\infty$, $\distortionnotation$ satisfies
\begin{align*}
    \lim{\sup_{{\frac{\seclen^2}{\muupperbound-\mulowerbound}}\rightarrow 0}} \distortionnotation < 3  \distortionnotation_{\text{opt}}.
\end{align*}
        \item Uniform:
        \begin{align*}
        \privacynotation &<  \frac{4\sqrt{3}\privacythreshold}{\seclen} +  \frac{\seclen}{\bra{\mupperbound-\mlowerbound}} + \frac{\seclen^2}{2\bra{\mupperbound-\mlowerbound}^2},\\
        \distortionnotation &= \frac{\seclen}{8}<\bra{2+\frac{\seclen}{2\sqrt{3}\privacythreshold}\cdot \bra{  \frac{\seclen }{{\mupperbound-\mlowerbound}} + \frac{\seclen^2}{2\bra{\mupperbound-\mlowerbound}^2}}} \distortionnotation_{\text{opt}}.
        \end{align*}
        Under the ``high-precision'' regime where ${\frac{\seclen^2}{\mupperbound-\mlowerbound}}\rightarrow 0$ as $\seclen, (\mupperbound-\mlowerbound)\to\infty$, $\distortionnotation$ satisfies
\begin{align*}
    \lim{\sup_{{\frac{\seclen^2}{\mupperbound-\mlowerbound}}\rightarrow 0}} \distortionnotation < 3  \distortionnotation_{\text{opt}}.
\end{align*}
        
        \item Exponential:
        \begin{align*}
        \privacynotation&=\frac{2\privacythreshold}{s},\\
        \distortionnotation&=\frac{1}{2}\seclen < 2\distortionnotation_{\text{opt}}.
        \end{align*}
        \item Shifted exponential:
        \begin{align*}
            \privacynotation %
            & < \frac{2\privacythreshold}{\seclen} + \frac{\seclen\ln 2 }{\hupperbound-\hlowerbound},\\
            \distortionnotation %
            & = \frac{\seclen\ln 2}{2} < \bra{2 +\frac{\seclen^2\ln 2 }{\privacythreshold\bra{\hupperbound-\hlowerbound}}}\distortionnotation_{\text{opt}}.
        \end{align*}
        Under the ``high-precision'' regime where ${\frac{\seclen^2}{\hupperbound-\hlowerbound}}\rightarrow 0$ as $\seclen, (\hupperbound-\hlowerbound)\to\infty$, $\distortionnotation$ satisfies
\begin{align*}
    \lim{\sup_{{\frac{\seclen^2}{\hupperbound-\hlowerbound}}\rightarrow 0}} \distortionnotation < 3  \distortionnotation_{\text{opt}}.
\end{align*}
    \end{packeditemize}
\end{proposition}

The proof is in \cref{sec:proof_upperbound_continuous_std}. 
For Gaussian,  exponential and shifted exponential distributions, we relax \cref{assu:quantile_continuous} and analyze the \privacy{}-\distortion{} performance of \cref{mech:std_continuous} in \cref{sec:proof_trade-off_mechanism_std_continuous_relaxed}. 
From these propositions, we have similar takeaways as the alpha-quantile case (
\cref{sec:case_study_quantiles}): (1) data holder can use $\seclen$ to control the trade-off between \distortion{} and \privacy{}, and (2) the mechanism is order-optimal under the ``high-precision'' regime.

\subsection{Discrete Distributions}
\label{sec:case_study_std_discrete}
Here, we consider the same discrete distributions studied in \cref{sec:case_study_mean_discrete}: Geometric distributions,   binomial distributions, and Poisson distributions. We first analyze the lower bound.

\begin{corollary}[Privacy lower bound, secret = standard deviation of a discrete distribution]
\label{thm:discrete_std_more}
Consider the secret function $\secretof{\rvparamnotation}=$ standard deviation of $\privatepdf$. For any  $\privacymetricthreshold\in\bra{0,1}$, when $\privacynotation\leq \privacymetricthreshold$, we have $\distortionnotation> \bra{\ceil{\frac{1}{\privacymetricthreshold}}-1}\cdot 2\ratio\privacythreshold$, where the value of $\ratio$ depends on the type of the distributions:
\begin{packeditemize}
    \item Geometric: 
    \begin{align*}
    \ratio= 
\inf_{\rvparamlowerbound< \rvparamnotation_1<\rvparamnotation_2\leq \rvparamupperbound}\frac{\bra{1-\rvparamnotation_2}^{h\bra{\rvparamnotation_1,\rvparamnotation_2}} - \bra{1-\rvparamnotation_1}^{h\bra{\rvparamnotation_1,\rvparamnotation_2}}}{2\bra{\frac{\sqrt{1-\rvparamnotation_2}}{\rvparamnotation_2} - \frac{\sqrt{1-\rvparamnotation_1}}{\rvparamnotation_1}}}~~,
\end{align*}
where $h\bra{\rvparamnotation_1,\rvparamnotation_2} = \floor{  \frac{\log\bra{\rvparamnotation_2} - \log\bra{\rvparamnotation_1}}{ \log\bra{1 - \rvparamnotation_1} - \log\bra{1 - \rvparamnotation_2 }  }  } + 1$. 
    \item Binomial: 
    \begin{align*}
    &\ratio=\inf_{\rvparamlowerbound< \rvparamnotation_1<\rvparamnotation_2\leq \rvparamupperbound} \\
    &~~\scalebox{1}{$
    \frac{I_{1-\rvparamnotation_2}\bra{n-h\bra{\rvparamnotation_1,\rvparamnotation_2}, 1+h\bra{\rvparamnotation_1,\rvparamnotation_2}}-I_{1-\rvparamnotation_1}\bra{n-h\bra{\rvparamnotation_1,\rvparamnotation_2}, 1+h\bra{\rvparamnotation_1,\rvparamnotation_2}}}{2\babs{\sqrt{n\rvparamnotation_2\bra{1-\rvparamnotation_2}}-\sqrt{n\rvparamnotation_1\bra{1-\rvparamnotation_1}}}},
    $}
    \end{align*}
    where $h\bra{\rvparamnotation_1,\rvparamnotation_2}=\floor{k'}$, $k' = n\ln\bra{\frac{1-\theta_2}{1-\theta_1}} \Big/ \ln\bra{{\frac{\theta_1\bra{1-\theta_2}}{\theta_2\bra{1-\theta_1}}}}$, %
    and $I$ represents the regularized incomplete beta function.
    \item Poisson: 
    \begin{align*}
    \ratio = 
    \inf_{\rvparamlowerbound< \rvparamnotation_1<\rvparamnotation_2\leq \rvparamupperbound}\frac{Q\bra{h\bra{\rvparamnotation_1,\rvparamnotation_2}, \rvparamnotation_2} - Q\bra{h\bra{\rvparamnotation_1,\rvparamnotation_2}, \rvparamnotation_1}}{2\bra{\sqrt{\rvparamnotation_1}-\sqrt{\rvparamnotation_2}}},
\end{align*}
where $h\bra{\rvparamnotation_1,\rvparamnotation_2} = \floor{\frac{\rvparamnotation_1 - \rvparamnotation_2}{\ln\bra{\rvparamnotation_1}-\ln\bra{\rvparamnotation_2}}} + 1$ and $Q$ is the regularized gamma function.
\end{packeditemize}
\end{corollary}
The proof is in \cref{sec:proof_lowerbound_discrete_std_more}. The above lower bounds can be computed numerically.

Since these distributions only have one parameter, we can use \cref{alg:dp} and \cref{alg:greedy} to derive a \datamechanism{}. 
The performance of greedy-based and dynamic-programming-based \datamechanisms{} for each distribution is shown in \cref{fig:algorithm_std_more}.

\begin{figure*}[th]
    \centering

    \begin{subfigure}{0.3\textwidth}
         \centering
        \includegraphics[width=1\linewidth]{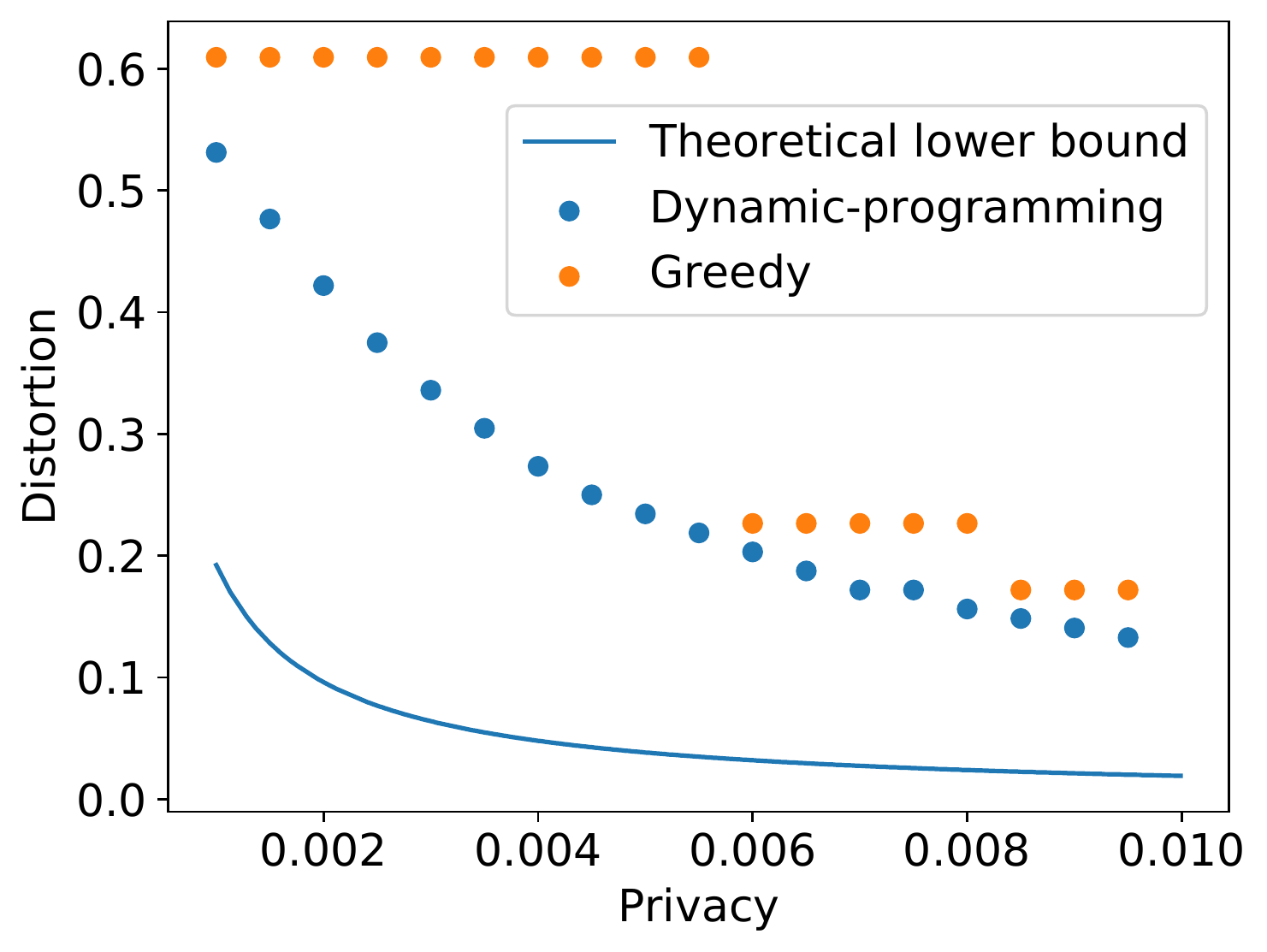}
         \caption{Distribution = Geometric}
         \label{fig:Geometric_std}
     \end{subfigure}
     \hfill
    \begin{subfigure}{0.3\textwidth}
         \centering
        \includegraphics[width=1\linewidth]{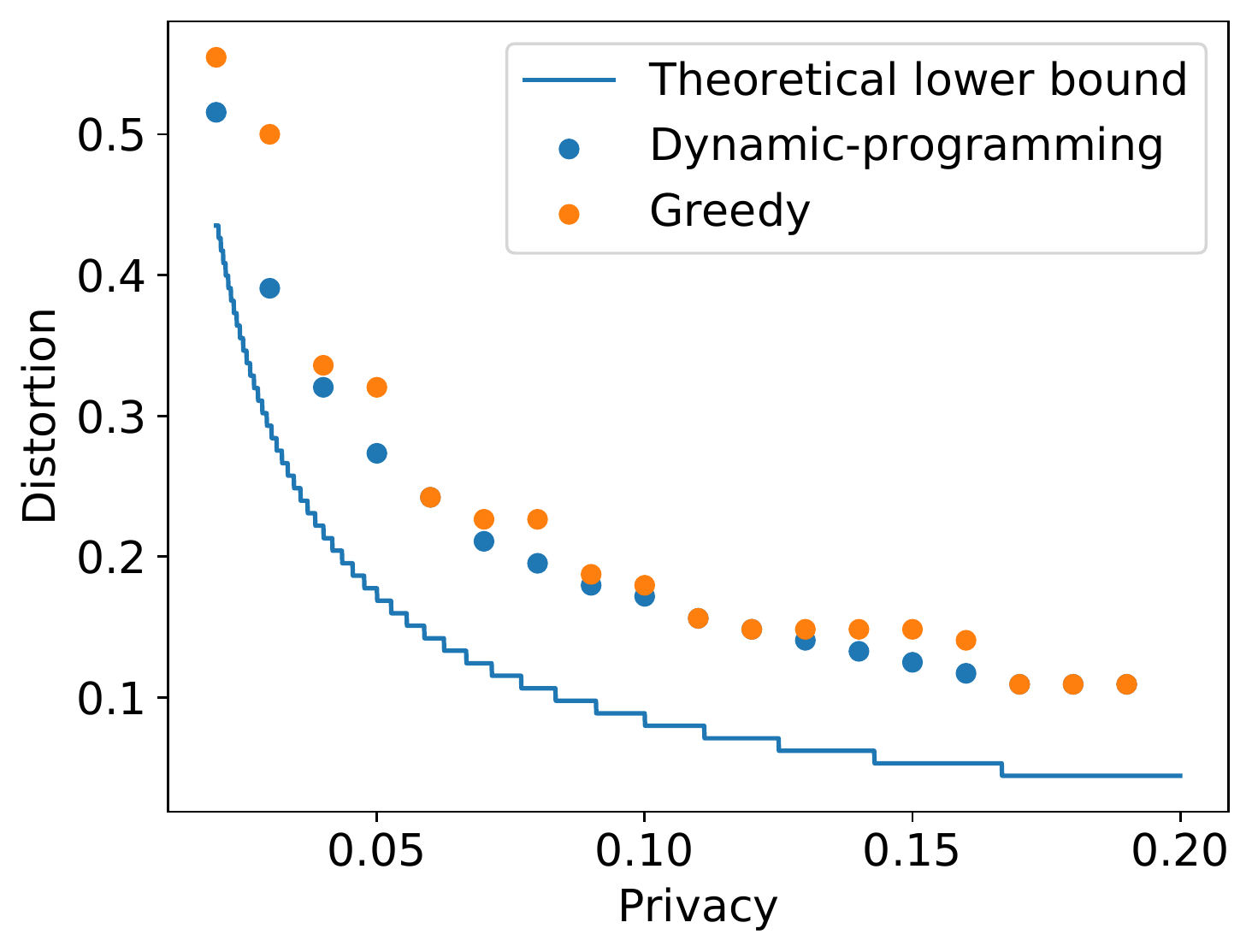}
         \caption{Distribution = Binomial}
         \label{fig:Binomial_std}
     \end{subfigure}
     \hfill
    \begin{subfigure}{0.3\textwidth}
         \centering
        \includegraphics[width=1\linewidth]{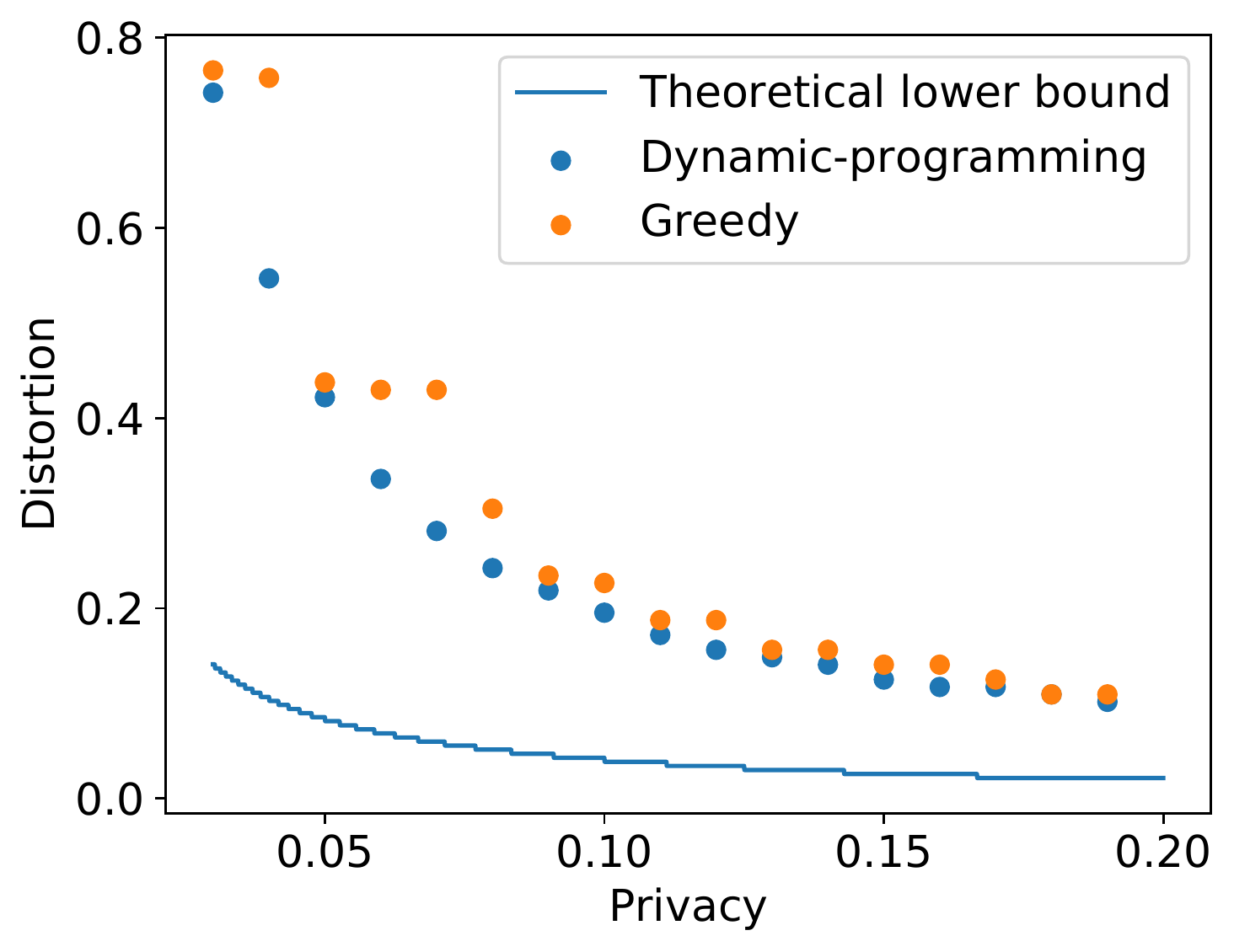}
         \caption{Distribution = Poisson}
         \label{fig:Poisson_std}
     \end{subfigure}
    \caption{ \Privacy{}-\distortion{} performance of \cref{alg:dp} and \cref{alg:greedy} for binomial and Poisson distribution when secret = standard deviation.
    }
    \label{fig:algorithm_std_more}
\end{figure*}

\subsection{Proof of %
\cref{thm:lowerbound_continuous_std}}
\label{sec:proof_lowerbound_continuous_std}
\subsubsection{Gaussian Distribution}
\label{sec:proof_lowerbound_continuous_std_GAUSSIAN}
\begin{proof}
    Let $\rvprivatewithparam{\mu_1,\sigma_2}, \rvprivatewithparam{\mu_2,\sigma_2} $ be two Gaussian random variables with means $\mu_1,\mu_2$ and sigmas $\sigma_1,\sigma_2$ respectively, where $\sigma_1\not= \sigma_2$. Let $\Phi$ denotes the CDF of the standard Gaussian distribution.
We can get that
\begin{align*}
\auxdistance{\rvprivatewithparam{\mu_1,\sigma_1}}{\rvprivatewithparam{\mu_2,\sigma_2}}&= 
\bra{\mu_1-\mu_2}\bra{\Phi\bra{\frac{\mu_1-\mu_2}{\sigma_2-\sigma_1}}-\frac{1}{2}}\\
&\quad +\sqrt{\frac{1}{2\pi}}\bra{\sigma_2-\sigma_1}e^{-\frac{1}{2}\bra{\frac{\mu_1-\mu_2}{\sigma_2-\sigma_1}}^2},
\\
\auxrange{\rvprivatewithparam{\mu_1,\sigma_1}}{\rvprivatewithparam{\mu_2,\sigma_2}}&= \abs{\sigma_1-\sigma_2}.
\end{align*}
Let $\frac{\mu_1-\mu_2}{\sigma_1-\sigma_2}\triangleq t$, we can get that
\begin{align*}
\frac{\auxdistance{\rvprivatewithparam{\mu_1,\sigma_1}}{\rvprivatewithparam{\mu_2,\sigma_2}}}{\auxrange{\rvprivatewithparam{\mu_1,\sigma_1}}{\rvprivatewithparam{\mu_2,\sigma_2}}} &= \sqrt{\frac{1}{2\pi}}e^{-\frac{1}{2}t^2}-t\bra{\frac{1}{2}-\Phi\bra{t}}\triangleq h\bra{t}.
\end{align*}
Therefore we can get that
\begin{align*}
\ratio = \min_{t}h\bra{t}.
\end{align*}
\end{proof}

\subsubsection{Uniform Distribution}
\label{sec:proof_lowerbound_continuous_std_UNIFORM}
\begin{proof}

Let $\rvprivatewithparam{m_1,n_1}, \rvprivatewithparam{m_2,n_2}$ be two uniform random variables. Let $\cdfof{\rvprivatewithparam{m_1,n_1}}, \cdfof{\rvprivatewithparam{m_2,n_2}}$ be their CDFs, and let $m_2\geq m_1$ without loss of generality. We can get that
\begin{align*}
\auxdistance{\rvprivatewithparam{m_1,n_1}}{\rvprivatewithparam{m_2,n_2}}
&= \distanceformulawass{\rvprivatewithparam{m_1,n_1}}{\rvprivatewithparam{m_2,n_2}}\\
&= \frac{1}{2} \int_{-\infty}^{+\infty} \abs{\cdfof{\rvprivatewithparam{m_1,n_1}}\bra{x}-\cdfof{\rvprivatewithparam{m_2,n_2}}\bra{x}} \mathrm{d}x\\
&=\begin{cases}
\frac{m_2-m_1+n_2-n_1}{4} & n_2\geq n_1\\
\frac{\bra{m_2-m_1}^2+\bra{n_1-n_2}^2}{4\bra{m_2-m_1+\bra{n_1-n_2}}} & n_2<n_1
\end{cases},
\\
\auxrange{\rvprivatewithparam{m_1,n_1}}{\rvprivatewithparam{m_2,n_2}}&= \left|\frac{1}{\sqrt{12}}\bra{n_1-m_1}-\frac{1}{\sqrt{12}}\bra{n_2-m_2}\right|\\
&= \frac{1}{\sqrt{12}}\abs{m_2-m_1-\bra{n_2-n_1}}.
\end{align*}
Therefore, we can get that when $n_2 \geq n_1$, we have
\begin{align*}
\frac{\auxdistance{\rvprivatewithparam{m_1,n_1}}{\rvprivatewithparam{m_2,n_2}}}{\auxrange{\rvprivatewithparam{m_1,n_1}}{\rvprivatewithparam{m_2,n_2}}} &=
\frac{\sqrt{3}}{2} \frac{m_2-m_1+n_2-n_1}{\abs{m_2-m_1-\bra{n_2-n_1}}}\\
& \geq \frac{\sqrt{3}}{2}.
\end{align*}
When $n_2<n_1$, we have
\begin{align*}
&\frac{\auxdistance{\rvprivatewithparam{m_1,n_1}}{\rvprivatewithparam{m_2,n_2}}}{\auxrange{\rvprivatewithparam{m_1,n_1}}{\rvprivatewithparam{m_2,n_2}}} =
\frac{\sqrt{3}}{2} \frac{\bra{m_2-m_1}^2+\bra{n_1-n_2}^2}{\bra{m_2-m_1+\bra{n_1-n_2}}^2}\\
& \quad= \frac{\sqrt{3}}{2}\frac{\bra{m_2-m_1}^2+\bra{n_1-n_2}^2}{\bra{m_2-m_1}^2+\bra{n_1-n_2}^2+2\bra{m_2-m_1}\bra{n_1-n_2}}\\
& \quad\geq \frac{\sqrt{3}}{2}\cdot\frac{\bra{m_2-m_1}^2+\bra{n_1-n_2}^2}{2\brb{\bra{m_2-m_1}^2+\bra{n_1-n_2}^2}}\\
& \quad= \frac{\sqrt{3}}{4}.
\end{align*}
Therefore we can get that
\begin{align*}
\nonumber
\ratio = \frac{\sqrt{3}}{4}.
\end{align*}
    
\end{proof}
\subsubsection{Exponential Distribution}
\label{sec:proof_lowerbound_continuous_std_EXP}
\begin{proof}
Let $\rvprivatewithparam{\lambda_1},\rvprivatewithparam{\lambda_2}$ be two exponential random variables.
We have
\begin{align*}
\frac{\auxdistance{\rvprivatewithparam{\lambda_1}}{\rvprivatewithparam{\lambda_2}}}{\auxrange{\rvprivatewithparam{\lambda_1}}{\rvprivatewithparam{\lambda_2}}} &= \frac{\frac{1}{
\lambda_1}-\frac{1}{\lambda_2}}{2\bra{\frac{1}{
\lambda_1}-\frac{1}{\lambda_2}}}
=\frac{1}{2}.
\end{align*}
Therefore we can get that
\begin{align*}
\ratio = \frac{1}{2}.
\end{align*}
    
\end{proof}

\subsubsection{Shifted Exponential Distribution}
\label{sec:proof_lowerbound_continuous_std_SHIFTED_EXP}
\begin{proof}

Let $\rvprivatewithparam{\lambda_1,h_1}, \rvprivatewithparam{\lambda_2,h_2}$ be random variables from shifted exponential distributions.
Let $\lambda_2\leq \lambda_1$ without loss of generality.
Let $a=\frac{\lambda_1}{\lambda_2}$ and $b=\bra{h_1/\lambda_1-h_2/\lambda_2}{\lambda_2}$.
We can get that $\pdfof{\rvprivatewithparam{\lambda_1,h_1}}\bra{x}=a\pdfof{\rvprivatewithparam{\lambda_2,h_2}}\bra{a\bra{x+b}}$, and 
\begin{align*}
&\auxdistance{\rvprivatewithparam{\lambda_1,h_1}}{\rvprivatewithparam{\lambda_2,h_2}}= 
\distanceformulawass{\rvprivatewithparam{\lambda_1,h_1}}{\rvprivatewithparam{\lambda_2,h_2}}\\
&= \frac{1}{2}\int_{h_1}^{+\infty}\brd{x-\bra{\frac{x}{a}-b}}\pdfof{\rvprivatewithparam{\lambda_1,h_1}}\bra{x}\mathrm{d}x\\
&=\frac{\lambda_2}{2\lambda_1}\int_{h_1}^{+\infty}\brd{\bra{1/\lambda_2-1/\lambda_1}x+h_1/\lambda_1-h_2/\lambda_2}e^{-\frac{1}{\lambda_1} \bra{x-h_1}}\mathrm{d}x\\
&=\begin{cases}
\frac{1}{2}\bra{h_2-h_1+{\lambda_2}-{\lambda_1}}-e^{\frac{h_2-h_1}{{
\lambda_2}-{\lambda_1}}}\bra{{
\lambda_2}-{\lambda_1}} & \bra{h_1<h_2}\\
\frac{1}{2}\bra{h_1-h_2+{
\lambda_1}-{\lambda_2}} & \bra{h_1\geq h_2}
\end{cases},\\
&\auxrange{\rvprivatewithparam{\lambda_1,h_1}}{\rvprivatewithparam{\lambda_2,h_2}}
= {
\lambda_1}-{\lambda_2}.
\numberthis\label{eq:lowerbound_shifted_exp_std_d}
\end{align*}

When $\lambda_1=\lambda_2$ and $h_1\not=h_2$, we have $\frac{\auxdistance{\rvprivatewithparam{\lambda_1,h_1}}{\rvprivatewithparam{\lambda_2,h_2}}}{\auxrange{\rvprivatewithparam{\lambda_1,h_1}}{\rvprivatewithparam{\lambda_2,h_2}}}=\infty$.

When $\lambda_1\not=\lambda_2$ and $h_1< h_2$, let $t = \frac{h_2-h_1}{{
\lambda_1}-{\lambda_2}}\in(0,+\infty)$. 
We have
\begin{align*}
\frac{\auxdistance{\rvprivatewithparam{\lambda_1,h_1}}{\rvprivatewithparam{\lambda_2,h_2}}}{\auxrange{\rvprivatewithparam{\lambda_1,h_1}}{\rvprivatewithparam{\lambda_2,h_2}}}
&= \frac{h_2-h_1+{\lambda_2}-{\lambda_1}-2e^{\frac{h_2-h_1}{{
\lambda_2}-{\lambda_1}}}\bra{{
\lambda_2}-{\lambda_1}}}{2\bra{{
\lambda_1}-{\lambda_2}}}\\
& = \frac{t+2e^{-t}-1}{2}\\
& \geq \frac{\ln 2}{2}.
\end{align*}
``$=$'' achieves when $t = t_0 = \ln 2$.

When $\lambda_1\not=\lambda_2$ and $h_1\geq h_2$, we have
\begin{align*}
\frac{\auxdistance{\rvprivatewithparam{\lambda_1,h_1}}{\rvprivatewithparam{\lambda_2,h_2}}}{\auxrange{\rvprivatewithparam{\lambda_1,h_1}}{\rvprivatewithparam{\lambda_2,h_2}}}
&= \frac{h_1-h_2+{
\lambda_1}-{\lambda_2}}{2\bra{{
\lambda_1}-{\lambda_2}}}
\geq \frac{{
\lambda_1}-{\lambda_2}}{2\bra{{
\lambda_1}-{\lambda_2}}}
=\frac{1}{2}.
\end{align*}

Therefore we can get that
\begin{equation}
\nonumber
\ratio = \frac{\ln 2}{2}.
\end{equation}
\end{proof}
\subsection{Proof of %
\cref{thm:upperbound_continuous_std}}
\label{sec:proof_upperbound_continuous_std}

The proof outline is almost the same as the ones in \cref{sec:proof_upperbound_continuous_quantile} and \cref{sec:proof_upperbound_continuous_quantile_Gaussian_uniform}. We omit the details and point to the proof sections where we can adapt from.

\subsubsection{Gaussian Distribution}
The proof is the same as \cref{sec:proof_upperbound_continuous_quantile_GAUSSIAN}, except that we use the $\auxdistance{\cdot}{\cdot}$ and $\auxrange{\cdot}{\cdot}$ from \cref{sec:proof_lowerbound_continuous_std_GAUSSIAN}.

\subsubsection{Uniform Distribution}

The proof is the same as \cref{sec:proof_upperbound_continuous_quantile_UNIFORM}, except that we use the $\auxdistance{\cdot}{\cdot}$ and $\auxrange{\cdot}{\cdot}$ from \cref{sec:proof_lowerbound_continuous_std_UNIFORM}.

\subsubsection{Exponential Distribution}

The proof is the same as \cref{sec:proof_upperbound_continuous_quantile_EXP}, except that we use the $\auxdistance{\cdot}{\cdot}$ and $\auxrange{\cdot}{\cdot}$ from \cref{sec:proof_lowerbound_continuous_std_EXP}.

\subsubsection{Shifted Exponential Distribution}
The proof is the same as \cref{sec:proof_upperbound_continuous_quantile_SFHITED_EXP}, except that we use the $\auxdistance{\cdot}{\cdot}$ and $\auxrange{\cdot}{\cdot}$ from \cref{sec:proof_lowerbound_continuous_std_SHIFTED_EXP}.

\subsection{\Privacy{}-\Distortion{} Performance of \cref{mech:std_continuous} with Relaxed Assumption}
\label{sec:proof_trade-off_mechanism_std_continuous_relaxed}

Based on \cref{assu:quantile_Gaussian_continuous} and \cref{assu:quantile_continuous_relaxed}, the \Privacy{}-\distortion{} performance of \cref{mech:std_continuous} is shown below.

\begin{proposition}
    \label{thm:upperbound_continuous_std_relaxed}
    Under \cref{assu:quantile_Gaussian_continuous} and \cref{assu:quantile_continuous_relaxed},
    \cref{mech:std_continuous} %
    has the following $\distortionnotation$ and $\privacynotation{}$ value/bound.
    \begin{packeditemize}
    \item Gaussian: 
        \begin{align*}
        \distortionnotation %
        & = \frac{\seclen}{2}\sqrt{\frac{2}{\pi}}e^{-\frac{1}{2}t_0^2} - \frac{t_0 \seclen}{2}\bra{1-2\Phi\bra{t_0}},\\
        \privacynotation %
        & < \frac{{2\privacythreshold}
\cdot \brb{\underline{c}+\mathcal{L}_{\mu, \sigma}\bra{\frac{\seclen}{2}-t^*-{\privacythreshold}}}}{\underline{c}\seclen+\frac{\mathcal{L}_{\mu, \sigma}}{2}\bra{\frac{\seclen}{2}-t^*}^2} + \\
&  M\bra{ 
 \muupperbound-\mulowerbound, \frac{k_{\mu}}{\muupperbound-\mulowerbound}, \mathcal{L}_\mu, 1}\cdot M\bra{ \sigmaupperbound-\sigmalowerbound, \frac{k_{\sigma}}{\sigmaupperbound-\sigmalowerbound}, \mathcal{L}_{\sigma}, 1}\cdot
 \bra{\sigmaupperbound-\sigmalowerbound}\abs{t_0}\seclen,
        \end{align*}
where $t_0$ is defined in \cref{mech:std_continuous}, $
\underline{c} = \frac{k_{\mu}k_{\sigma}}{\bra{\muupperbound-\mulowerbound}\cdot\bra{\sigmaupperbound-\sigmalowerbound}}
$, function $M$ satisfies
$$
 M\bra{x,  c, \mathcal{L}, \mathcal{A}} = 
\begin{cases}
\frac{\mathcal{A}}{x} + \frac{\mathcal{L}x}{2}, & \text{if } c \leq \frac{\mathcal{A}}{x} - \frac{\mathcal{L}x}{2}\\
c + \sqrt{2\mathcal{L}\bra{\mathcal{A}-cx}}, & \text{if } c > \frac{\mathcal{A}}{x} - \frac{\mathcal{L}x}{2}
\end{cases},
$$
\small
$
\mathcal{L}_{\mu, \sigma} = \mathcal{L}_{\sigma} M\bra{ \frac{\muupperbound-\mulowerbound}{\abs{t_0}}, \frac{k_{\mu}}{\muupperbound-\mulowerbound}, \abs{t_0}\mathcal{L}_\mu, \frac{1}{\abs{t_0}}} + \abs{t_0} \mathcal{L}_{\mu} M\bra{ \sigmaupperbound-\sigmalowerbound, \frac{k_{\sigma}}{\sigmaupperbound-\sigmalowerbound}, \mathcal{L}_{\sigma}, 1}
$,
\normalsize
and
$t^* = \frac{\seclen}{2} + \frac{\underline{c}}{\mathcal{L}_{\mu, \sigma}} - {\privacythreshold} - \sqrt{\bra{\frac{\underline{c}}{\mathcal{L}_{\mu, \sigma}} - {\privacythreshold}}^2 + \frac{2\underline{c}\seclen}{\mathcal{L}_{\mu, \sigma}}}$.
        
        \item Exponential:
        \begin{align*}
            \distortionnotation&=\frac{1}{2}\seclen, \\\privacynotation&\leq \frac{{2\privacythreshold}\cdot\brb{\underline{c}+\mathcal{L}\bra{\seclen-x^*+{\privacythreshold}}}}{\underline{c}\seclen+\frac{\mathcal{L}}{2}\bra{\seclen-x^*}^2},
\end{align*}
where $x^* = \seclen + \frac{\underline{c}}{\mathcal{L}} + \privacythreshold - \sqrt{\bra{\frac{\underline{c}}{\mathcal{L}} + \privacythreshold}^2 + \frac{2\underline{c}\seclen}{\mathcal{L}}}$.
        \item Shifted exponential:
        \begin{align*}
            \distortionnotation %
            & = \frac{\seclen \ln 2}{2},\\
            \privacynotation %
            & < \frac{{2\privacythreshold}
\cdot \brb{\underline{c}+\mathcal{L}_{\lambda, h}\bra{\frac{\seclen}{2}-t^*-{\privacythreshold}}}}{\underline{c}\seclen+\frac{\mathcal{L}_{\lambda, h}}{2}\bra{\frac{\seclen}{2}-t^*}^2} + \\
&\ln 2\cdot M\bra{ 
 \hupperbound-\hlowerbound, \frac{k_{h}}{\hupperbound-\hlowerbound}, \mathcal{L}_h, 1}\cdot M\bra{ \lambdaupperbound-\lambdalowerbound, \frac{k_{\lambda}}{\lambdaupperbound-\lambdalowerbound}, \mathcal{L}_{\lambda}, 1}\cdot
\bra{\lambdaupperbound-\lambdalowerbound}\seclen,
        \end{align*}
where $\underline{c} = \frac{k_{h}k_{\lambda}}{\bra{\hupperbound-\hlowerbound}\cdot\bra{\lambdaupperbound-\lambdalowerbound}}$, function $M$ satisfies
$$
 M\bra{x,  c, \mathcal{L}, \mathcal{A}} = 
\begin{cases}
\frac{\mathcal{A}}{x} + \frac{\mathcal{L}x}{2}, & \text{if } c \leq \frac{\mathcal{A}}{x} - \frac{\mathcal{L}x}{2}\\
c + \sqrt{2\mathcal{L}\bra{\mathcal{A}-cx}}, & \text{if } c > \frac{\mathcal{A}}{x} - \frac{\mathcal{L}x}{2}
\end{cases},
$$ 
\small
$
\mathcal{L}_{\lambda, h} = \mathcal{L}_{\lambda} M\bra{ \frac{\hupperbound-\hlowerbound}{\ln 2}, \frac{k_{h}}{\hupperbound-\hlowerbound}, \ln 2 \cdot\mathcal{L}_h, \frac{1}{\ln 2}} + \ln 2 \cdot \mathcal{L}_{h} M\bra{ \lambdaupperbound-\lambdalowerbound, \frac{k_{\lambda}}{\lambdaupperbound-\lambdalowerbound}, \mathcal{L}_{\lambda}, 1}
$,
\normalsize
and
$t^* = \frac{\seclen}{2} + \frac{\underline{c}}{\mathcal{L}_{\lambda, h}} - \privacythreshold - \sqrt{\bra{\frac{\underline{c}}{\mathcal{L}_{\lambda, h}} - \privacythreshold}^2 + \frac{2\underline{c}\seclen}{\mathcal{L}_{\lambda, h}}}$.
    \end{packeditemize}
\end{proposition}

The proofs are the same as \cref{sec:proof_trade-off_mechanism_quantile_continuous_relaxed_Gaussian}, \cref{sec:upperbound_continuous_quantile_exponential_relaxed} and \cref{sec:upperbound_continuous_quantile_shifted_exponential_relaxed}, except that we use the $\auxdistance{\cdot}{\cdot}$, and $\auxrange{\cdot}{\cdot}$ from \cref{sec:proof_lowerbound_continuous_std_GAUSSIAN}, \cref{sec:proof_lowerbound_continuous_std_EXP}, and \cref{sec:proof_lowerbound_continuous_std_SHIFTED_EXP}.

\subsection{Proof of %
\cref{thm:discrete_std_more}}
\label{sec:proof_lowerbound_discrete_std_more}

\subsubsection{Geometric Distribution}

\begin{proof}
Let $\rvprivatewithparam{\rvparamnotation_1}$ and $\rvprivatewithparam{\rvparamnotation_2}$  be two Geometric random variables with parameters $\rvparamnotation_1$ and $\rvparamnotation_2$ respectively.
We assume that $\rvparamnotation_1 > \rvparamnotation_2$ without loss of generality. Let $k'$ satisfy $\bra{1-\rvparamnotation_1}^{k'} \rvparamnotation_1 = \bra{1-\rvparamnotation_2}^{k'} \rvparamnotation_2$ and $k_0 = \floor{k'}  + 1$. Then we can get that
\begin{align*}
\auxdistance{\rvprivatewithparam{\rvparamnotation_1}}{\rvprivatewithparam{\rvparamnotation_2}}
&= 
\distanceformulaTV{\rvprivatewithparam{\rvparamnotation_1}}{\rvprivatewithparam{\rvparamnotation_2}}\\
&= \frac{1}{2}\bra{1-\rvparamnotation_2}^{k_0} - \frac{1}{2}\bra{1-\rvparamnotation_1}^{k_0},
\\
\auxrange{\rvprivatewithparam{\rvparamnotation_1}}{\rvprivatewithparam{\rvparamnotation_2}}
&= \frac{\sqrt{1-\rvparamnotation_2}}{\rvparamnotation_2} - \frac{\sqrt{1-\rvparamnotation_1}}{\rvparamnotation_1}.
\end{align*}

Therefore, we can get that
\begin{align*}
\ratio= 
\inf_{\rvparamlowerbound< \rvparamnotation_1<\rvparamnotation_2\leq \rvparamupperbound}\frac{\bra{1-\rvparamnotation_2}^{k_0} - \bra{1-\rvparamnotation_1}^{k_0}}{2\bra{\frac{\sqrt{1-\rvparamnotation_2}}{\rvparamnotation_2} - \frac{\sqrt{1-\rvparamnotation_1}}{\rvparamnotation_1}}}.
\end{align*}
\end{proof}

\subsubsection{Binomial Distribution}
\begin{proof}
Let $\rvprivatewithparam{\rvparamnotation_1}$ and $\rvprivatewithparam{\rvparamnotation_2}$  be two binomial random variables with parameters $\rvparamnotation_1$ and $\rvparamnotation_2$ respectively with fixed number of trials $n$.
We assume that $\rvparamnotation_1 > \rvparamnotation_2$ without loss of generality.
Let $k'$ satisfy $\binom{n}{k'}\rvparamnotation_1^{k'}\bra{1-\rvparamnotation_1}^{n-k'} = \binom{n}{k'}\rvparamnotation_2^{k'}\bra{1-\rvparamnotation_1}^{n-k'}$ and $k_0=\lfloor k' \rfloor$. We can get that 
\begin{align*}
\auxdistance{\rvprivatewithparam{\rvparamnotation_1}}{\rvprivatewithparam{\rvparamnotation_2}}&= 
\distanceformulaTV{\rvprivatewithparam{\rvparamnotation_1}}{\rvprivatewithparam{\rvparamnotation_2}}\\
&= \frac{1}{2}I_{1-\rvparamnotation_2}\bra{n-k_0, 1+k_0}-\frac{1}{2}I_{1-\rvparamnotation_1}\bra{n-k_0, 1+k_0},
\\
\auxrange{\rvprivatewithparam{\rvparamnotation_1}}{\rvprivatewithparam{\rvparamnotation_2}}
&= \babs{\sqrt{n\rvparamnotation_2\bra{1-\rvparamnotation_2}}-\sqrt{n\rvparamnotation_1\bra{1-\rvparamnotation_1}}},
\end{align*}
where $I$ represents the regularized incomplete beta function.

Therefore, we can get that
\begin{align*}
\ratio= 
\inf_{\rvparamlowerbound< \rvparamnotation_1<\rvparamnotation_2\leq \rvparamupperbound}\frac{I_{1-\rvparamnotation_2}\bra{n-k_0, 1+k_0}-I_{1-\rvparamnotation_1}\bra{n-k_0, 1+k_0}}{\babs{\sqrt{n\rvparamnotation_2\bra{1-\rvparamnotation_2}}-\sqrt{n\rvparamnotation_1\bra{1-\rvparamnotation_1}}}}.
\end{align*}

\end{proof}

\subsubsection{Poisson Distribution}
\begin{proof}
Let $\rvprivatewithparam{\rvparamnotation_1}$ and $\rvprivatewithparam{\rvparamnotation_2}$  be two Poisson random variables with parameters $\rvparamnotation_1$ and $\rvparamnotation_2$ respectively. We assume that $\rvparamnotation_1>\rvparamnotation_2$ without loss of generality.
Let $k'$ satisfy $\rvparamnotation_1^{k'} e^{-\rvparamnotation_1} = \rvparamnotation_2^{k'} e^{-\rvparamnotation_2}$ and $k_0 = \lfloor k' \rfloor + 1$. Then we can get that
\begin{align*}
\auxdistance{\rvprivatewithparam{\rvparamnotation_1}}{\rvprivatewithparam{\rvparamnotation_2}}&= 
\distanceformulaTV{\rvprivatewithparam{\rvparamnotation_1}}{\rvprivatewithparam{\rvparamnotation_2}}\\
&= \frac{1}{2}Q\bra{k_0, \rvparamnotation_2} - \frac{1}{2}Q\bra{k_0, \rvparamnotation_1},
\\
\auxrange{\rvprivatewithparam{\rvparamnotation_1}}{\rvprivatewithparam{\rvparamnotation_2}}
&= \sqrt{\rvparamnotation_1} - \sqrt{\rvparamnotation_2},
\end{align*}
where $Q$ is the regularized gamma function.

Therefore, we can get that
\begin{align*}
\ratio= 
\inf_{\rvparamlowerbound< \rvparamnotation_1<\rvparamnotation_2\leq \rvparamupperbound}\frac{Q\bra{k_0, \rvparamnotation_2} - Q\bra{k_0, \rvparamnotation_1}}{2\bra{\sqrt{\rvparamnotation_1}-\sqrt{\rvparamnotation_2}}}.
\end{align*}

\end{proof}
\section{Case Study with \Secret{} = Fraction}
\label{sec:case_study_fraction}

The fraction of discrete distributions can reveal sensitive information.
In this section, %
we first present the results for ordinal distributions, where there is a specific formula for the fractions at each bin (i.e., binomial, Poisson, geometric that we discussed in \cref{sec:case_study_mean_discrete,sec:case_study_std_discrete}). We then present the results for categorical distributions, %
where there is no constraint on the fractions of the bins so long as they are normalized.

\subsection{Ordinal Distribution}
\label{sec:case_study_frac_discrete_ordinal}

Here, we consider the same three discrete distributions studied in \cref{sec:case_study_mean_discrete,sec:case_study_std_discrete}: geometric distributions,  binomial distributions, and Poisson distributions. We first analyze the lower bound. We assume that the secrete is the fraction of the $\fraction$-th bin.

\begin{corollary}[Privacy lower bound, secret = fraction of an ordinal distribution]
\label{thm:discrete_fraction}
Consider the secret function $\secretof{\rvparamnotation}=\privatepdf\bra{\fraction{}}$. For any  $\privacymetricthreshold\in\bra{0,1}$, when $\privacynotation\leq \privacymetricthreshold$, we have $\distortionnotation> \bra{\ceil{\frac{1}{\privacymetricthreshold}}-1}\cdot 2\ratio\privacythreshold$, where the value of $\ratio$ depends on the type of the distributions:
\begin{packeditemize}
    \item Geometric: 
    \begin{align*}
    \ratio= 
\inf_{\rvparamlowerbound< \rvparamnotation_1<\rvparamnotation_2\leq \rvparamupperbound}\frac{\bra{1-\rvparamnotation_2}^{h\bra{\rvparamnotation_1,\rvparamnotation_2}} - \bra{1-\rvparamnotation_1}^{h\bra{\rvparamnotation_1,\rvparamnotation_2}}}{2\babs{\bra{1-\rvparamnotation_2}^{\fraction{}}\rvparamnotation_2-\bra{1-\rvparamnotation_1}^{\fraction{}}\rvparamnotation_1}}~~,
\end{align*}
where $h\bra{\rvparamnotation_1,\rvparamnotation_2} = \floor{  \frac{\log\bra{\rvparamnotation_2} - \log\bra{\rvparamnotation_1}}{ \log\bra{1 - \rvparamnotation_1} - \log\bra{1 - \rvparamnotation_2 }  }  } + 1$. 
    \item Binomial: 
    \begin{align*}
    &\ratio=\inf_{\rvparamlowerbound< \rvparamnotation_1<\rvparamnotation_2\leq \rvparamupperbound} \\
    &~~\scalebox{1}{$
    \frac{I_{1-\rvparamnotation_2}\bra{n-h\bra{\rvparamnotation_1,\rvparamnotation_2}, 1+h\bra{\rvparamnotation_1,\rvparamnotation_2}}-I_{1-\rvparamnotation_1}\bra{n-h\bra{\rvparamnotation_1,\rvparamnotation_2}, 1+h\bra{\rvparamnotation_1,\rvparamnotation_2}}}{2\babs{\binom{n}{\fraction{}}\rvparamnotation_2^{\fraction{}}\bra{1-\rvparamnotation_2}^{n-\fraction{}}-\binom{n}{\fraction{}}\rvparamnotation_1^{\fraction{}}\bra{1-\rvparamnotation_1}^{n-\fraction{}}}},
    $}
    \end{align*}
    where $h\bra{\rvparamnotation_1,\rvparamnotation_2}=\floor{k'}$, $k' = n\ln\bra{\frac{1-\theta_2}{1-\theta_1}} \Big/ \ln\bra{{\frac{\theta_1\bra{1-\theta_2}}{\theta_2\bra{1-\theta_1}}}}$, %
    and $I$ represents the regularized incomplete beta function.
    \item Poisson: 
    \begin{align*}
    \ratio = 
    \inf_{\rvparamlowerbound< \rvparamnotation_1<\rvparamnotation_2\leq \rvparamupperbound}\frac{Q\bra{h\bra{\rvparamnotation_1,\rvparamnotation_2}, \rvparamnotation_2} - Q\bra{h\bra{\rvparamnotation_1,\rvparamnotation_2}, \rvparamnotation_1}}{2\babs{\frac{\rvparamnotation_1^\fraction{} e^{-\rvparamnotation_1}}{\fraction{}!}-\frac{\rvparamnotation_2^\fraction{} e^{-\rvparamnotation_2}}{\fraction{}!}}},
\end{align*}
where $h\bra{\rvparamnotation_1,\rvparamnotation_2} = \floor{\frac{\rvparamnotation_1 - \rvparamnotation_2}{\ln\bra{\rvparamnotation_1}-\ln\bra{\rvparamnotation_2}}} + 1$ and $Q$ is the regularized gamma function.
\end{packeditemize}
\end{corollary}
The proof is in \cref{sec:proof_lowerbound_discrete_fraction}. The above lower bounds can be computed numerically.

Since these distributions only have one parameter, we can use \cref{alg:dp} and \cref{alg:greedy} to derive a \datamechanism{}. 
The performance of greedy-based and dynamic-programming-based \datamechanisms{} for each distribution is shown in \cref{fig:algorithm_fraction}.

\begin{figure*}[th]
    \centering
    \begin{subfigure}{0.3\textwidth}
         \centering
        \includegraphics[width=1\linewidth]{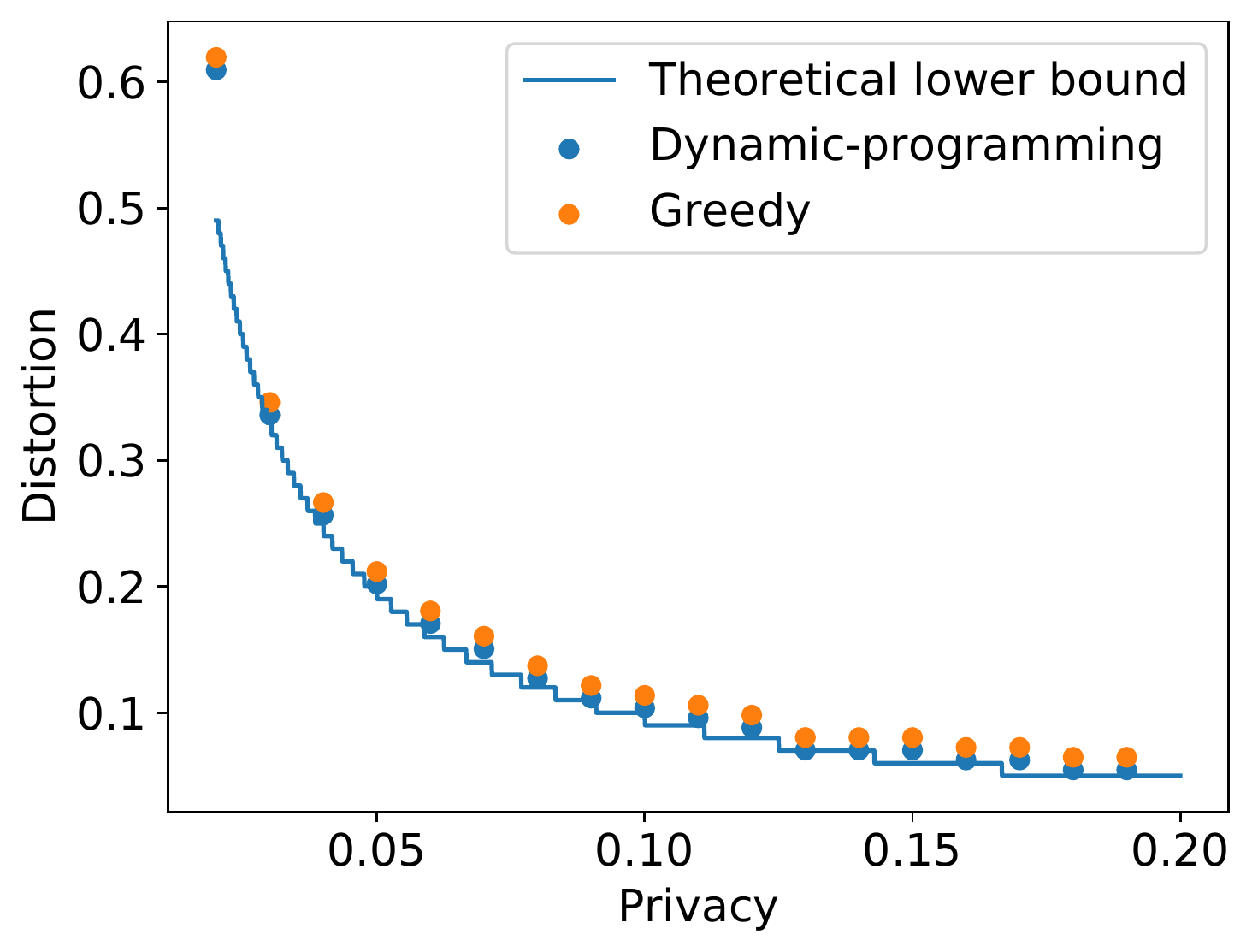}
         \caption{Distribution = Geometric}
         \label{fig:Geometric_fraction}
     \end{subfigure}
     \hfill
    \begin{subfigure}{0.3\textwidth}
         \centering
        \includegraphics[width=1\linewidth]{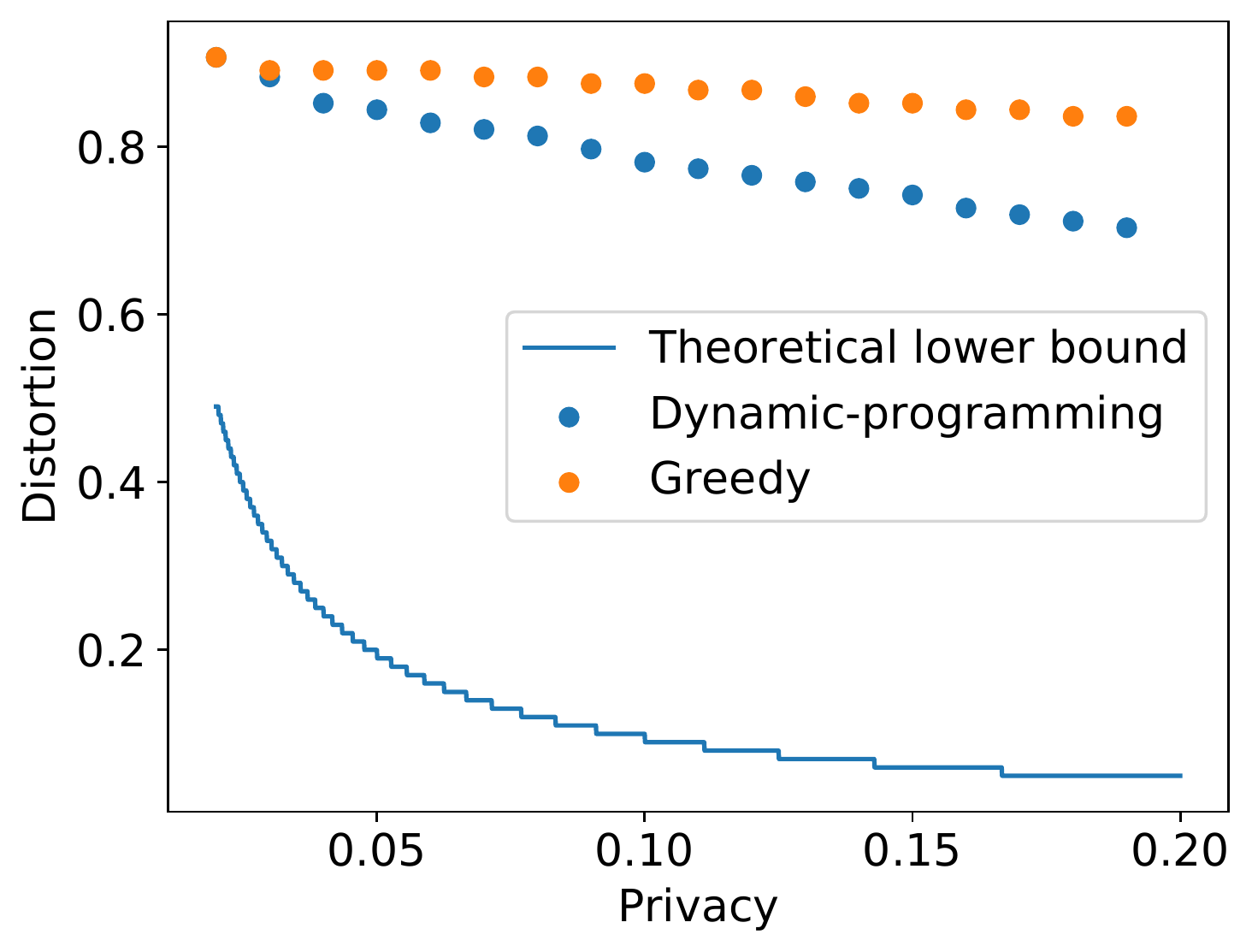}
         \caption{Distribution = Binomial}
         \label{fig:Binomial_fraction}
     \end{subfigure}
     \hfill
    \begin{subfigure}{0.3\textwidth}
         \centering
        \includegraphics[width=1\linewidth]{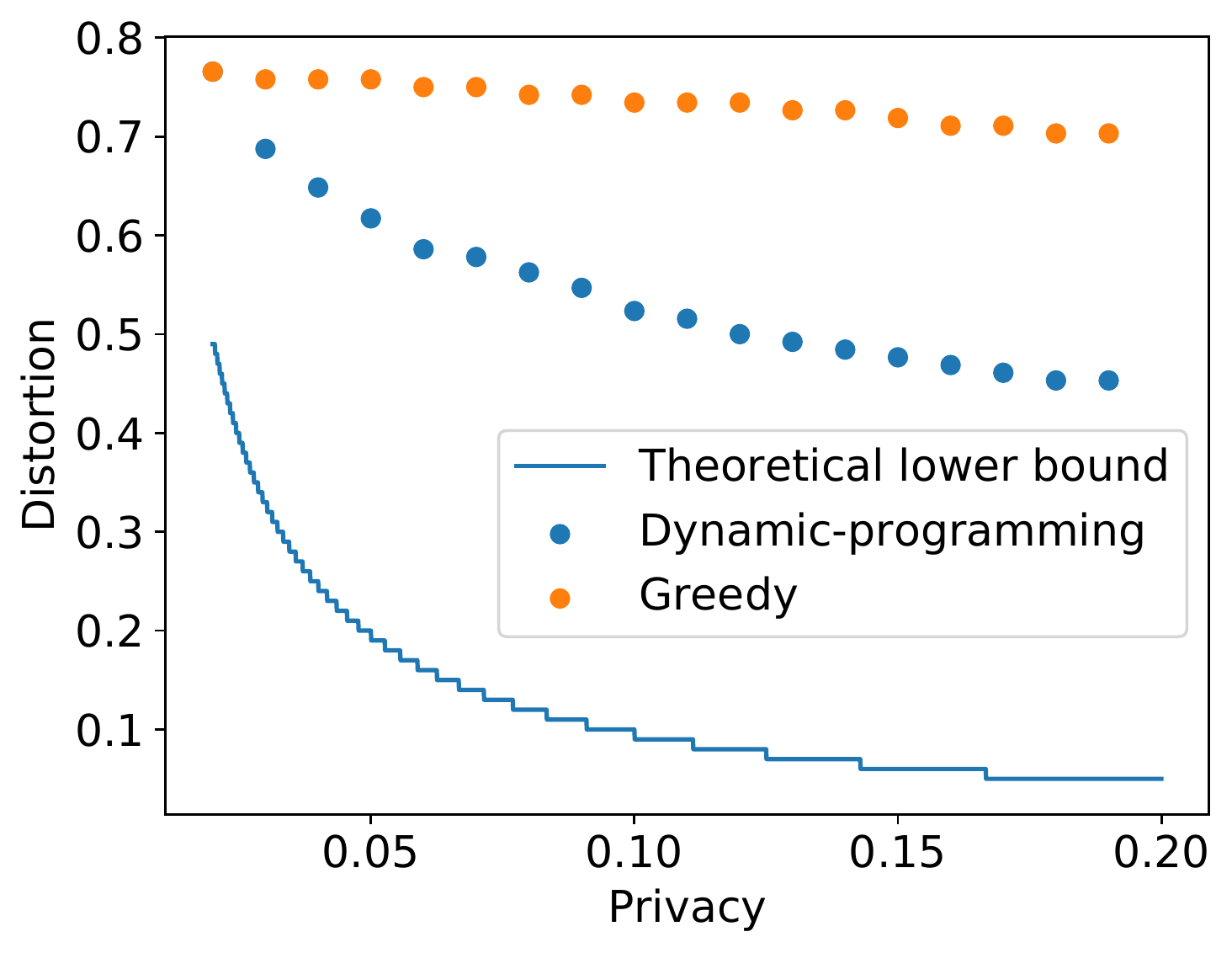}
         \caption{Distribution = Poisson}
         \label{fig:Poisson_fraction}
     \end{subfigure}
    \caption{ \Privacy{}-\distortion{} performance of \cref{alg:dp} and \cref{alg:greedy} for geometric, binomial and Poisson distribution when secret = fraction.
    }
    \label{fig:algorithm_fraction}
\end{figure*}

\subsection{Categorical Distribution}
\label{sec:case_study_fraction_discrete_general}
In this section, we consider categorical distributions where the fraction of each bin can be changed freely (as long as they are normalized). 
We assume that $\rvparamnotation=\bra{p_1,p_2,\ldots,p_\numcat}$ s.t. $p_i\in \brb{0,1}~\forall i\in[\numcat]$ and $\sum_i p_i=1$. 
Note that this is completely different from the distributions discussed in \cref{sec:case_study_frac_discrete_ordinal} where the parameter of the distribution is one-dimensional.

We first analyze the lower bound. Without loss of generality, we assume that we want to protect the fraction of the $\fraction{}$-th bin, i.e. $p_\fraction$.

\begin{corollary}[Privacy lower bound, secret = fraction of a general discrete distribution]
\label{thm:discrete_general_fraction}
Consider the secret function $\secretof{\rvparamnotation}=p_1$. For any  $\privacymetricthreshold\in\bra{0,1}$, when $\privacynotation\leq \privacymetricthreshold$, we have $\distortionnotation> \bra{\ceil{\frac{1}{\privacymetricthreshold}}-1}\cdot \privacythreshold$.
\end{corollary}

The proof is in \cref{sec:proof_lowerbound_discrete_general_fraction}. Next, we present the \datamechanism{} under the following assumption.
\begin{assumption}
    \label{assu:fraction_discrete_general}
    
    The prior distribution of $\bra{p_1,\ldots,p_\numcat}$ is a uniform distribution over all the probability simplex
    $\{(p_1,\ldots,p_\numcat)|p_i\in \brba{0,1}$$ \ \forall i\in[\numcat] \text{ and } \sum_i p_i=1 \}$.
\end{assumption}

\begin{mechanism}[For secret = fraction of a categorical distribution]
    \label{mech:fraction_discrete_general}
     The parameters of the mechanism are as follows.
     \begin{align*}
        \subsetofprivateparamof{p_1,...,p_\numcat} &= \bigg\{\bigg(
        p_1-\frac{t}{\numcat-1}, \ldots,p_{\fraction-1}-\frac{t}{\numcat-1},
        p_{\fraction}+t,\\
        &
        p_{\fraction+1}-\frac{t}{\numcat-1},\ldots,p_{\numcat}-\frac{t}{\numcat-1}
        \bigg)
        \bigg| t\in \brba{-\frac{\seclen}{2}, \frac{\seclen}{2}}\bigg\}
        ~~,\\
        \releaseparamofindex{p_1,\ldots,p_\numcat} &= \bigg( p_1-T,\ldots,p_{\fraction-1}-T,p_{\fraction}+\bra{\numcat-1}T,\\
        &~~p_{\fraction+1}-T,\ldots,p_{\numcat+1}-T \bigg) ~~,
    \end{align*}
    where $T=\min\brc{p_1,\ldots, p_{\fraction-1},p_{\fraction+1},\ldots,p_\numcat, 0}$, and
    \begin{align*}
        \privateparamindexsetnotation &=  \bigg\{ \bra{p_1,\ldots,p_{\numcat}} \bigg|
        \forall i ~p_i\in\brab{-\frac{\seclen} {2\bra{\numcat-1}},1}, 
        \sum_i p_i=1, \\
        &p_{\fraction{}}=\bra{k+0.5}\seclen, \text{where } 
        k\in\brc{0,1,\ldots,\numcat-1}
        \bigg\}.
    \end{align*}
    Here $\seclen>0$ is a hyper-parameter of the mechanism that divides 1.
\end{mechanism}

This \datamechanism{} achieves the following \privacy{}-\distortion{} trade-off. %
\begin{proposition}
    \label{thm:upperbound_discrete_general_fraction}
    Under \cref{assu:fraction_discrete_general},
    \cref{mech:fraction_discrete_general} %
    has the following $\privacynotation{}$ and $\distortionnotation$  value/bound.
        \begin{align*}
        \privacynotation 
        & <\frac{2\privacythreshold}{\seclen}+ 1 - \bra{1-\frac{\seclen}{\numcat-1}}^{\numcat-1},\\
        \distortionnotation 
        & = \frac{\seclen}{2} 
        <
        \bra{2+\frac{\seclen}{\privacythreshold}}\distortionnotation_{\text{opt}}.
        \end{align*}
    Under the regime $\sup \bra{\seclen}\rightarrow \mathcal{A}\privacythreshold$, where 
    $\mathcal{A}$ is a constant larger than $2$, $\distortionnotation$ satisfies
\begin{align*}
    \lim_{\sup \bra{\seclen}\rightarrow \mathcal{A}\privacythreshold} \distortionnotation < (2+\mathcal{A})  \distortionnotation_{\text{opt}}.
\end{align*}
$\distortionnotation_{\text{opt}}$ is the minimal \distortion{} an optimal \datamechanism{} can achieve given the \privacy{} \cref{mech:fraction_discrete_general} achieves.
\end{proposition}

The proof is in \cref{sec:proof_upperbound_discrete_general_fraction}. To ensure that $\privacynotation{}<1$, $\seclen$ should satisfy $\seclen> 2\privacythreshold$. According to \cref{{thm:upperbound_discrete_general_fraction}}, the mechanism is order-optimal with multiplicative factor $2+\mathcal{A}$ when $\sup \bra{\seclen}\rightarrow \mathcal{A}\privacythreshold$, where $\mathcal{A}>2$. %

\subsection{Proof of %
\cref{thm:discrete_fraction}}
\label{sec:proof_lowerbound_discrete_fraction}

\subsubsection{Geometric Distribution}

\begin{proof}
Let $\rvprivatewithparam{\rvparamnotation_1}$ and $\rvprivatewithparam{\rvparamnotation_2}$  be two Geometric random variables with parameters $\rvparamnotation_1$ and $\rvparamnotation_2$ respectively.
We assume that $\rvparamnotation_1 > \rvparamnotation_2$ without loss of generality. Let $k'$ satisfy $\bra{1-\rvparamnotation_1}^{k'} \rvparamnotation_1 = \bra{1-\rvparamnotation_2}^{k'} \rvparamnotation_2$ and $k_0 = \floor{k'}  + 1$. Then we can get that
\begin{align*}
\auxdistance{\rvprivatewithparam{\rvparamnotation_1}}{\rvprivatewithparam{\rvparamnotation_2}}
&= \distanceformulaTV{\rvprivatewithparam{\rvparamnotation_1}}{\rvprivatewithparam{\rvparamnotation_2}}\\
&= \frac{1}{2}\bra{1-\rvparamnotation_2}^{k_0} - \frac{1}{2}\bra{1-\rvparamnotation_1}^{k_0},
\\
\auxrange{\rvprivatewithparam{\rvparamnotation_1}}{\rvprivatewithparam{\rvparamnotation_2}}
&= \babs{\bra{1-\rvparamnotation_2}^{\fraction{}}\rvparamnotation_2-\bra{1-\rvparamnotation_1}^{\fraction{}}\rvparamnotation_1}.
\end{align*}

Therefore, we can get that
\begin{align*}
\ratio= 
\inf_{\rvparamlowerbound< \rvparamnotation_1<\rvparamnotation_2\leq \rvparamupperbound} \frac{\bra{1-\rvparamnotation_2}^{k_0} - \bra{1-\rvparamnotation_1}^{k_0}}{2\babs{\bra{1-\rvparamnotation_2}^{\fraction{}}\rvparamnotation_2-\bra{1-\rvparamnotation_1}^{\fraction{}}\rvparamnotation_1}}~~.
\end{align*}
\end{proof}

\subsubsection{Binomial Distribution}
\begin{proof}
Let $\rvprivatewithparam{\rvparamnotation_1}$ and $\rvprivatewithparam{\rvparamnotation_2}$  be two binomial random variables with parameters $\rvparamnotation_1$ and $\rvparamnotation_2$ respectively with fixed number of trials $n$.
We assume that $\rvparamnotation_1 > \rvparamnotation_2$ without loss of generality.
Let $k'$ satisfy $\binom{n}{k'}\rvparamnotation_1^{k'}\bra{1-\rvparamnotation_1}^{n-k'} = \binom{n}{k'}\rvparamnotation_2^{k'}\bra{1-\rvparamnotation_2}^{n-k'}$ and $k_0=\lfloor k' \rfloor$. We can get that
\begin{align*}
\auxdistance{\rvprivatewithparam{\rvparamnotation_1}}{\rvprivatewithparam{\rvparamnotation_2}}&= 
\distanceformulaTV{\rvprivatewithparam{\rvparamnotation_1}}{\rvprivatewithparam{\rvparamnotation_2}}\\
&= \frac{1}{2}I_{1-\rvparamnotation_2}\bra{n-k_0, 1+k_0}-\frac{1}{2}I_{1-\rvparamnotation_1}\bra{n-k_0, 1+k_0},
\\
\auxrange{\rvprivatewithparam{\rvparamnotation_1}}{\rvprivatewithparam{\rvparamnotation_2}}
&= n\bra{\rvparamnotation_1-\rvparamnotation_2},
\end{align*}
where $I$ represents the regularized incomplete beta function.

Therefore, we can get that
\begin{align*}
\ratio = 
\inf_{\rvparamlowerbound< \rvparamnotation_1<\rvparamnotation_2\leq \rvparamupperbound}\frac{I_{1-\rvparamnotation_2}\bra{n-k_0, 1+k_0}-I_{1-\rvparamnotation_1}\bra{n-k_0, 1+k_0}}{2\babs{\binom{n}{\fraction{}}\rvparamnotation_2^{\fraction{}}\bra{1-\rvparamnotation_2}^{n-\fraction{}}-\binom{n}{\fraction{}}\rvparamnotation_1^{\fraction{}}\bra{1-\rvparamnotation_1}^{n-\fraction{}}}}.
\end{align*}
\end{proof}

\subsubsection{Poisson Distribution}

\begin{proof}
Let $\rvprivatewithparam{\rvparamnotation_1}$ and $\rvprivatewithparam{\rvparamnotation_2}$  be two Poisson random variables with parameters $\rvparamnotation_1$ and $\rvparamnotation_2$ respectively. We assume that $\rvparamnotation_1>\rvparamnotation_2$ without loss of generality.
Let $k'$ satisfy $\rvparamnotation_1^{k'} e^{-\rvparamnotation_1} = \rvparamnotation_2^{k'} e^{-\rvparamnotation_2}$ and $k_0 = \lfloor k' \rfloor + 1$. Then we can get that
\begin{align*}
\auxdistance{\rvprivatewithparam{\rvparamnotation_1}}{\rvprivatewithparam{\rvparamnotation_2}}&= 
\distanceformulaTV{\rvprivatewithparam{\rvparamnotation_1}}{\rvprivatewithparam{\rvparamnotation_2}}\\
&= \frac{1}{2}Q\bra{k_0, \rvparamnotation_2} - \frac{1}{2}Q\bra{k_0, \rvparamnotation_1},
\\
\auxrange{\rvprivatewithparam{\rvparamnotation_1}}{\rvprivatewithparam{\rvparamnotation_2}}
&= \babs{\frac{\rvparamnotation_1^\fraction{} e^{-\rvparamnotation_1}}{\fraction{}!}-\frac{\rvparamnotation_2^\fraction{} e^{-\rvparamnotation_2}}{\fraction{}!}},
\end{align*}
where $Q$ is the regularized gamma function.

Therefore, we can get that
\begin{align*}
\ratio=
\inf_{\rvparamlowerbound< \rvparamnotation_1<\rvparamnotation_2\leq \rvparamupperbound}\frac{Q\bra{k_0, \rvparamnotation_2} - Q\bra{k_0, \rvparamnotation_1}}{2\babs{\frac{\rvparamnotation_1^\fraction{} e^{-\rvparamnotation_1}}{\fraction{}!}-\frac{\rvparamnotation_2^\fraction{} e^{-\rvparamnotation_2}}{\fraction{}!}}}.
\end{align*}

\end{proof}
\subsection{Proof of %
\cref{thm:discrete_general_fraction}}
\label{sec:proof_lowerbound_discrete_general_fraction}
\begin{proof}
Let $\rvprivatewithparam{p_1^1,p_2^1,\ldots,p_\numcat^1}$ and $\rvprivatewithparam{p_1^2,p_2^2,\ldots,p_\numcat^2}$ be two categorical random variables. We have
\begin{align*}
&\quad\auxdistance{\rvprivatewithparam{p_1^1,p_2^1,\ldots,p_\numcat^1}}{\rvprivatewithparam{p_1^2,p_2^2,\ldots,p_\numcat^2}}\\
&= 
\distanceformulaTV{\rvprivatewithparam{p_1^1,p_2^1,\ldots,p_\numcat^1}}{\rvprivatewithparam{p_1^2,p_2^2,\ldots,p_\numcat^2}}\\
&\geq \frac{1}{2}\babs{p_{\fraction{}}^{1}-p_{\fraction{}}^{2}},\numberthis\label{eq:lowerbound_discrete_general_d}
\\
&\quad\auxrange{\rvprivatewithparam{p_1^1,p_2^1,\ldots,p_\numcat^1}}{\rvprivatewithparam{p_1^2,p_2^2,\ldots,p_\numcat^2}}\\
&= \babs{p_{\fraction{}}^{1}-p_{\fraction{}}^{2}}.
\end{align*}

Therefore, we can get that
\begin{align*}
\ratio \geq 
\frac{1}{2}.
\end{align*}
\end{proof}
\subsection{Proof of %
\cref{thm:upperbound_discrete_general_fraction}}
\label{sec:proof_upperbound_discrete_general_fraction}

\begin{proof}

We first focus on the proof for $\privacynotation$.

We separate the space of possible data parameters into two regions: 
\small
$S_1=\Big\{(p_1,\ldots,p_\numcat)|p_i\in \brb{\frac{\seclen}{2\bra{\numcat-1}},1-\frac{\seclen}{2\bra{\numcat-1}}}$ $\forall i\in[\numcat] \text{ and } \sum_i p_i=1 \Big\}$
\normalsize
and $S_2= \brc{(p_1,\ldots,p_\numcat)|p_i\in \brba{0,1}~\forall i\in[\numcat] \text{ and } \sum_i p_i=1 } \setminus S_1$.
The high-level idea of our proof is as follows.
Note that for any parameter $\rvparamnotation\in S_{1}$, there exists a $\subsetofprivateparamof{p_1,\ldots,p_\numcat}$ s.t. $\rvparamnotation \in \subsetofprivateparamof{p_1,\ldots,p_\numcat}$ and $\subsetofprivateparamof{p_1,\ldots,p_\numcat}\subset S_{1}$. Therefore, we can bound the attack success rate if  $\rvparamnotation\in S_{1}$. At the same time, the probability of $\rvparamnotation\in S_{2}$ is bounded.
Therefore, we can bound the overall attacker's success rate (i.e., $\privacynotation$). More specifically, let the optimal attacker be $\secretestimatestarnotation$. We have
\begin{align*}
    \privacynotation{} 
    &= \probof{ \secretestimatestarof{\releaservparamnotation}\in\brb{ \secretofparam - \privacythreshold, \secretofparam + \privacythreshold } }\\
    &= \int_{ \rvparamnotation\in S_{1}}p(\rvparamnotation)\probof{ \secretestimatestarof{\releaservparamnotation}\in\brb{ \secretofparam - \privacythreshold, \secretofparam + \privacythreshold } }d\rvparamnotation \\
    &\quad+  \int_{ \rvparamnotation\in S_{2}}p(\rvparamnotation)\probof{ \secretestimatestarof{\releaservparamnotation}\in\brb{ \secretofparam - \privacythreshold, \secretofparam + \privacythreshold }} d\rvparamnotation\\
    &<\frac{2\privacythreshold}{\seclen} + \bra{1 - \bra{1-\frac{\seclen}{\numcat-1}}^{\numcat-1}}.
\end{align*}

For the \distortion{},
it is straightforward to get that $\distortionnotation = \frac{\seclen}{2}$ from \cref{eq:lowerbound_discrete_general_d}, and $\distortionnotation_{\text{opt}}>\bra{\ceil{\frac{1}{\privacynotation{}}}-1}\cdot \privacythreshold \geq \privacythreshold$ from \cref{thm:lowerbound_continuous_quantile}.
We can get that $\bra{\privacynotation - \bra{1 - \bra{1-\frac{\seclen}{\numcat-1}}^{\numcat-1}}} \cdot \distortionnotation = \privacythreshold$ and 
\begin{align*}
\distortionnotation & = \distortionnotation_{\text{opt}} + \distortionnotation - \distortionnotation_{\text{opt}}\\
& < \distortionnotation_{\text{opt}} + \distortionnotation - \bra{\ceil{\frac{1}{\privacynotation}}-1}\cdot \privacythreshold\\
&\leq \distortionnotation_{\text{opt}} + \privacythreshold + \distortionnotation - {{\frac{\privacythreshold}{\privacynotation}}}\\
&= \distortionnotation_{\text{opt}} + \privacythreshold + \frac{\bra{1 - \bra{1-\frac{\seclen}{\numcat-1}}^{\numcat-1}}}{\frac{2\privacythreshold}{\seclen} + \bra{1 - \bra{1-\frac{\seclen}{\numcat-1}}^{\numcat-1}}}\cdot\distortionnotation\\
&=\bra{1+\frac{\seclen}{2\privacythreshold}\bra{1 - \bra{1-\frac{\seclen}{\numcat-1}}^{\numcat-1}}} \bra{\distortionnotation_{\text{opt}} + 2\ratio\privacythreshold } \\
&\leq \bra{2+\frac{\seclen}{\privacythreshold}\bra{1 - \bra{1-\frac{\seclen}{\numcat-1}}^{\numcat-1}}} \distortionnotation_{\text{opt}}\\
&< \bra{2+\frac{\seclen}{\privacythreshold}}\distortionnotation_{\text{opt}}.
\end{align*}

\end{proof}

\section{Additional Results}
\label{app:additional_results}

\revision{In this section, we provide additional results on how released data from our mechanisms can support downstream applications.}

\revision{We consider the salaries from people with Master’s and PhD degrees in this Kaggle dataset \url{https://www.kaggle.com/datasets/rkiattisak/salaly-prediction-for-beginer}. We plot its histogram in \cref{fig:salary_hist_all}.
We can see that there are two peaks. They correspond to people with age<=40 and age>40 (see \cref{fig:salary_hist_sep}).}

\begin{figure}[t]
    \centering
    \includegraphics[width=0.5\linewidth]{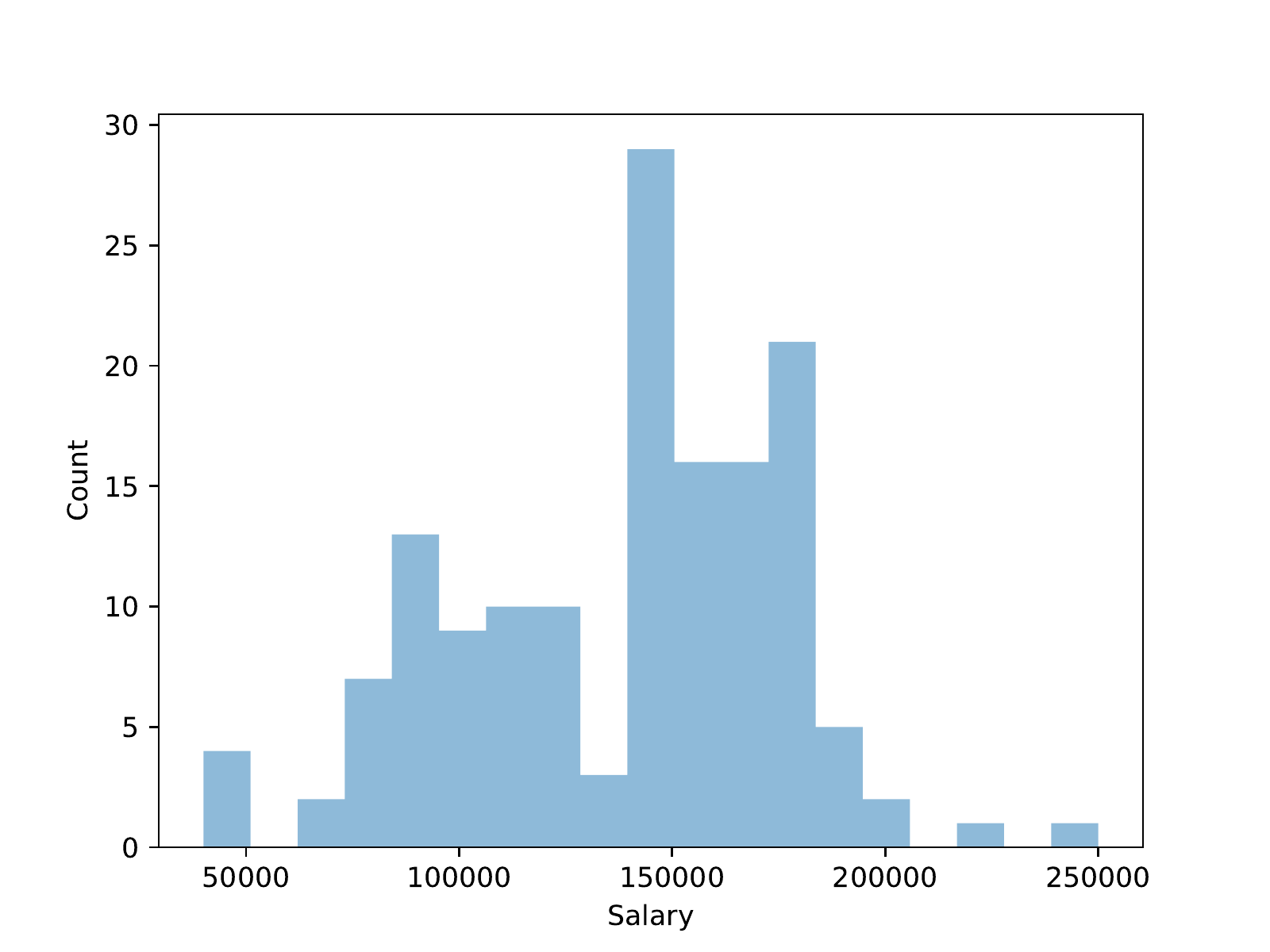}
    \caption{Histogram of salary dataset.}
    \label{fig:salary_hist_all}
\end{figure}

\begin{figure}[t]
    \centering
    \includegraphics[width=0.5\linewidth]{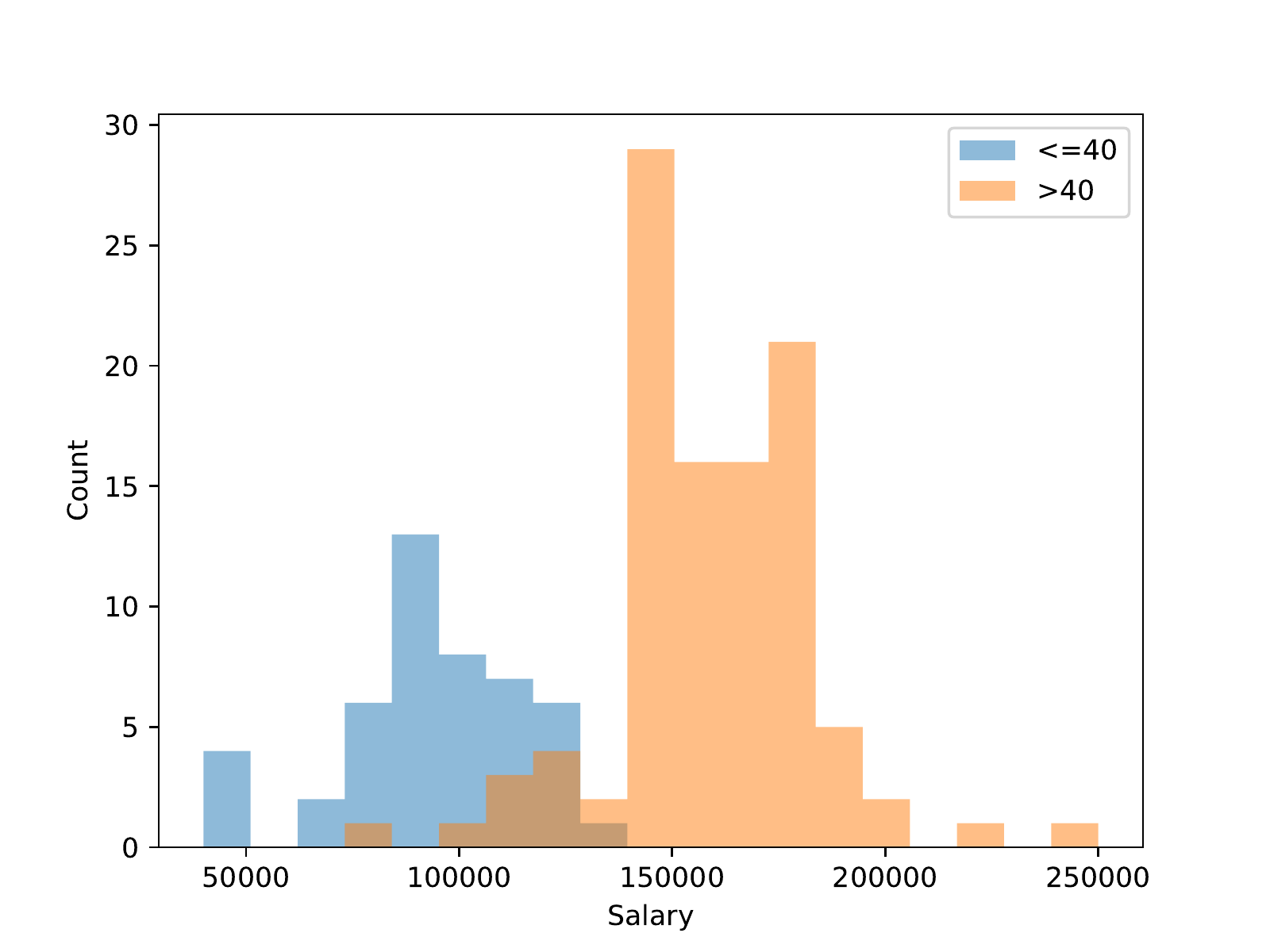}
    \caption{Histogram of salary dataset for people with age <= 40 and > 40.}
    \label{fig:salary_hist_sep}
\end{figure}

\begin{figure}[t]
    \centering
    \includegraphics[width=0.5\linewidth]{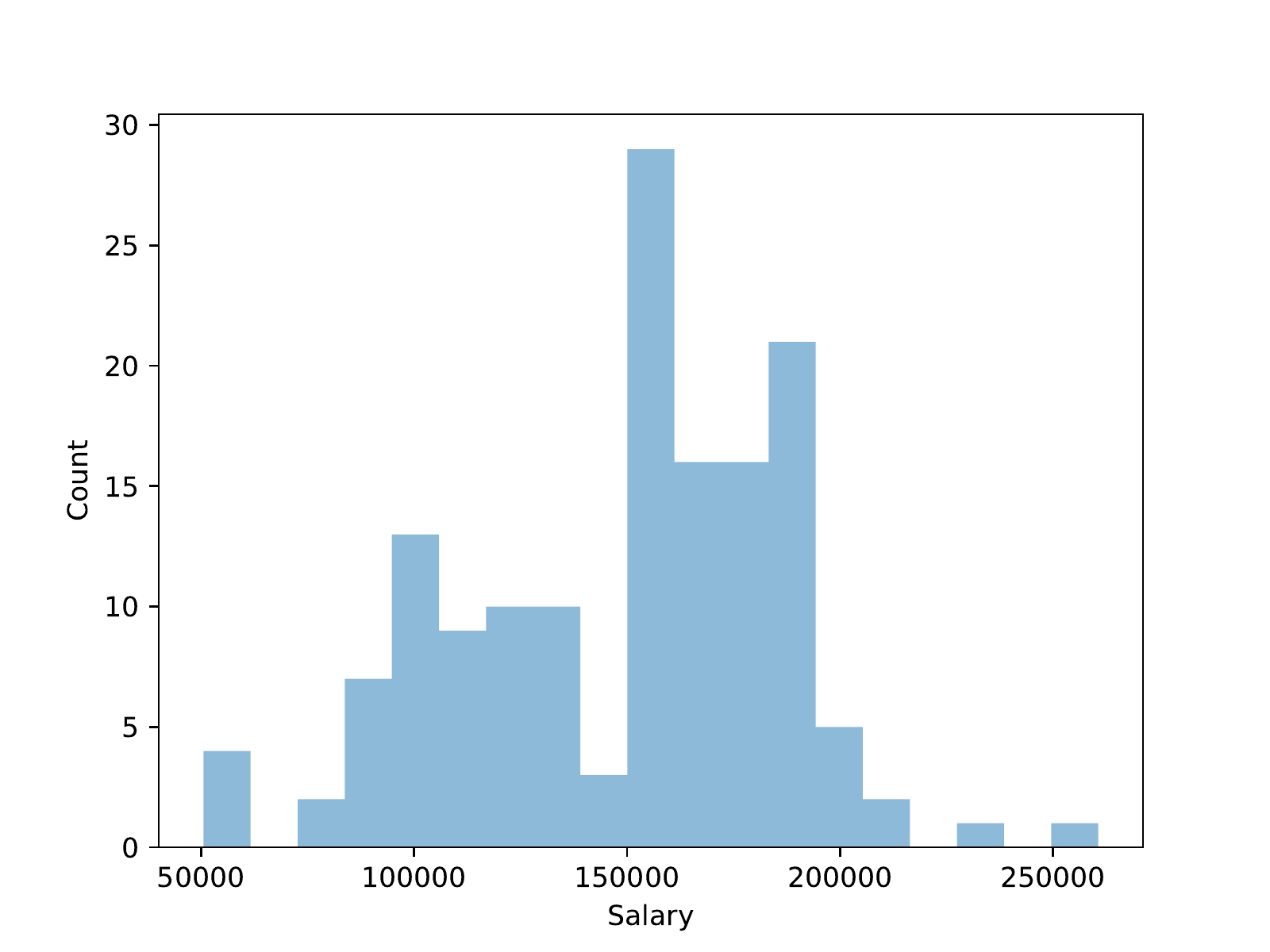}
    \caption{Histogram of salary dataset for applying the mechanism in \cref{sec:case_study_extend_dataset}.}
    \label{fig:salary_hist_all_after}
\end{figure}
 
\revision{Assume the goal is to release this dataset and preserve the salary difference between people with age<=40 and age>40, while protecting the mean salaries. We can apply our mechanism for mean (\cref{sec:case_study_extend_dataset}) on this dataset. The histogram of the released data is shown in \cref{fig:salary_hist_all_after}. Data receivers can obtain the salary difference between people with age<=40 and age>40 accurately by computing the difference between the two peaks, while the mean salaries are protected under our mechanism.} 

\revision{Here we use the salary difference between people with age<=40 and age>40 as an example. In general, any downstream tasks that depend only on the “shape” of the distribution will not be affected by our mechanism, since our mechanism shifts all samples by the same amount. }

\end{appendices}

\end{document}